\documentclass[11pt,a4paper,reqno]{amsart}
\usepackage[left=2.5cm,right=2.5cm,bottom=3cm,top=3cm,bindingoffset=0.5cm]{geometry}

\usepackage[utf8]{inputenc}
\usepackage{color}
\usepackage{amsmath}
\usepackage{mathtools}
\usepackage{amssymb}
\usepackage{amsthm}
\usepackage{fancyhdr}
\usepackage{csquotes}
\usepackage{bbm}
\usepackage{comment}
\usepackage{dsfont}
\usepackage{tikz}
\usepackage{cancel}

\usepackage{physics}

\usepackage[T1]{fontenc}
\usepackage{amsfonts}
\usepackage{vmargin}
\usepackage{setspace}
\usepackage{xcolor}
\usepackage{esint}

\definecolor{darkgreen}{rgb}{0,0.5,0}
\definecolor{darkred}{cmyk}{0,1,1,0.4}
\definecolor{vertfonce}{rgb}{0.20, 0.46, 0.25}
\definecolor{rougefonce}{rgb}{0.64, 0.09, 0.20}
\usepackage[breaklinks=true,
            colorlinks=true,
            linkcolor=rougefonce,
            citecolor=vertfonce]{hyperref}

\usepackage{mathrsfs, enumerate}

\theoremstyle{definition} \newtheorem{definition}{Definition}[section]
\theoremstyle{plain} \newtheorem{theorem}[definition]{Theorem}
\theoremstyle{plain} 
\theoremstyle{plain} \newtheorem{proposition}[definition]{Proposition}
\theoremstyle{plain} \newtheorem{lemma}[definition]{Lemma}
\theoremstyle{plain} 
\theoremstyle{definition} \newtheorem{remark}[definition]{Remark}
\theoremstyle{definition} 

\numberwithin{equation}{section}
\fancypagestyle{plain}{%
  \fancyhf{}%
}

\newcommand{\N}{\mathbb{N}}
\newcommand{\Z}{\mathbb{Z}}

\newcommand{\R}{\mathbb{R}}
\newcommand{\C}{\mathbb{C}}

\DeclareMathOperator{\supp}{\mathrm{supp}}

\newcommand{\bp}{\mathbf{p}}
\newcommand{\bx}{\mathbf{x}}
\newcommand{\br}{\mathbf{r}}
\newcommand{\by}{\mathbf{y}}
\newcommand{\bA}{\mathbf{A}}
\def\b{|}
\newcommand{\im}{\mathrm{i}}

\DeclarePairedDelimiterX\innerp[2]{\langle}{\rangle}{
	#1, #2
}

\usepackage[explicit]{titlesec}
\titleformat{\section}{\centering\Large\bfseries}{\thesection \ --}{0.7em}{\Large\bfseries #1}
\titleformat{\subsection}{\centering\large\bfseries}{\thesubsection \ --}{0.4em}{\large\bfseries #1}
\titleformat{\subsubsection}{\centering\bfseries}{\thesubsubsection \ --}{0.4em}{\bfseries #1}

\makeatletter

\@addtoreset{equation}{section}
\makeatother

\allowdisplaybreaks[1]

\usepackage{hyperref}

\usepackage{placeins}

\usepackage{etoolbox}

\usepackage{amsthm, amsmath, amsfonts, amssymb, appendix, dsfont, latexsym, mathrsfs, mathtools}
\usepackage{esint}
\usepackage{graphicx}
\usepackage{tikz}

\makeatletter
\usepackage{delarray, a4, color}
\usepackage{euscript}

\theoremstyle{plain}
\newtheorem*{theorem*}{Theorem}

\theoremstyle{definition}

\theoremstyle{remark}

\numberwithin{equation}{section}

\DeclareMathAlphabet{\mathpzc}{OT1}{pzc}{m}{it}
\def\eps{\varepsilon}

\def\C {\mathbb{C}}

\def\N {\mathbb{N}}

\def\R {\mathbb{R}}

\def\Z {\mathbb{Z}}

\newcommand\1{{\ensuremath {\mathds 1} }}

\newcommand{\keyword}[1]{\textbf{#1}}

\newcommand{\0}{\mathbf{0}}

\newcommand{\bJ}{\mathbf{J}}

\newcommand{\bX}{\mathbf{X}}

\newcommand{\bz}{\mathbf{z}}

\newcommand{\cE}{\mathcal{E}}

\newcommand{\cH}{\mathcal{H}}

\newcommand{\sA}{\textup{A}} 

\newcommand{\sSU}{\textup{SU}}

\newcommand{\sx}{\textup{x}}

\newcommand{\CLGN}{C_{\rm LGN}}

\newcommand{\NLL}{\mathrm{NLL}}
\newcommand{\PsiTrial}{\Psi^{\mathrm{trial}}_N}

\newcommand{\mat}[1]{
\begin{bmatrix}
#1
\end{bmatrix}
}

\DeclareMathOperator{\infspec}{\mathrm{inf\, spec\,}}

\DeclareMathOperator{\Curl}{\mathrm{curl}}
\DeclareMathOperator{\diag}{\mathrm{diag}}

\newcommand{\inp}[1]{\left\langle#1\right\rangle}

\newcommand{\slot}{\cdot}

\newcommand{\bAR}{\mathbf{A}^{\!R}}

\newcommand{\sym}{\mathrm{sym}}
\newcommand{\asym}{\mathrm{asym}}

\newcommand{\loc}{\mathrm{loc}}

\newcommand{\bDelta}{{\mbox{$\triangle$}\hspace{-8.0pt}\scalebox{0.8}{$\triangle$}}}

\def\XXint#1#2#3{{\setbox0=\hbox{$#1{#2#3}{%
\int}$ }
\vcenter{\hbox{$#2#3$ }}\kern-.6\wd0}}

\makeatletter
\patchcmd{\@setaddresses}{\indent}{\noindent}{}{}
\patchcmd{\@setaddresses}{\indent}{\noindent}{}{}
\patchcmd{\@setaddresses}{\indent}{\noindent}{}{}
\patchcmd{\@setaddresses}{\indent}{\noindent}{}{}
\makeatother

\title{Microscopic derivation of the stationary Chern--Simons--Schr\"odinger equation for almost-bosonic anyons}

\author[A. Ataei]{Alireza ATAEI}
\address{(A. Ataei) Department of Mathematics, Uppsala University, Box 480, SE-751 06, Uppsala, Sweden}
\email{\url{alireza.ataei@math.uu.se}}

\author[D. Lundholm]{Douglas LUNDHOLM}
\address{\noindent (D. Lundholm) Department of Mathematics, Uppsala University, Box 480, SE-751 06, Uppsala, Sweden}
\email{\url{douglas.lundholm@math.uu.se}}

\author[T. Girardot]{Th\'eotime GIRARDOT}
\address{(T. Girardot) Gran Sasso Science Institute GSSI, Viale Francesco Crispi, 7 Rectorate, Via Michele Iacobucci, 2, 67100 L'Aquila, Italy  }
\email{\url{theotime.girardot@gssi.it}}


\begin{document}

\maketitle
\begin{abstract}
In this work we consider the $N$-body Hamiltonian describing the microscopic structure of a quantum gas of almost-bosonic anyons. This description includes both extended magnetic flux and spin-orbit/soft-disk interaction between the particles which are confined in a scalar trapping potential. We study a physically well-motivated ansatz for a sequence of trial states, consisting of Jastrow repulsive short-range correlations and a condensate, with sufficient variational freedom to approximate the ground state (and possibly also low-energy excited states) of the gas. In the limit $N \to \infty$, while taking the relative size of the anyons to zero and the total magnetic flux $2\pi\beta$ to remain finite, we rigorously derive the stationary Chern--Simons--Schr\"odinger/average-field--Pauli effective energy density functional for the condensate wave function. This includes a scalar self-interaction parameter $\gamma$ which depends both on $\beta$, the diluteness of the gas, and the spin-orbit coupling strength $g$, but becomes independent of these microscopic details for a particular value of the coupling $g=2$ in which supersymmetry is exhibited (on all scales, both microscopic and mesoscopic) with $\gamma=2\pi|\beta|$.
Our findings confirm and clarify the predictions we have found in the physics literature.
\end{abstract}

\setcounter{tocdepth}{3}
\tableofcontents


\section[\qquad Introduction and main results]{Introduction and main results}

\subsection{Introduction}

In this work, we consider the mathematical description of a quantum gas consisting of a large number $N$ of 2D anyons, i.e.\ effective identical quasiparticles in two spatial dimensions with \emph{exchange} statistics properties that are intermediate to those of bosons and fermions. We model anyons in their magnetic gauge representation, as bosons with a magnetic flux of intensity $2\pi\alpha$ attached to each particle, so that exchanges, considered as loops in the configuration space, generate the appropriate phase, a multiple of $e^{i\pi\alpha}$ (here we assume that exchanges commute, i.e.\ \emph{abelian} anyons). It is a difficult mathematical problem to obtain the behavior of such a gas on the side of \emph{exclusion} statistics, which in the case of bosons and fermions is the most concrete and observable way to distinguish between them, and for which powerful effective models have been developed, including the Hartree equation, the nonlinear Schr\"odinger (NLS) / Gross--Pitaevskii (GP) equation, and the Thomas--Fermi (TF) equation, respectively their celebrated density functional theories. Further, to add to the difficulty, realistic anyons also have some size and additional internal structure, rather than being pointlike and ideal. We refer to \cite{LunQva-20,Girardot-21,Lundholm-24} for recent reviews of these matters. Here we focus on a microscopic model for anyons which admits some simple internal structure, determined by a radius/size parameter $R$ and a spin-orbit coupling parameter $g$; see \cite{Grundberg-etal-91,ComMasOuv-95,Mashkevich-96}. In the almost-bosonic limit, in which the total number of magnetic flux units $\beta := \alpha(N-1)$ is finite, while at the same time $R \to 0$ at an arbitrary rate with $N \to \infty$, we use a Hartree--Jastrow trial state wave function to obtain an effective mesoscopic description of the anyon gas in terms of an energy density functional with a Chern--Simons magnetic self-interaction. Similar effective models have been considered in the physics literature in the context of the fractional quantum Hall effect (FQHE); see 
\cite{ZhaHanKiv-89,JacWei-90,LopFra-91,IenLec-92,JacPi-92,Zhang-92,Dunne-95,Dunne-99,Khare-05,HorZha-09,CorDubLunRou-19}.
In the mathematics literature, the most relevant works are 
\cite{LieSeiSolYng-05,Tarantello-08,LunRou-15,LarLun-16,CorLunRou-17,CorOdd-18,MicNamOlg-19,Oddis-20,Girardot-20,RajSig-20,GirRou-21,LamLunRou-22,AtaLunNgu-24,GirLee-24}, and we will explain these connections below.

In the first part of the article, Section \ref{sec:Mainresult}, we give the mathematical framework by introducing the Hamiltonian and the parameters of the microscopic model. We define our trial state wave function $\PsiTrial$, see \eqref{eq:trial-f}--\eqref{def:trial}, and state our main theorem together with the most direct remarks about the result. Then follows a longer discussion to better explain and contextualize the choice of our model, something which we felt was presently missing in the literature, especially from a mathematically rigorous perspective. In the next part, Section \ref{sec:twobody}, we motivate our precise Jastrow factor by studying simplified versions of the two-body problem. Finally, we bring the proofs of our main results in Section \ref{sec:manybody}.
In the appendix we have gathered some previous knowledge of more technical character, that we use.

We would like to stress the fact that, despite Section~\ref{sec:manybody} being mainly computational, making the understanding obtained in Section~\ref{sec:twobody} rigorous in the actual many-body setting, it also clarifies the interplay between many approximation schemes that are well known in the physics literature. This is discussed in Remark~\ref{rem:approx} 
and in Section~\ref{sec:R>0}. We therefore believe that even this more technical part may be of interest not only to mathematicians but also to physicists interested in the many-anyon problem.

\subsection{Mathematical main result}\label{sec:Mainresult}

In this part, we define the precise microscopic model and state our main result, the derivation of the effective mesoscopic model. To facilitate the readership, we postpone the discussions on the physics motivation and background to later sections.

Let $\alpha, R \in \R_+ := [0,\infty)$, $g \in \R$, and $V \in C^{\infty}(\R^2;\R_+)$ be parameters of our microscopic model. 
Define the magnetic gauge potentials
$\sA^R = (\bAR_1,\ldots,\bAR_N)$, with
$\bAR_j\colon \R^{2N} \to \R^2$,
\begin{equation}\label{def:ARj}
    \bAR_j(\sx):=\sum_{\substack{k=1\\ k\neq j}}^N\frac{(\bx_j-\bx_k)^{\perp}}{|\bx_j -\bx_k|_R^2}
    \quad \text{where}\quad 
    |\bx|_R:=\max(R, |\bx|),
\end{equation}
for every $\sx = (\bx_1,\cdot \cdot \cdot, \bx_N) \in \R^{2N}, N \in \N$, where $(x,y)^{\perp}:=(-y ,x)$ for any $(x,y) = \bx \in \R^2$. 
To these potentials correspond the magnetic fields $B^R_j\colon \R^{2N} \to \R$,
\begin{equation} 
    B^R_j(\sx):=\Curl_{\bx_j} \bAR_j(\sx) = 
    \begin{cases} \displaystyle
    2\pi\sum_{\substack{k=1\\ k\neq j}}^N \frac{\1_{B(0,R)}(\bx_j -\bx_k)}{\pi R^2}, 
        & R>0,\\ \displaystyle
    2\pi\sum\limits_{\substack{k=1\\ k\neq j}}^N\delta(\bx_j-\bx_k), 
        & R=0,
    \end{cases}
\end{equation} 
where $\delta(\bx)$ is the unit Dirac mass centered at the origin. 
For every $N \in \N$, we consider the following Hamiltonian operator
\begin{equation}\label{def:HN}
 H_{N}:=\sum_{j=1}^{N}\left(-\im\nabla_{x_j}+\alpha \bAR_j(\sx)\right)^{2}
  +  \frac{g\alpha}{2} \sum_{j=1}^{N}B^R_j(\sx)
  +\sum_{j=1}^{N}V(\bx_{j}),
\end{equation}
acting as a self-adjoint operator in the Hilbert space $L^2_{\sym}(\R^{2n};\C)$. We pick $u\in C_c^{\infty}(\R^2;\C)$ and respectively define the condensate function together with the Jastrow product 
\begin{equation}
    \Phi(\bx_1,\cdot \cdot \cdot,\bx_N):=\prod_{j=1}^N u(\bx_j) \quad\text{and}\quad F(\bx_1,\cdot \cdot \cdot,\bx_N):=\prod_{1\leq j<k\leq N}f(\bx_j-\bx_k),
\end{equation}
for every $(\bx_1,\cdot \cdot \cdot,\bx_N) \in \R^{2N}$, where
\begin{equation}
f(\bx) = \begin{cases}\label{eq:trial-f}
 \frac{4 \left(\frac{R}{b}\right)^{\alpha} }{2\left(1+ \left(\frac{R}{b}\right)^{2\alpha}\right) + g\left(1-\left(\frac{R}{b}\right)^{2\alpha}\right)} , \quad &\textup{if } \bx \in B(0,R),\\\\
\frac{(2+g) \left(\frac{|\bx|}{b}\right)^{\alpha} + (2-g) \left(\frac{R}{b}\right)^{2\alpha} \left(\frac{b}{|\bx|}\right)^{\alpha} }{2\left(1+ \left(\frac{R}{b}\right)^{2\alpha}\right) + g\left(1-\left(\frac{R}{b}\right)^{2\alpha}\right)}, \quad &\textup{if } \bx \in A(R,b),\\\\
1, \quad &\textup{if } \bx \in \R^2 \setminus B(0,b),
\end{cases}
\end{equation}
depending on the parameter $b >R$.
Here and in the following $B$ denotes balls/disks and $A$ annuli, centered at 0. Now, define the trial state function
\begin{equation}\label{def:trial}
    \Psi_N^{\mathrm{trial}}:=F\Phi 
        \in L^2_\sym(\R^{2N}),
\end{equation}
and the quantity
\begin{equation}\label{def:G}
    G(s ,g):= \frac{1+\frac{g}{2} - (1-\frac{g}{2}) e^{-s}}{1+\frac{g}{2} + (1-\frac{g}{2}) e^{-s}}, 
\end{equation}
defined for every $s \geq 0$ and $g>-2.$ Finally, consider the one-body density associated with $\Psi_N^{\mathrm{trial}}$ by 
\begin{equation}\label{eq:one-body-density}
    \varrho_{\Psi_N^{\mathrm{trial}}}(\bx) := N \int_{\R^{2(N-1)}} \left|\Psi_N^{\mathrm{trial}}(\bx,\bx_2,\ldots,\bx_N)\right|^2 \,\dd\bx_2 \ldots \dd\bx_N,
\end{equation}
and the self-generated gauge potential
\begin{equation} 
    \bA[\varrho]:=\frac{\bx^{\perp}}{|\bx|^2}\ast \varrho, 
    \quad \textup{in } \R^2,
\end{equation}
satisfying $\Curl \bA[\varrho] = 2 \pi \varrho$, proportional to the probability density $\varrho := |u|^2 \in L^1(\R^2;\R_+)$. 
Our main 
result is as follows:
\begin{theorem}\label{thm:E}
        For each integer $N \ge 2$,
        let the Hamiltonian $H_N$ be defined as in \eqref{def:HN} with the 
        $N$-dependent parameters
        $\alpha=\alpha_N \ge 0$, $R=R_N \ge 0$, and $g=g_N \geq 0$. Consider $ \Psi_N^{\mathrm{trial}}$ as the trial state given in \eqref{def:trial} with
        $b_N>R_N$ and
        $u\in C^{\infty}_c(B(0,R_1))$ for some fix $R_1>0$ and where $\int_{\R^2}|u|^2=1$. 
        Let 
        \begin{align*}
            \beta_N := (N-1)\alpha_N\quad\textup{and}\quad\omega_N := -\frac{\log R_N}{N},
        \end{align*} and assume that 
        \begin{equation}
        \begin{aligned}
        \label{eq:assumptions}
         &\lim_{N \to \infty} \beta_N = \beta \in \R_+,\qquad
         \lim_{N \to \infty} \omega_N = \omega \in [0,+\infty],\\
         &\lim_{N \to \infty} N^2 b_N =0,  \,\,\,\qquad\quad  \lim_{N \to \infty} \frac{\log b_N}{N} = 0. \nonumber
        \end{aligned}  
        \end{equation}
    Then,
    \begin{equation}\label{eq:effectivedensityfunction}
       \lim_{N \to \infty}     \frac{\bra{ \Psi_N^{\mathrm{trial}}}\ket{H_N  \Psi_N^{\mathrm{trial}}}}{N\norm{ \Psi_N^{\mathrm{trial}}}^2}=\int_{\mathbb{R}^{2}}\left|(-\im\nabla+\beta \bA[|u|^{2}])u\right|^{2}+\int_{\R^2}V|u|^{2}+ 2\pi \beta\, G( 2\beta\omega ,g)\int_{\R^2}|u|^4, \nonumber
    \end{equation}
    together with 
    \begin{equation}\label{eq:effectivedensity}
        \lim_{N \to \infty} \norm{ N^{-1}\varrho_{\Psi_N^{\mathrm{trial}}} - |u|^2 }_{L^1(\R^2)} = 0.
    \end{equation}
\end{theorem}

It is worth noticing some specific values of $G( 2\beta\omega ,g)$ in the different regimes of the coupling constant $g$ and rate of convergence of the radius $R_N = e^{-N\omega_N} \to 0$:
\begin{align}
2\pi\beta  G( 2\beta\omega ,g)=  \begin{cases}
        g\pi \beta\quad\text{in the sub-critical case}\quad \omega=0,\\
        2\pi\beta \quad\text{in the super-critical case}\quad \omega=+\infty,\\ 
        2\pi\beta \quad\text{in the supersymmetric case}\quad g=2,\\
        2\pi\beta \coth{\beta \omega}\quad\text{in the hard-disk case}\quad g=+\infty,\\
        2\pi\beta \tanh{\beta \omega}\quad\text{in the non-interacting/spinless case}\quad g=0.
    \end{cases}
    \label{eq:DeltaEcases}
\end{align}

We also have the following immediate remarks concerning the result:

\begin{remark}[Sub-critical/polynomial limit]\label{rem:omega=0}
    For $\omega=0$ we may take $f=1$. In this case $F=1$ and thus $b$ is irrelevant. This allows any polynomial rate $R_N \to 0$ for which \eqref{eq:effectivedensityfunction} remains valid with $G(0,g) = g/2$.
    This case has been proved already in \cite{LunRou-15} for $g=0$ and carries over immediately to any $g \in \R$.
    For $\omega>0$, or a non-polynomial rate, the energy 
    (l.h.s. of \eqref{eq:effectivedensityfunction})
    diverges for $f=1$ due to the singular two-body interaction.
    Thus, the extensions to the critical and the super-critical cases are the main novelty of this work.
\end{remark}

\begin{remark}[Scattering length]\label{rem:scattlen}
    We can define a relative 2-body ``scattering energy'' (compare \cite[Appendix~C]{LieSeiSolYng-05} and \cite[Chapter~5.1]{Rougerie-21} for the ``scattering length'' for bosons) for this problem as
    $$
        E_2 := \inf \left\{ \int_{B(0,1)} \left( |\nabla f|^2 + \alpha^2 \frac{|\bx|^2}{|\bx|_R^4}|f|^2 + \alpha g \frac{\1_{B(0,R)} }{R^{2}}|f|^2 \right) : f \in H^1(B(0,1)), \ f|_{|\bx|=1}=1 
        \right\}.
    $$
    Our Jastrow pair correlations $f$ in $\PsiTrial$ are defined so as to approximate this minimizer on a scale $|\bx|<b$, and yields an energy per particle $(N-1)E_2 \simeq 2\pi\beta G(2\beta\omega,g)$ as $N \to \infty$ (heuristically, we estimate the total energy by multiplying the relative mass 2 times number of pairs $N(N-1)/2$).
    Therefore, we believe that our choice of $f$ and obtained 2-body correlation energy in \eqref{eq:effectivedensityfunction} is optimal to leading order for a wide range of the parameters.
    A similar expression was found using a harmonic regularization in \cite{Mashkevich-96},
    and we discuss these aspects further in Section~\ref{sec:twobody}.
\end{remark}

\begin{remark}[Supersymmetry and scale invariance]\label{rem:susy}
    For the special case of $g=2$, where the spin interaction is supersymmetric, we derive that $G(2\beta \omega,g) = 1$ 
    for all $\omega$ and
\begin{align*}
     2\pi \beta \int_{\R^2}|u|^4 = \int_{\R^2}  \left(\Curl \left( \beta \bA[|u|^{2}] \right)\right)|u|^2
\end{align*}
    is of the form of a Pauli spin-orbit interaction aligned with the self-generated field.
\end{remark}

\begin{remark}[Hard-disk limit]\label{rem:harddisk}
    In the hard-disk limit $g_N \to +\infty$, we cannot simultaneously treat sub-critical cases $\omega_N\to 0$. In that regime, $2\pi\beta G(2\beta\omega,g=+\infty) \simeq \frac{2\pi}{\omega}$, denoting a breakdown of our approach. 
    Indeed, considering the scattering energy $E_2$ above, its minimization constrains $f$ to be exactly zero inside a disk of radius $R$. In the limit $\omega \to 0$ we should then rather consider the high-density regime of such a strongly interacting model (see \cite{Dyson-57,LieSeiSolYng-05}). Our limit case $\omega \to 0$, $g=+\infty$ nevertheless indicates that the dependence on $\beta$ would disappear to leading order, as may be expected.
\end{remark}

\begin{remark}[Simplified Jastrow ansatz]\label{rem:simplejastrow}
    ``Jastrow-type'' trial states have been considered in the interacting bosonic context ever since the old works of Bijl \cite{Bijl-40}, Dingle \cite{Dingle-49}, Jastrow \cite{Jastrow-55}, Dyson \cite{Dyson-57}, as well as more recently 
    \cite{MicNamOlg-19}.
    We could also consider a much simpler function $f_{\theta}(\bx)=|\bx|^{\alpha\theta}$ instead of \eqref{eq:trial-f} to build our Jastrow product \eqref{def:trial}, with $\theta \in [0,1]$ a variational parameter. A final optimization over $\theta$ would then provide an accurate upper bound, but only in some specific regimes of the parameters $g, R, \beta$. We refer to Section~\ref{sec:ftheta} for a more detailed discussion and the precise limit \eqref{eq:effectivedensityfunction1} we obtain in this simplified case.
\end{remark}

\begin{remark}[Attractive case and stability]\label{rem:attractive}
    Above and in the rest of the work we assume $g \ge 0$, however,
    we can also consider $g<0$ which corresponds to the attractive case. Our proof relies on the Jastrow factor $f$ satisfying $ 0\leq f\leq 1$, and this holds still if
    \begin{align*}
    g>-2, \qquad 
    g \geq  -2 \frac{1-e^{-\beta \omega}}{1+e^{-\beta \omega}}.
\end{align*} Hence, up to the threshold above, we expect the result of Theorem \ref{thm:E} holds, and in this regime of $g$ we always have
\begin{align*}
     G( 2\beta\omega ,g) \geq 0,
\end{align*}
which gives a repulsive average energy functional. Although our proof fails to work for more negative $g$, we believe that the same result should apply for 
\begin{align*}
     g > -2 \frac{1+e^{-2\beta \omega}}{1-e^{-2\beta \omega}} =: g_{\mathrm{crit}},
\end{align*}
such that at the critical $g_{\mathrm{crit}}$ we have 
\begin{align*}
   \lim_{g \to g_{\mathrm{crit}}^+} G( 2\beta\omega ,g) = - \infty.
\end{align*}
We predict that for $ g \leq g_{\mathrm{crit}}$, the many-body system fails to be stable. Moreover, since on the effective mesoscopic side we are interested in a minimizer of the energy functional over $u \in H^1(\R^2)$ with the constraint $\int_{\R^2} |u|^2=1$, it is crucial that $2\pi \beta G(2\beta \omega,g) \geq - \gamma_{\ast}(\beta)$ holds; see Section \ref{sec:appendix-CSS}. In conclusion, we predict that for the stability of the system we need 
\begin{align*}
    g \geq \max \left( -2 \frac{1+e^{-2\beta \omega}}{1-e^{-2\beta \omega}} , -2 \frac{1-e^{-2\beta \omega} +  \frac{\gamma_{\ast}(\beta)}{2\pi \beta} \left(1+e^{-2\beta \omega}\right) }{1+e^{-2\beta \omega} +  \frac{\gamma_{\ast}(\beta)}{2\pi \beta} \left(1-e^{-2\beta \omega}\right)} \right).
\end{align*}
\end{remark}

\begin{remark}
\label{rem:approx}
It is interesting and perhaps illustrative to note that our model \eqref{def:HN} together with the Jastrow trial state \eqref{def:trial} combines the three main perturbative fashion approaches studied in \cite{Ame_95}. The action of the Hamiltonian on the  Jastrow factor is indeed discussed in \cite[Section A]{Ame_95} in order to renormalize the logarithmic divergence for $\alpha^2 /|\bx|^2$. The results are then compared to two other approaches. \cite[Section B]{Ame_95} consists in regularizing the wave function along the diagonal and the last one, \cite[Section C]{Ame_95}, in adding by hand a repulsive delta potential (with no physical effect) to the Hamiltonian. These two last procedures correspond, respectively, to our smoothing procedure $R\to 0$ and the spin-orbit coupling of strength $g$. Our calculations then precisely show how the three different approaches combine in the final quantity $G(2\beta\omega ,g)$ of \eqref{def:G} and clarify the effect of each of them.
\end{remark}

\subsection{Discussion: origins and consequences for physics}

We now turn to the physical background for the model defined above, mention the previous results that are the most relevant for our work, and discuss the motivation for our trial state ansatz as well as possible consequences of our main results for the mathematical physics of the anyon gas.
This subsection can be skipped if one is already familiar with the context.

\subsubsection{Pointlike and ideal anyons}

Anyons are identical quantum particles which carry a generalization \cite{LeiMyr-77,GolMenSha-80,GolMenSha-81,Wilczek-82a,Wilczek-82b} of the exchange quantum statistics of bosons and fermions. This is possible in two spatial dimensions if orientation symmetry is suitably broken (as well as in one dimension, which we do not consider here). Namely, for bosons and fermions in the Euclidean plane $\R^2$, one requires that their $N$-body Schr\"odinger wave function $\Psi\colon \R^{2N} \to \C$ obeys the symmetry condition
\begin{equation}\label{eq:Psi-perm}
	\Psi(\bx_1,\ldots,\bx_k,\ldots,\bx_j,\ldots,\bx_N) 
	= \pm \Psi(\bx_1,\ldots,\bx_j,\ldots,\bx_k,\ldots,\bx_N),
\end{equation}
exchanging/permuting particles $j \leftrightarrow k$,
where $+$ defines bosons, and $-$ defines fermions.
We denote the Hilbert space of square-integrable such permutation symmetric resp. antisymmetric functions by $L^2_\sym(\R^{2N})$ resp. $L^2_\asym(\R^{2N})$.
For anyons, we require\footnote{This equation assumes a simple counter-clockwise exchange of two particles with no other particles involved. In case that $p$ other particles are enclosed under a simple exchange, the exchange phase must be $e^{\im(2p+1)\alpha\pi}$. In general, one must use the braid group of $N$ strands to specify the phase of an exchange, and also note that the function $\Psi$ as defined is necessarily multivalued except when $\alpha \in \Z$ (bosons or fermions).}
\begin{equation}\label{eq:Psi-exch}
    \Psi(\bx_1,\ldots,\bx_k,\ldots,\bx_j,\ldots,\bx_N) 
	= e^{\im\alpha\pi} \Psi(\bx_1,\ldots,\bx_j,\ldots,\bx_k,\ldots,\bx_N),
\end{equation}
where the simple exchange phase $e^{\im\alpha\pi}$, or equivalently the statistics parameter $\alpha \in \R (\mathrm{mod}\,2)$ defines the type of anyon considered (these are known as \emph{abelian} anyons; the concept can be generalized to nonabelian unitary matrices or exchange operators, although we do not consider that here).
Note that for these \emph{pointlike} and \emph{ideal} anyons, there is a periodicity $\alpha \to \alpha + 2$, and often (but not always) also complex conjugation symmetry $\alpha \to -\alpha$, so we may often impose the simplifying assumption $\alpha \in [0,1]$, where $\alpha=0$ corresponds to bosons and $\alpha=1$ to fermions.

A typical microscopic $N$-body Hamiltonian for the energy of a gas of such particles is
\begin{equation}\label{eq:hamil-start}
 H_{N} = \sum_{j=1}^{N} (-\im\nabla_{\bx_j})^{2}
 + \sum_{j=1}^{N} V(\bx_j)
 + \sum\limits_{\substack{j,k=1\\ j<k}}^N W(\bx_j-\bx_k),
\end{equation}
where we use the non-relativistic kinetic energy $\bp^2/(2m)$ and choose units suitably to set the physical constants $c=e=\hbar=1$ and $2m=1$. We allow for an external one-body potential $V$, and possibly a two-body interaction potential $W$, which will also aid our discussion on the different possible types of anyons at fixed $\alpha$.
We shall typically make the assumption that $V\colon \R^2 \to \R_+$ is trapping, 
in the sense that $V(\bx)\to \infty$ when $|\bx|\to \infty$. More precise detail on $V$ will be provided later when stating the results.
The operator \eqref{eq:hamil-start} is required to be self-adjoint for unitary quantum dynamics, and is formally defined via a suitable quadratic form (Friedrichs extension; see Appendix~\ref{sec:appendix-domains} and cf. \cite{LunSol-13b,CorOdd-18,LunQva-20}), and involves some subtleties we shall return to below (then also fixing $W$).

In practice, one often uses an equivalent description (known as the magnetic/singlevalued gauge rather than the anyon/multivalued gauge), which implements the exchange conditions \eqref{eq:Psi-exch} by applying the canonical isomorphism
$$
    \R^2 \ni  (x,y) = \bx
    \quad \leftrightarrow \quad
    z = x + \im y \in \C,
$$
and similarly for $\bx_j \leftrightarrow z_j$, and writing
$$
    \Psi(\bx_1,\ldots,\bx_N) 
    =\prod_{j<k}e^{\im\alpha\phi_{jk}}\Psi_0(\bx_1,\ldots,\bx_N),
    \quad \text{where} \quad
    \phi_{jk}=\arg\frac{z_{j}-z_{k}}{|z_{j}-z_{k}|},
$$
with $\Psi_0 \in L^2_\sym$ a bosonic wave function. 
The Hamiltonian \eqref{eq:hamil-start} then becomes
\begin{equation}\label{eq:hamil-transmuted}
    H_N = \sum_{j=1}^{N}(-\im\nabla_{\bx_j} + \alpha \bA_j)^{2}
    + \sum_{j=1}^{N}V(\bx_j)
    + \sum_{j<k} W(\bx_j-\bx_k),
\end{equation}
where we have introduced the magnetic gauge potentials
\begin{equation}\label{eq:mag-pointlike}
    \bA_j(\sx) :=\sum\limits_{\substack{k=1\\ k\neq j}}^N \frac{(\bx_{j}-\bx_{k})^{\perp}}{|\bx_{j}-\bx_{k}|^{2}}
    \quad\text{with}\quad 
    \Curl_{\bx_j} \bA_j(\sx) = 2\pi\sum\limits_{\substack{k=1\\ k\neq j}}^N \delta(\bx_j-\bx_k).
\end{equation}
Note that indeed $U^{-1} \nabla_\bx U = \nabla_\bx + \im\bx^\perp/|\bx|^2$ if $U(z) = e^{\im\arg z} = z/|z|$, which represents an Aharonov-Bohm unit phase at $\bx=\0$ or $z=0$.
Thus, the operator $\alpha\bA_j$ is the statistical gauge vector potential felt by the particle $j$ due to the influence of all the other particles $k \neq j$, by the attachment of Aharonov-Bohm magnetic point fluxes of intensity $2\pi\alpha$ to each particle.

Now, in this description it is clear that there are potential problems or ambiguities in the definition of the Hamiltonian \eqref{eq:hamil-transmuted} since the potentials $\bA_j$ are singular (not $L^2_\loc$) close to the diagonals (coincidence set) of the configuration space
$$
    \bDelta := \left\{ \sx = (\bx_1,\ldots,\bx_N) \in \R^{2N} : \exists j \neq k \ \text{s.t.} \ \bx_j = \bx_k \right\},
$$
and indeed this turns out to be a crucial observation, which has led to much discussion in the literature.
We shall only bring up the most relevant points for our purposes,
and refer to the review articles \cite{Lundholm-23,LunQva-20} and Ph.D. theses \cite{Oddis-20,Girardot-21} for more details on the definition of the anyon gas and on the statistics transmutation phenomenon.

The complete theory of self-adjoint extensions for the class of operators \eqref{eq:hamil-transmuted}, starting from the minimal domain $C^\infty_c(\R^{2N} \setminus \bDelta) \cap L^2_\sym$, is in the general case $N$ still poorly understood, however for $N=2$ (and reducing to the relative coordinates $\bx = \bx_1-\bx_2$) it is mathematically well understood (see \cite{BorSor-92,AdaTet-98,CorOdd-18,Oddis-20,BorCorFer-24}), and we shall therefore use this special case with $V=0=W$ to set our notations and expectations. Thus, starting from the minimal domain $C^\infty_c(\R^2 \setminus \{0\}) \cap L^2_\sym$ for the kinetic energy operator $T = (-\im\nabla_\bx + \alpha\bx^\perp/|\bx|^2)^2$ with $\alpha \in (0,1)$, there is a one-parameter family of self-adjoint extensions, parametrizing solutions to the defect indices with behavior
$$
    \Psi(\bx) \sim a r^\alpha + b r^{-\alpha},
    \qquad \text{as $r = |\bx| \to 0$}.
$$
Most of these extensions have scale and some of those have negative bound states, but exactly two of the extensions are both non-negative and scale-covariant. These are the Friedrichs extension ($b=0$) and the Krein extension ($a=0$), corresponding to the minimal form domain (i.e. the largest energy or the most positive; in $H^1_\sym(\R^2)$) resp. the maximal/extended form domain (i.e. the smallest energy or the least positive; not in $H^1_\sym(\R^2)$); cf. \cite{Simon-79b,AloSim-80}. In the physics literature, these are referred to as hard-core/repulsive/regular resp. soft-core/attractive/singular anyons; see \cite{Girvin-etal-90,Grundberg-etal-91,ManTar-91,Mashkevich-96}.
Further, in the well-studied $N$-anyon problem with harmonic confinement, $V(\bx) = |\bx|^2$, both of these types of anyons can be found among the exact (generalized) eigenfunctions of the problem and are then also referred to as the interpolating solutions (see \cite{Chou-91b,Chou-91a,ChiSen-92}) resp. the noninterpolating solutions (see \cite{MurLawBhaDat-92}).
Namely, for small enough $\alpha N$, the ground-state energy\footnote{%
Here we mean that we include the respective solutions in the domain, and ignore any issues at $\bDelta$. This can be done consistently at least for $N=2$, and for all $N$ in the positive case $\Psi_+$.}
$E(N) := \infspec H_N$ of \eqref{eq:hamil-transmuted} with $V(\bx)=|\bx|^2$ and $W=0$ (outside $\bDelta$) is exactly
$$
    E(N) = 2N \pm \alpha N(N-1),
$$
with the ground state
\begin{equation}\label{eq:true-gs}
    \Psi_{\pm} \propto \prod_{j<k} |\bx_j-\bx_k|^{\pm\alpha} e^{-\frac{1}{2}|\sx|^2}.
\end{equation}
The regular state $\Psi_+$ is in $H^1_\sym$ and in the domain of the Friedrichs extension, while the singular $\Psi_-$ is neither and even drops out of $L^2$ for large enough $\alpha$ (for $\alpha N \ge 2$).
In \cite{ChiSen-92} also a more general bound was given for the ground-state energy of regular states with total angular momentum $L$ ($L=0$ for $\Psi_{\pm}$ in \eqref{eq:true-gs}):
\begin{equation}\label{eq:angmom-bound}
    E(N) \ge 2\left(N + \left|L + \alpha \frac{N(N-1)}{2}\right| \right)
    \qquad \forall \alpha \in [0,1].
\end{equation}
Except for the eigenstates \eqref{eq:true-gs}, and some similar ones with higher energy, the spectrum of $H_N$ is largely unknown for $N > 2$
(see \cite{MurLawBraBha-91,SpoVerZah-91,SpoVerZah-92,ChiSen-92} for numerics and schematics).
Some rigorous bounds for arbitrary $\alpha$ were considered in \cite{LunSol-13a,LunSei-17,LunQva-20}.

For definiteness and following \cite{Lundholm-23}, at general $N$ we shall refer to the Friedrichs extension of \eqref{eq:hamil-transmuted}, i.e. in $H^1_\sym(\R^{2N})$ and $H^2_\loc(\R^{2N} \setminus \bDelta)$, with the expected regular behavior
$$
    \Psi(\sx) \sim |\bx_j-\bx_k|^\alpha 
    \quad \text{as} \quad
    \bx_j - \bx_k \to 0, 
    \qquad \text{as \keyword{ideal anyons}}, 
$$
and to any generalized solutions to the anyon Hamiltonian with the singular behavior
$$
    \Psi(\sx) \sim |\bx_j-\bx_k|^{-\alpha}
    \quad \text{as} \quad
    \bx_j - \bx_k \to 0, 
    \qquad \text{as \keyword{kreinyons}}.
$$
Intuitively, we may think of the corresponding choice of domain for self-adjoint $H_N$ as a choice of the pair interaction $W$ being either a positive or negative $\delta(\bx)$, i.e. yielding a point interaction supported on $\bDelta$.
See the Appendix \ref{sec:appendix-domains} for more precise information about the form domain for ideal anyons, clarifying that it actually consists of functions symmetrized from $H^1_\asym$  (i.e. more akin to fermions than to bosons).

\subsubsection{Extended anyons and average-field approximation}

Another means of dealing with the singular nature of the statistics gauge potentials is by regularizing the flux.
In the ``average-field'' approach\footnote{See Wilczek's book \cite{Wilczek-90} intro chapter concerning the terminology ``average-field'' vs. ``mean-field''.}, we replace the above sharply localized point-particle fluxes by the smoother magnetic field generated by their averaged (locally, on a mesoscopic scale) density distribution. This is done by first extending the particles' fluxes to finite but microscopic disks (i.e. ``smearing them out''), then averaging them by condensing or 
reducing their many degrees of freedom to a collective one-body state with a mesoscopic local density, and finally taking a limit of small extension at a rate depending on the actual energy and length scales of the condensed/collective problem. Note that actual emergent anyons may indeed be expected to have a finite size (on the order of the magnetic length of a strong external field; see, e.g., \cite{LunRou-16,Yakaboylu-etal-19,LamLunRou-22}), but since that is typically much smaller than the domain of an experiment, it is legitimate to consider a pointlike limit if the relevant microscopic information can be retained.

Thus, in the first step of this program, let us consider the 2D Coulomb potential generated by a unit charge smeared over the disc of radius $R$:
\begin{equation}\label{eq:wr}
 w_{R}(\bx) := \left(\log\b\cdot\b *\frac{\1_{B(0,R)}}{\pi R^2}\right)(\bx),
 \quad \text{with the convention} \quad 
 w_{0} := \log\b\cdot\b.
\end{equation}
By radial symmetry, we find that
$$
    w_R(\bx) = \begin{cases}
        \log |\bx|, & |\bx| > R, \\
        \log R + \frac{1}{2}(|\bx|^2/R^2-1), & |\bx| \le R,
    \end{cases}  
$$
and thus $\nabla w_R(\bx) = \bx_R^{-1} := \bx/|\bx|_R^2$, where 
$|\bx|_R:=\max(R, |\bx|)$. We define \keyword{$R$-extended anyons} by the smeared magnetic potentials
\begin{equation}\label{eq:ARj-def}
    \bAR_j(\sx)
    := \sum_{\substack{k=1\\ k\neq j}}^N (\nabla^\perp w_R) (\bx_j-\bx_k)
    = \sum_{\substack{k=1\\ k\neq j}}^N \frac{(\bx_j-\bx_k)^{\perp}}{|\bx_j -\bx_k|_R^2},
\end{equation}
so that, with the fundamental solution $\Delta \log |\bx| = 2\pi\delta(\bx)$, their corresponding magnetic fields are
\begin{equation}\label{eq:BRj-def}
    B^R_j(\sx) := \Curl_{\bx_j} \bAR_j(\sx) 
    = 2\pi \sum_{\substack{k=1\\ k\neq j}}^N \left(\frac{\1_{B(0,R)}}{\pi R^2}\right) (\bx_j-\bx_k)
\end{equation}
times the fraction of flux units $\alpha$.
Thus, in the limit of zero extension $R \to 0$, we consistently obtain pointlike anyons \eqref{eq:mag-pointlike}.

Summarizing, the $N$-body Hamiltonian operator that we consider is now
\begin{equation}\label{eq:HNR-def}
 H_N^R = \sum_{j=1}^{N} \left(-\im\nabla_{\bx_j}+\alpha \bAR_j(\sx)\right)^{2}
  + \sum_{j=1}^{N}V(\bx_j)
  + \sum_{j<k} W(\bx_j-\bx_k).
\end{equation}
Assuming regular enough $V$ and $W$, the self-adjointness of this operator is now much more straightforward, namely it is essentially self-adjoint on the minimal domain $C^\infty_c \cap L^2_\sym(\R^{2N})$ since $\alpha\bA_j^{R}$ with $R>0$ is a bounded and thus relatively small perturbation (with zero divergence in our gauge; see, e.g., \cite[Theorem~3.2]{Simon-79b}).

Extended anyons were considered (non-rigorously) in \cite{ChoLeeLee-92,Trugenberger-92b,Trugenberger-92}, which also introduced an associated dimensionless scale parameter, the
$$
    \text{\keyword{magnetic filling ratio} \quad $\nu := R\sqrt{\rho}$,} 
$$
to describe the phases of the $R$-extended anyon gas at average density $\rho$. The dilute limit $\nu \to 0$ then corresponds to pointlike anyons, while $\nu \sim 1$ is a ``smearing'' regime where the internal structure of anyons matters more, and $\nu \gg 1$ describes bosons in an almost-constant magnetic field. Thus, the larger $\nu$ is, the more valid the average-field approximation is expected to be. 
Note that in the almost-constant limit, with $N$ bosons in a fully degenerate lowest Landau level of a constant magnetic field of intensity $B \approx 2\pi\alpha\rho$ considered on an area $L^2$, 
and at density $\rho = N/L^2$,
the energy per unit area will be
$$
    E(N)/L^2 \approx |B| N/L^2 \approx 2\pi|\alpha| \rho^2.
$$
This suggests the rough local density approximation for the energy of a normalized ground state $\Psi \in L^2_\sym$ with one-body density $\varrho(\bx)$, adjusted to the trapping potential $V$:
\begin{equation}\label{eq:almost-const-approx}
    \inp{\Psi, H_N^R \Psi} \approx \int_{\R^2} \left( 2\pi|\alpha| \varrho(\bx)^2 + V(\bx)\varrho(\bx) \right) d\bx.
\end{equation}
Further heuristic motivation for the average-field approach, also for more ideal anyons, is given in \cite[Section III]{Lundholm-23} and references therein.
Rigorous derivations, refining \eqref{eq:almost-const-approx}, have been made in the almost-bosonic $\alpha \sim N^{-1}$ \cite{LunRou-15} and almost-fermionic $1-\alpha \sim N^{-1/2}$ \cite{GirRou-21} limits, and the former will be our main interest below.
Rough, yet rigorous, lower bounds for the energy of the homogeneous extended anyon gas were proved in \cite{LarLun-16} for arbitrary $\alpha$, $W=0$, and with a somewhat reasonable interpolation in $\nu$ from the almost-ideal (with a possible clustering aspect) to the almost-constant limit.

Note that in the extended anyons approach the domain of the Hamiltonian is really not the issue, but rather the precise limit of interplay between the internal structure, possibly including a scale-dependent interaction $W=W_{R,N}$, and the magnetic filling $\nu = \nu_{R,N}$, is important in order to know what kind of anyons we get at $R \to 0$ (ideal or otherwise).
We use the spin-coupling approach, already discussed in the literature, as a guide for the limit. This we discuss next.

\subsubsection{Spin-orbit coupled Hamiltonian}

We now supply additional microscopic details to our model by considering a very specific two-body interaction, namely
\begin{equation}\label{eq:W-def}
    W(\bx) = W_R(\bx) := \frac{2g\alpha}{R^2} \1_{B(0,R)} (\bx),
\end{equation}
so that
\begin{equation}\label{eq:W-mag}
    \sum_{j<k} W(\bx_j-\bx_k) 
    = \frac{g\alpha}{R^2} \sum_{j=1}^{N} \sum_{k \neq j} \1_{B(0,R)} (\bx_j-\bx_k)
    = \frac{g\alpha}{2} \sum_{j=1}^{N}B^R_j(\sx).
\end{equation}
Since the potential $W$ is bounded for $R>0$, the domain of $H_N^R$ with the above scalar interaction is the same as for $W=0$.

The interaction \eqref{eq:W-def} is a soft-disk potential with coupling constant $2g\alpha/R^2$, $\int_{\R^2} W = 2\pi\alpha g$, where $g \in \R$ is a new parameter.
In principle we could choose the constants differently, however we are here interested in this specific combination in order to view it as a coupling to the exact magnetic field of each particle, as in \eqref{eq:W-mag}, with coupling parameter $g/2$.

In fact, we interpret $g$ as a spin-orbit coupling parameter.
Consider a Pauli operator in 2D with external magnetic potential $\bA$ and field $B = \Curl \bA$, acting on $L^2(\R^2) \otimes \C^2$:
$$
    H^{\rm Pauli}_{\bA} = (-\im\nabla + \bA)^2 \otimes \1_{\C^2} + \frac{g}{2} B \mat{+1 & 0 \\ 0 & -1}
$$
(where in our units the Bohr magneton is $\mu_{\rm B} = e\hbar/(2m) = 1$).
For exactly $g=\pm 2$, this exhibits supersymmetry (see \cite{deCRit-83}) by completing the square of a 2D Dirac operator:
$$
    H^{\rm Pauli}_{\bA} = \left[ \sum_{k=1,2} \sigma_k (-\im\partial_k + A_k) \right]^2 \ge 0,
$$
where $\sigma_k \in \C^{2 \times 2}$ are Pauli matrices s.t. $\sigma_1\sigma_2\sigma_3 = \pm\im\1$.
Thus, the interacting $N$-anyon Hamiltonian
\begin{equation}\label{def:HNRg}
    H_N^{R,g} := \sum_{j=1}^{N} \left[ 
        \left(-\im\nabla_{\bx_j}+\alpha \bAR_j(\sx)\right)^{2}
        + \frac{g}{2} \alpha B^R_j(\sx) + V(\bx_j)
        \right]
\end{equation}
can indeed be interpreted as the projection to one of the spin components of such a spin-orbit coupled Pauli operator in the field of the other particles. 
In general, $g \in \R$ is interpreted as the magnetic moment or the gyromagnetic ratio (the ``Land\'e $g$-factor''), discussed to some extent in \cite{Grundberg-etal-91,ComMasOuv-95,Mashkevich-96}. However, we are not aware of any discussion which fixes its precise value in an emergent anyon model, and shall therefore leave it as one of the main microscopic parameters of our model and study its effect on the emergent mesoscopic model. 
We turn to this now.


\subsubsection{The almost-bosonic average-field approximation}

Now that flux is smeared out on some microscopic scale $R>0$, the next step of the average-field approximation program will be, heuristically, to replace $\bA^R_j$ and $B^R_j$ by their marginal expectation values in the state $\Psi$, which yields a self-consistent model with  
\begin{equation}\label{eq:avg-gauge}
    \bA^R_j \approx \bA^R[\varrho_\Psi] 
    := (\nabla^\perp w_R) * \varrho_\Psi
    \qquad \text{and} \qquad
    \alpha B^R_j \approx 2\pi\alpha \left(\frac{\1_{B(0,R)}}{\pi R^2}\right) * \varrho_\Psi,
\end{equation}
i.e.\ locally (mesoscopically) dependent on the one-body density \eqref{eq:one-body-density} associated to $\Psi$
(the marginal of the probability distribution $|\Psi|^2$; see \eqref{eq:def-density}).
We see that this is well motivated for rather large magnetic filling ratio $\nu$ by the almost-constant approximation \eqref{eq:almost-const-approx}, however, here we are actually interested in the dilute limit, i.e. taking $\nu$ small in order to re-obtain pointlike anyons (ideal, or possibly refined by their internal structure).
Deriving rigorously
the resulting effective model  
is our main interest in this work and will now be discussed.

Note that the case $\alpha =0$ plugged into \eqref{def:HNRg} provides the Hamiltonian of non-interacting trapped bosons. It has the product ground state $\Psi = \otimes^N u$ describing Bose--Einstein condensation in the normalized one-body ground state $u \in L^2(\R^2)$ of the Schr\"odinger operator $-\Delta + V$, satisfying the stationary Schr\"odinger equation $-\Delta u + Vu = \lambda u$.
Thus, sufficiently close to bosons $\alpha \approx 0$, we may compare to the successful GP/NLS approach for the dilute Bose gas (see \cite{LieSeiSolYng-05,MicNamOlg-19,Rougerie-21,Solovej-25} for mathematical review). Namely, for soft interaction we can still consider a product state (mean-field/Hartree) ansatz $\Psi = \otimes^N u$, while in the case of stronger interaction (such as hard disks) we expect that two-body correlations will play a significant role, such as discussed by Jastrow \cite{Jastrow-55}, Dyson \cite{Dyson-57}, and others.

Therefore,
in this work we will consider the following scaling of the microscopic parameters as $N\to\infty$:
\begin{equation}\label{eq:scaling}
    \alpha:=\beta(N-1)^{-1},\quad \beta\in \R^+,\quad R:=R_N\to 0\quad \text{and}\quad g\ge 0 .
\end{equation}
The limit $\alpha\to 0$ stands for the almost-bosonic anyons limit and, coupled with $R\to 0$, it thus 
corresponds to an average-field regime with anyons that are asymptotically pointlike and magnetically interacting with a self-generated field $B(\bx) \simeq 2\pi\alpha\varrho_\Psi(\bx) \simeq 2\pi\beta |u(\bx)|^2$ with finite total magnetic flux $2\pi\beta$. 
Note that at fixed $\beta$ and $R$ all terms of \eqref{def:HNRg} are of order $N$. This is the purpose of the rescaling and choice of the parameters in \eqref{eq:scaling}, to consider various effects at approximately the same energy scale, and then subsequently allowing $\beta \to 0$ to weaken the effect of statistics or $\beta \to \infty$ to strengthen the effect and move out of the almost-bosonic regime. 
Similarly, we expect $g \to 0$ resp. $g \to +\infty$ to soften resp. harden the interaction, with implications for the strength of effective correlations.
The limit of the radius $R_N\to 0$ will be studied considering different speeds (or scales), characterized by the parameter 
$\omega \in [0,+\infty]$,
also leading to possible different effective correlations. 
In the earlier rigorous work, $R_N \to 0$ was taken at a polynomial rate, typically describing a higher-density regime, however,
it turns out that the critical rate for our study is exponential, motivating our definitions $R_N = e^{-N\omega_N}$ and $\omega := \lim_{N \to \infty} \omega_N$,
where $\omega=0$ in the polynomial case, and for $\omega>0$ the resulting dependence is typically on the parameter combination (``\emph{scale}'') $s=2\beta\omega$.
In terms of the filling ratio $\nu$, locally 
$$
    \nu(\bx) := R\sqrt{\varrho_\Psi(\bx)}
    \simeq \sqrt{N}e^{-N\omega}|u(\bx)|,
$$
this always corresponds to a dilute limit unless the local density $\varrho(\bx) = |u(\bx)|^2$ is extremely high.
However, this can be controlled 
due to the recently developed regularity and stability theory for the resulting effective model, which we discuss next.

\subsubsection{The Chern--Simons--Schr\"odinger energy}

On the mesoscopic level we consider the following one-particle average-field(-Pauli) energy functional, also called Chern--Simons--Schr\"odinger stationary functional:
\begin{equation}
 \cE_{\beta, \gamma, V}[u] := \int_{\mathbb{R}^{2}}\left|(-\im\nabla+\beta \bA[|u|^{2}])u\right|^{2}+\int_{\R^2}V|u|^{2}+\gamma\int_{\R^2}|u|^4,
 \label{def:Eaf}
\end{equation}
where the coupling parameter $\gamma\in \R$ accounts for scalar self-interaction, and where 
\begin{equation}\label{def:Au}
    \bA[\varrho] := 
    \left( \nabla^\perp \log |\slot| \right) \ast \varrho =
    \frac{\bx^{\perp}}{|\bx|^2} \ast \varrho
\end{equation}
is the averaged vector potential \eqref{eq:avg-gauge} at $R=0$. 
Note its magnetic field $B[\varrho] = \Curl \bA[\varrho] = 2\pi \varrho$, so the last term is equivalent to a Pauli spin-coupling term $\frac{\gamma}{2\pi} \int B[|u|^2]|u|^2$, supersymmetric if $\gamma=\pm 2\pi\beta$. 
The variational equation of the functional \eqref{def:Eaf} subject to the constraint $\int_{\R^2} |u|^2 = 1$ is
the Chern--Simons--Schr\"odinger stationary equation:
\begin{equation}\label{eq:CSS-EL1}
\Big[-\!\bigl(\nabla + \im\beta\bA[|u|^2]\bigr)^2 
- 2\beta \nabla^\perp w_0 * \bigl[\beta\bA[|u|^2]|u|^2 + \bJ[u]\bigr] 
+ 2\gamma|u|^2 + V\Big] u = \lambda u,
\end{equation}
where $\bJ[u]:=\frac{\im}{2}(u\nabla\overline{u}-\overline{u}\nabla u)$ and $\lambda = \lambda(u) \in \R$.
We refer to the Appendix \ref{sec:appendix-CSS} for general wellposedness related knowledge about the above equations.
In particular, above a critical coupling, $\gamma > -\gamma_*(\beta)$, and for smooth and trapping $V$, minimizers $u \in H^1(\R^2;\C) \cap L^2_V$ exist and are smooth (and approximable by compact supports).
However, they are typically not unique, due to vortex formation.

The passage from the $N$-body description to the above effective equation appears through the relations 
\begin{equation}\label{eq:infspec-limit}
    N^{-1} \infspec H_N^{R,g} \simeq \inf_{\substack{u\in H^1 \cap L^2_V \\ \norm{u}_2=1}} \cE_{\beta, \gamma, V}[u],
\end{equation}
and, for an approximate ground state $\Psi_N$ of the l.h.s.,
respectively a minimizer $u$ of the r.h.s.,
the convergence of the density
\begin{equation}\label{eq:groundstates-limit}
    N^{-1} \varrho_{\Psi_N} \simeq |u|^2,
    \qquad \text{as $N \to \infty$.}
\end{equation}
Assuming a sufficiently trapping potential, such as $V(\bx) = |\bx|^p$, $p>0$,
the above relations have been made rigorous in \cite{LunRou-15} with $g=0$ in \eqref{def:HNRg} 
and $\gamma =0$ in \eqref{def:Eaf} and for a convergence $R\to 0$ slower than $R=N^{-\eta}$, with the largest such $\eta$ obtained in \cite{Girardot-20}.
The hardest constraints on $\eta$ are set by the approach to the lower bounds, while for an upper bound in \eqref{eq:infspec-limit} by a pure product trial state, it is sufficient that $0< \eta < \infty$.
However, as we show in Theorem~\ref{thm:E}, we can also make rigorous the relation \eqref{eq:infspec-limit} between energies at the level of upper bounds with reasonable trial states for any $g \ge 0$ and for any rate of convergence of $R\to 0$, with some corresponding effective coupling $\gamma \ge 0$. This necessitates a more precise description of the ground state by taking into account correlations between particles over the scale of their average distance $N^{-1/2}$.

Note that, due to the emergence of vortices for $\beta \gtrsim 1$ and thus non-uniqueness of minimizers, \eqref{eq:groundstates-limit} can only hold in an averaged sense, and an effective Thomas-Fermi-type density functional theory for the averaged, macroscopic density profile was derived for $\gamma=0$ and $\beta \gg 1$ in \cite{CorLunRou-17,CorLunRou-18}. The model has also been studied numerically for a wider range of the parameters; see \cite{CorDubLunRou-19,Lundholm-24}.

\subsection{Discussion: results}

We now have the necessary background context to discuss our result.
For $g \neq 0$,
with the heuristic replacement of $\bA^R_j$ by the average field $\bA[\varrho_\Psi]$ in the spin-coupling term in \eqref{def:HNRg}, we may expect the effective scalar coupling $\gamma=g\pi \beta$, which indeed is correct in the sub-critical regime $\omega=0$. 
However, we show that this needs to be revised to $\gamma=2\pi\beta G(s,g)$ in the critical and super-critical regimes $\omega>0$.
Indeed, for strictly pointlike ($\omega=+\infty$) and \emph{ideal} anyons one might expect $\gamma=2\pi\beta$ by a commonly used regularization procedure; see e.g. \cite{ChiSen-92,SenChi-92,Ouvry-94,Ame_95,KimLee-97}.
Similarly, for nonideal (extended and/or spin-coupled) anyons but with a sufficiently fast convergence to pointlike ($\omega \to +\infty$), typically one again expects the ideal ones to emerge, except in a resonant attractive case at $g \to -2$ (possibly corresponding to kreinyons;
see \cite{Grundberg-etal-91,ComMasOuv-95,Mashkevich-96}).
By our Theorem~\ref{thm:E} we thus confirm these predictions for all $g \ge 0$, and also reach the very interesting observation that $\gamma=2\pi\beta$ for all scales $\omega$ exactly if $g=2$, the supersymmetric case.
The only discussion that we have found in the physics literature concerning a possible precise value for $\gamma$ for arbitrary couplings $g$ and scales $s$ is that by Mashkevich in \cite{Mashkevich-96} (see also \cite{AmeBak-95} from the view of scattering amplitudes). However, for comparison we need to translate his result into our setting, and this we will do in Section~\ref{sec:twobody}, where we also compare it to a scale-dependent ``scattering length/energy'' for extended anyons.
First, we close this introductory section by summarizing our understanding so far concerning the conditions for a good low-energy trial state to accurately describe the leading order in \eqref{eq:infspec-limit}.

Following reasoning by Jastrow, Dyson, and others for the dilute interacting Bose gas,
let us argue for our choice of $N$-anyon trial state \eqref{eq:trial-f}-\eqref{def:trial}:
\begin{itemize}
    \item {\bf symmetry:} $\Psi_N^{\mathrm{trial}} \in L^2_\sym$, as it must be for anyons in the bosonic-magnetic representation, and additionally, for small $\alpha$ we expect the more bosonic features of the ground state to dominate.

    \item {\bf short vs. long range correlation:} The Jastrow factor $F$ is meant to deal with the short-range correlation structure in the gas, while long-range correlations are dealt with self-consistently in the Hartree condensate factor $\Phi$. This corresponds to a sufficiently dilute or softly interacting regime in which two-body and three- and many-body interactions are on different length scales --- microscopic vs. mesoscopic.
    
    \item {\bf exclusion:} We expect the exclusion principle to play a role in the energy minimization, i.e. the relative size of the function close to the diagonals $\bDelta$ as $R \to 0$. This because the anyonic flux induces a centrifugal barrier or inverse-square repulsion between each pair of particles (see \cite{Lundholm-16,Lundholm-23}, and Appendix~\ref{sec:appendix-domains}). Our ansatz for $F$ allows both a weakening (for $g < 2$) and strengthening (for $g > 2$) of exclusion compared to that expected of ideal anyons ($g=2$ or $s \to \infty$; see below).

    \item {\bf product:} Due to the above, a product state depending on particle pairs is both convenient and expected to capture the correlations near diagonals, including a lower probability that several particles are close (again, this corresponds to a certain dilute and repulsive regime). 
  
    \item {\bf interpolating:} Our ansatz is able to interpolate continuously between zero exclusion (at $s \to 0$ i.e. only condensate, no Jastrow) and the exclusion expected for ideal anyons (at $s \to \infty$; see \eqref{eq:true-gs} and Appendix~\ref{sec:appendix-domains}).

    \item {\bf radial:} Due to symmetry and near-diagonals interaction expected only at the lowest angular momentum for each pair of particles (2-anyon problem, again assuming diluteness), we choose an ansatz that is radial in the pairwise coordinates (this is standard for this type of problem; this sector of zero pairwise angular momentum is called the ``s-wave'' in the physics literature).

    \item {\bf continuity:} In the repulsive problem we expect continuity in all the variables for the ground state wave function, even in the pointlike limit since $\Psi \in H^1$ always, and is typically even more regular on $\R^{2N} \setminus \bDelta$ and vanishing on $\bDelta$. We may expect continuity in the derivatives as well (except possibly at the diagonals), however in a variational ansatz (not a solution to the corresponding Schr\"odinger equation) this may be seen as only a secondary requirement.
    
    \item {\bf scattering and simplicity:} The two-body correlations $f$ in our trial state ansatz are chosen to minimize the ``scattering energy'' (see Remark~\ref{rem:scattlen}) to leading order in the limit $N \to \infty$.
    Furthermore, their explicit form is at the same time being technically advantageous because the radial derivative yields terms of the same form as the anyonic interaction on distances $R < |\bx| < b$ and is otherwise identically zero.

    \item {\bf scale covariance:} As we see below, scale-covariant properties of the two-body interaction for interacting extended anyons single out the radial bahaviors $r^{\alpha}$ and $r^{-\alpha}$ on the interparticle scale, the latter being suppressed as $R/b \to 0$ for repulsive interactions. Our ansatz formally interpolates between these extremes for $g \in [-2,2]$, although we only consider the more well-behaved ($H^1$, continuous, and stable) case $g \ge 0$ in this work.

    \item {\bf angular momentum:}
    By the heuristic lower bound \eqref{eq:angmom-bound} for ideal anyons, we expect the total angular momentum per particle of an approximate ground state in our scaling to be $L/N \approx -\beta/2$. Indeed, with our Jastrow ansatz $\Psi_N^{\mathrm{trial}}$ and $u$ an approximate minimizer of \eqref{def:Eaf}, based on the works \cite{CorLunRou-17,CorDubLunRou-19,AtaLunNgu-24} we expect the emergence of vortices at a linear order in $|\beta|$ and arranged in a locally homogeneous vortex lattice, each contributing one unit of angular momentum (negative if $\beta>0$; also note the degree of vortices of NLLs between $\frac{\beta}{2}-1$ and $\beta-2$; see Appendix~\ref{sec:appendix-CSS}).
    
\end{itemize}
Note that our choice of $b$ is constrained by the convergence of the norm and density of the wave function and by other technical error terms, and it is not to be viewed as a variational parameter but rather a cutoff to focus our efforts to control the energy using $F$ close to the diagonals (for the critical speed the relevant scale is exponential rather than polynomial which leaves a lot of room to play with here).

We also note that we have in the Hamiltonian \eqref{def:HN} considered a simplest possible microscopic model for the anyon gas which brings out those features we are interested in, namely interpolation of both statistics parameter, size/scale, and hardness/softness of anyons. One can easily imagine more realistic models, such as including other types of two-body interactions (see e.g. \cite{Nguyen-24}), three- or many-body interactions (possibly leading to clustering, exclusion, or emergent Landau levels; see e.g. \cite{Lundholm-16,AtaLunNgu-24,LevLunRou-25}), external magnetic fields (see e.g. \cite{Ataei-25}), as well as dynamics (see e.g. \cite{Ataei-24,GirLee-24}).

\medskip\noindent\textbf{Acknowledgments.} 
D.L. thanks the Department of Mathematics of Politecnico di Milano for its kind hospitality during the spring of 2025 intensive period ``Quantum Mathematics at Polimi''.
Financial support from the Swedish Research Council 
(D.L., grant no.\ 2021-05328, ``Mathematics of anyons and intermediate quantum statistics'') 
is gratefully acknowledged. The author T.G. strongly thanks the Gran Sasso Science Institute's math area for providing its financial support.

\section[\qquad The two-body problem]{The two-body problem}
\label{sec:twobody}

In this section we study the two-body problem and discuss the choice of our trial state \eqref{eq:trial-f}--\eqref{def:trial} by the examination
of three simplified cases. We first examine the case $R=0$ in Section \ref{sec:R=0}.  In Section \ref{sec:R>0} we discuss an approximation for the case $R>0$ and $g\neq 0$ that has been considered in the literature. In Section \ref{sec:ftheta} we compute the scattering energy for a simplified family of trial functions $f_{\theta}(\bx)=|\bx|^{\theta \alpha}$. After having provided these preparatory heuristics, we derive, in Section \ref{sec:trueJastrow}, the exact function \eqref{eq:trial-f} we will use in the $N$-body trial state \eqref{def:trial}.

For the comparison with the $N$-body problem, we switch to the pairwise relative coordinates in the Hamiltonian \eqref{eq:HNR-def},
\begin{align}\label{eq:hamil-pairwise}
    H_N^R = \frac{1}{N}(-\Delta_\bX) 
    + \sum_{1 \le j < k \le N} \Biggl[& \frac{1}{N} \left(-\im(\nabla_{\bx_j} - \nabla_{\bx_k}) + \alpha (\bA^R_j-\bA^R_k)\right)^2 \\ \nonumber
    &+ \frac{1}{N-1}\bigl( V(\bx_j) + V(\bx_k) \bigr)
    + W(\bx_j-\bx_k)
    \Biggr],
\end{align}
where $\bX := \frac{1}{N}\sum_j \bx_j$ is the total center of mass, 
$\br_{jk} := \bx_j-\bx_k$, and
$$
    \nabla_{\bx_j} - \nabla_{\bx_k} = 2\nabla_{\br_{jk}}, \qquad
    \bA_j^R - \bA_k^R = 2 \frac{(\bx_j-\bx_k)^\perp}{|\bx_j-\bx_k|_R^2} + \sum_{l \neq j,k} \left[ \frac{(\bx_j-\bx_l)^\perp}{|\bx_j-\bx_l|_R^2} - \frac{(\bx_k-\bx_l)^\perp}{|\bx_k-\bx_l|_R^2} \right].
$$
Further, if $V(\bx) = c|\bx|^2$, then it also admits a similar separation of variables:
$$
    \frac{1}{N-1}\sum_{j<k} \left( V(\bx_j) + V(\bx_k) \right) 
    = \sum_j V(\bx_j) 
    = N c|\bX|^2 + \frac{1}{N} \sum_{j<k} c|\br_{jk}|^2.
$$
The heuristics for the pair contribution to the $N$-body energy is based on taking $N=2$ in \eqref{eq:hamil-pairwise} and multiplying by the number of pairs $\binom{N}{2}$.
Further, the influence of any other particles $l \neq j,k$ is ignored, which assumes a sufficiently dilute gas
(it may be different in higher density regimes).

\subsection{The motivation for the Jastrow pair correlation}

    In the regime of small radius $R$ compared to the typical interparticle distance $N^{-1/2}$, pair interactions become more likely than any higher-order interaction and thus the two-body problem becomes of main interest to reach the next order in the energy. The idea is that the $N$-body energy resulting of the sum of each pair interaction must be of order of the energy of a two-body interaction that multiplies the number of pairs. In the following, we develop this heuristic in more detail and compare it with the mathematical and physics literature.
    
    The two-body problem associated with the Hamiltonian \eqref{def:HNRg} has been well studied in the mathematics \cite{Oddis-20,CorOdd-18} and physics literature \cite{Grundberg-etal-91,ManTar-91,BorSor-92,Mashkevich-96,KimLee-97}, though with varying degree of generality and detail.
    The two-anyons Hamiltonian is defined on $L^2(\R^4)$ as
    \begin{align}
       H^{R,g}_2:= \sum_{j=1}^2\left[ (-\im \nabla_j +\alpha \bA_j^R)^2+\frac{g\alpha}{2}B^R_j \right],
    \end{align}
    that we can rewrite using the center of mass coordinates,
 denoting 
\begin{align*}
    \mathbf{r}:=\bx_1 -\bx_2
    \quad\text{and}\quad
    \mathbf{X}:=\frac{\bx_1 +\bx_2}{2},
\end{align*} 
to obtain
\begin{align*}
      -\Delta_1 -\Delta_2 +2\alpha^2\frac{|\mathbf{r}|^2}{|\mathbf{r}|^4_R}-2\im \alpha \frac{\mathbf{r}^{\perp}}{|\mathbf{r}|^2_R}\cdot (\nabla_1 -\nabla_2)
        = -\frac{\Delta_\mathbf{X}}{2} +2\left(-\Delta_{\mathbf{r}} +\alpha^2\frac{|\mathbf{r}|^2}{|\mathbf{r}|^4_R}-2\im \alpha \frac{\mathbf{r}^{\perp}}{|\mathbf{r}|^2_R}\cdot \nabla_{\br} \right),
\end{align*}
using that $B^R_1(\sx)=B^R_2(\sx)$ and $\bA^R_1(\sx)=-\bA^R_2(\sx)$.
Note that the cross term in $\mathbf{r}^{\perp}$ of the above operator equation vanishes on radial functions\footnote{One could consider a phase dependence as well, however for only two particles and $\alpha \in [0,1]$ this will only increase the energy, due to Poincar\'e and Hardy inequalities; see Appendix~\ref{sec:appendix-domains}.}.
Thus, a toy model in our case of extended and spin-orbit interacting anyons, is then the following energy functional:
\begin{align}\label{eq:cE-two}
    \cE_2[f]  
    := \int_{B(0,\ell)} \left( |\nabla f|^2 + \alpha^2 \frac{|\br|^2}{|\br|_R^4}|f|^2 + \alpha g \frac{\1_{B(0,R)} }{R^{2}}|f|^2 \right),
\end{align}
considered on radial functions $f \in H^1(B(0,\ell))$.

Canonically, the length $\ell > R$ introduced here is taken to be the local interparticle spacing, which depends on the density $\varrho = \varrho(\bx)\in L^1(\R^2)$ 
as $\ell = \ell(\bx) \sim \varrho(\bx)^{-1/2}$. 
More precisely, it can be defined such that exactly one of the other particles is found within that radius, i.e.
$\pi\ell^2\varrho \simeq 1$.
This choice of scale is required to justify the energy heuristics for the homogeneous gas, since then approximately $N/2$ such areas will cover the entire area of the gas.
However, in our final scaling of the parameters with $N$, it is only the order of the length scale compared to $R$ that is important, and we can and will typically replace this variable local scale with a uniform scale $b=b_N$, $R \ll b \ll \ell$, not too small (its precise size will be constrained by the error terms arising in the many-body problem). For the two-body heuristics it is also possible to use the reference scale $\ell = 1$, as in Remark~\ref{rem:scattlen}.

Thus, in order to preserve the mesoscopic structure at the scale of 
$\ell$,
we want to minimize this energy for $R \ll \ell \lesssim 1$, 
where $f \simeq 1$ on the inter-particle scale $r \sim \ell$,
although for $\ell$ large enough (compared to both $R$ and $g$) we may also consider a corresponding Neumann problem with $\partial_r f|_{r=\ell} = 0$. The procedure can be compared to the definition of the scattering length for the two-body boson problem (see \cite[Theorem C.1]{LieSeiSolYng-05}) except that the long-range potential $\bx_R^{-2}$ drastically changes the profile of the solutions compared to the case of a compactly supported interaction, and therefore demands some care 
(we cannot take $\ell \to \infty$ in \eqref{eq:cE-two}, for example). 

\subsubsection{The case $R=0$}\label{sec:R=0}

In the simplified case $R=0$ (and $g=0$), the relative $2$-body problem \eqref{eq:cE-two} becomes
\begin{align}
    \cE_2[f] = \int_{B(0,\ell)} \left(|\nabla f|^2 +\alpha^2 \frac{|f|^2}{|\bx|^2}\right) \, \dd \bx,\label{def:E2R0}
\end{align}
over $f \in H^1(B(0,\ell))$ such that $\frac{f}{|\bx|} \in L^2(B(0,\ell))$ and $f =1$ on $\partial B(0,\ell)$. The above expression can easily be explicitly minimized. By symmetrization and taking the variation of the energy functional, we obtain that the minimizing function satisfies
\begin{align*}
  -\Delta f + \frac{\alpha^2}{r^2} f = 0, \textup{ in } B(0,\ell), \quad f=1, \textup{ in }  \partial B(0,\ell).
\end{align*}
By solving the ODE, we get 
\begin{align*}
    f(r) = a r^{\alpha} + b r^{-\alpha}, \quad r>0,
\end{align*}
for some $a,b \in \R.$ Since $f \in H^1(B(0,\ell))$ and $f=1$ in $\partial B(0,\ell)$, we obtain that the only possible solution is
\begin{align*}
     f(r) = \left( \frac{r}{\ell}\right)^{\alpha}, \quad 0 \leq r \leq \ell.
\end{align*}
Hence, the Jastrow factor based on $r^{\alpha}$ minimizes the two-body scattering energy, with
$$
    E_2 = \cE_2[f] = 4\pi \alpha^2 \int_0^\ell r^{2\alpha-1} dr
    = 2\pi\alpha \ell^{2\alpha}
    \qquad \Rightarrow \qquad
    (N-1)E_2 \stackrel{N \to \infty}{\simeq} 2\pi\beta.
$$
This indeed corresponds to the Jastrow factor in the exact solutions \eqref{eq:true-gs} for ideal anyons.
On the other hand, when $R$ becomes very large compared to $\ell$, we know from \cite{LunRou-15} that $f=1$ is enough to describe the energy at leading order,
and therefore we need to revise this ansatz for extended anyons.

\subsubsection{The case $R>0$ and $g \neq 0$}\label{sec:R>0}

For $R>0$ and $g \neq 0$, the two-body energy problem \eqref{eq:cE-two} does not simplify immediately, and it is then helpful to use some approximations to facilitate the corresponding minimization problem
\begin{equation}\label{eq:scattenergy}
       \left(-\Delta + \alpha^2 \bx_R^{-2} + \alpha g R^{-2} \1_{B(0,R)}\right)f = \lambda f,
    \qquad f|_{r=\ell} = 1,
\end{equation}
without resorting to cumbersome special functions.

A version of the above eigenvalue problem appears in \cite[Equation (2)]{Mashkevich-96}, where, instead of restricting to a finite domain,
a harmonic regularization 
$V(\bx) = \frac{1}{2}m\omega_\ell^2|\bx|^2$ 
is introduced as a reference, 
and the two-particle contribution to the 
energy is obtained via perturbation in $\alpha$. Thus, in that work, the trial state is taken to be on the form $f\psi_0$, with $\psi_0$ the ground state of the harmonic oscillator at $\alpha=0$, i.e. the gaussian
$\psi_0(r) \propto e^{-\frac{1}{4}m\omega_\ell r^2}$ with energy $\omega_\ell$.
The radial Hamiltonian is defined as
\begin{equation}\label{eq:H-Mashkevich}
    \cH = \frac{1}{m}\left(
        -\frac{d^2}{dr^2} - \frac{1}{r}\frac{d}{dr} + \frac{[L-\alpha\eps(r)]^2}{r^2} - \frac{2\sigma\alpha \eps'(r)}{r}
    \right) + \frac{1}{4}m\omega_{\ell}^2 r^2,
\end{equation}
allowing for a more general flux profile by the function $\eps(r)$.
In our notational system this is equivalent to an angular momentum $L=0$,  a mass $m=1/2$, the function $\eps(r) = r^2/|r|_R^2$, and a trapping potential with the coupling constant (frequency) $\omega_{\ell}$. 
This quantity defines the size of the relative disk $B(0,\ell)$ of the two-body interaction via the relation $\omega_{\ell}\ell^2 \sim 1$.  To be complete in our translation, we also need to multiply the spin $\sigma=\mp 1/2$ supersymmetric with $g/2$ as mentioned in \cite[footnote 17]{Mashkevich-96}.

The author observes that, in the supersymmetric case $g=\pm 2$, we can exactly solve
$\cH (f\psi_0) = \lambda (f\psi_0)$
for the eigenvalue $\lambda = \omega_\ell + \Delta E$ 
and the constraint $\left|\left(\partial_r f\right)(\ell)\right| \leq \alpha (R/\ell)$. 
The solution, with $\Delta E = \pm \alpha \omega_\ell$ and which we denote $f_{\rm susy}$, is
\begin{align}\label{eq:2-body-solution-supersym}
\begin{cases}
f_{\rm susy}(r)=e^{\pm\frac{\alpha}{2}(\frac{r^2}{R^2}-1)}\quad \text{for}\quad r\leq R,\\
f_{\rm susy}(r) = \left(\frac{r}{R}\right)^{\pm \alpha}\quad \text{for}\quad r> R.
\end{cases} 
\end{align}
He then suggests that the general case with a uniform magnetic flux of radius $R>0$ and spin-orbit coupling $g$ should be as in \cite[Equation (21)]{Mashkevich-96}, where the solutions 
for small $\alpha$
are approximated by
\begin{align}\label{eq:2-body-solution-g}
\begin{cases}
f_g(r)=1+\frac{g\alpha}{4}\frac{r^2}{R^2} \quad \text{for}\quad r\leq R,\\
f_g(r) = \frac{1+g(2+\alpha)/4}{2}\left(\frac{r}{R}\right)^{ \alpha} + \frac{1-g(2-\alpha)/4}{2}\left(\frac{r}{R}\right)^{ -\alpha} \quad \text{for}\quad r> R,
\end{cases} 
\end{align}
with the effective boundary condition at $r=R$:
$$
    \frac{f_g'(R)}{f_g(R)} = \frac{g\alpha}{2R}.
$$
According to \cite[Equation (22)]{Mashkevich-96}, the first order perturbation of the energy for small $\alpha$, obtained by using the solution \eqref{eq:2-body-solution-g}, 
is 
$$
    \Delta E = \alpha\omega_\ell \frac{1+g/2-(1-g/2)q^\alpha}{1+g/2+(1-g/2)q^\alpha},
$$
where $q \sim (R/\ell)^2$ and hence
$q^\alpha \sim e^{-2\beta\omega}$. This last expression agrees with our function $G(2\beta\omega,g)$ of \eqref{def:G} and can be used to predict our result by plugging $\omega_{\ell}\simeq 2\ell^{-2}$, $\alpha \simeq \beta (N-1)^{-1}$ and $\pi\ell^2\varrho \simeq 1$
in the limit $N \to \infty$.

In the next paragraph we compare this result to the two-body energy obtained using a simpler specific ansatz $f_{\theta} \propto r^{\alpha\theta}$.

\subsubsection{A simplified repulsive ansatz for $R>0$}\label{sec:ftheta}

In the limit as $R/\ell \to 0$, which happens non-uniformly due to the varying density in the problem, we suggest that a reasonable ansatz 
for $f$ is a rescaling of the function
\begin{align}
    f_\theta(r) := (r/\ell)^{\alpha\theta}
    \quad\text{of norm}\quad 
    \int_{B(0,\ell)}|f_{\theta}|^2=\frac{\pi}{1+\alpha\theta}\ell^2,
\end{align}
where the variational parameter $\theta \in [-1,1]$ can capture an interpolation between the two extremes of kreinyons $f\propto r^{-\alpha}$ and ideal anyons $f\propto r^{\alpha}$. However, as discussed above, due to additional complications like stability, we will here only consider the repulsive case $\theta \in [0,1]$ (and even extend to $\theta \in \R_+$).

If we compute the scattering energy of $f_{\theta}$ for $\theta \in \R_+$, by plugging it in \eqref{eq:cE-two}, we get
\begin{align}\label{eq:cE-twobis}
    \cE_2[f_\theta]
    &= \pi\alpha \left(
        \theta + \frac{1}{\theta}[1-(R/\ell)^{2\alpha\theta}]
        + g
        (R/\ell)^{2\alpha\theta}\right)\\
    &\simeq \pi\alpha \tilde{G}(\theta,2\beta\omega,g)
        \quad \text{as} \ N \to \infty,
        \nonumber
\end{align}
where we recover a similar 
$\theta$-dependent version of $G$ defined in \eqref{def:G}:
$$
    \tilde{G}(\theta,s,g)
    :=\theta + \frac{1}{\theta}(1-e^{-s\theta}) + ge^{-s\theta}.
$$
Similarly, the Neumann energy is
\begin{align}\label{eq:cE-twobis-Neumann}
    \frac{\cE_2[f_\theta]}{\norm{f_{\theta}}^2}
    &= \frac{\alpha}{\ell^2} (1+\alpha\theta) \left(
        \theta + \frac{1}{\theta}[1-(R/\ell)^{2\alpha\theta}]
        + g
        (R/\ell)^{2\alpha\theta}\right)\\
    &\simeq \alpha\ell^{-2} \tilde{G}(\theta,2\beta\omega,g)
        \quad \text{as} \ N \to \infty.
        \nonumber
\end{align}
Considering a factor 2 coming from the relative mass, that there are $N(N-1)/2$ pairs, and that we have chosen $\ell$ such that the expected number of particles is one, $\pi \ell^2 \varrho \sim 1$,
thus, heuristically, 
in the limit $R/\ell \sim e^{-\omega N}$ and $\alpha=\beta(N-1)^{-1}$ as $N \to \infty$, the energy per particle is precisely
$$
    E_N/N \simeq \pi \beta \tilde{G}(\theta,2\beta\omega,g) \varrho.
$$

We can then check that,
by optimizing over $\theta$,
we recover some of the critical regimes of Theorem \ref{thm:E}:
\begin{align}\label{eq:Gcases}
 \inf_{\theta\geq 0} \tilde{G}(\theta ,2\beta\omega ,g)=\begin{cases}
 g\quad\text{when}\quad \omega =0,\\
2\quad\text{when}\quad \omega =+\infty,\\
2\beta\omega\quad\text{when}\quad g =0\quad \text{and}\quad \beta\omega \ll 1,
\end{cases} 
\end{align}
which is in agreement with \eqref{eq:DeltaEcases} 
in the considered regimes. 
This makes of the $ f_\theta(r) \propto r^{\alpha\theta}$ ansatz a choice fairly simple to compute, with the property to produce the conjectured optimal energy in these specific cases. Further, although we will not go into these details here, we claim that it is actually possible to carry this ansatz to the many-body problem and prove the limit
\begin{equation}
    \label{eq:effectivedensityfunction1}
       \lim_{N \to \infty}     \frac{\bra{ \Psi_N^{\mathrm{trial}}}\ket{H_N  \Psi_N^{\mathrm{trial}}}}{N\norm{ \Psi_N^{\mathrm{trial}}}^2}=\int_{\mathbb{R}^{2}}\left|(-\im\nabla+\beta \bA[|u|^{2}])u\right|^{2}+\int_{\R^2}V|u|^{2}+ \pi \beta\, \tilde{G}( \theta , 2\beta\omega ,g)\int_{\R^2}|u|^4,
\end{equation}
by replacing our $f$ of \eqref{eq:trial-f} by $f_{\theta}$ in the trial state \eqref{def:trial}.
However, this approach yields a slightly too large energy for intermediate values of $g$ and scales $\omega$ (for example in the supersymmetric case $g=2$), and therefore we will proceed to revise the ansatz.

\subsection{Our more precise Jastrow pair correlation}\label{sec:trueJastrow}

In this part, we rigorously construct the precise Jastrow pair correlation function $f$ used in our trial state \eqref{def:trial} by an approximation of the two-body problem. Note that we assume throughout that
\begin{align}
\label{eq:conditionong}
  0\leq g \quad\text{and} \quad 0<R<b.
\end{align}

The idea is to study the minimizer of 
\begin{align}
\label{eq:twobodyenergy}
 \cE_2[f] := \int_{B(0,b)} \left( |\nabla f|^2 + \alpha^2 \frac{|f|^2}{r^2} \1_{A(R,b)} +\alpha^2 \frac{r^2 |f|^2}{R^4} \1_{B(0,R)} +  g \alpha \frac{ |f|^2}{R^2}\1_{B(0,R)} \right),
\end{align}
over all the radially symmetric functions $f \in H^1(B(0,b))$ with $f=1$ on $\partial B(0,b)$.

\begin{proposition}[Two-body scattering energy]
    Let $E_2$ be the minimum/infimum of \eqref{eq:twobodyenergy} over all the radially symmetric functions $f \in H^1(B(0,b))$ with $f=1$ on $\partial B(0,b)$. Then, 
\begin{equation}
\label{eq:minimumoftwobodyenergy}
2\pi  \alpha \frac{1+\frac{g}{2} - (1-\frac{g}{2}) \left(\frac{R}{b}\right)^{2\alpha}}{1+\frac{g}{2} + \left(1-\frac{g}{2}\right) \left(\frac{R}{b}\right)^{2\alpha}} - 2\pi \alpha^2 g^2   
\leq E_2 \leq 
2\pi  \alpha \frac{1+\frac{g}{2} - (1-\frac{g}{2}) \left(\frac{R}{b}\right)^{2\alpha}}{1+\frac{g}{2} + \left(1-\frac{g}{2}\right) \left(\frac{R}{b}\right)^{2\alpha}} + \frac{\pi}{2} \alpha^2.
\end{equation}
Moreover, 
\begin{align*}
  2\pi  \alpha \frac{1+\frac{g}{2} - (1-\frac{g}{2}) \left(\frac{R}{b}\right)^{2\alpha}}{1+\frac{g}{2} + \left(1-\frac{g}{2}\right) \left(\frac{R}{b}\right)^{2\alpha}} \leq \cE_2[f] \leq 2\pi  \alpha \frac{1+\frac{g}{2} - (1-\frac{g}{2}) \left(\frac{R}{b}\right)^{2\alpha}}{1+\frac{g}{2} + \left(1-\frac{g}{2}\right) \left(\frac{R}{b}\right)^{2\alpha}} + \frac{\pi}{2} \alpha^2 ,
\end{align*}
where $f$ is the Jastrow function in \eqref{eq:trial-f}.
\end{proposition}
\begin{proof}

By $\cE_2[\min(|f|,1)] \leq \cE_2[f]$ and using the direct method in the calculus of variation, there exists a minimizer $f$ for $\cE_2[f]$ which is radially symmetric and satisfies $0 \leq  f \leq 1$ in $B(0,b)$. Since $|f| \leq 1$, we get 
\begin{align}
\label{eq:radiusRenergyerror}
  \int_{B(0,R)}  \alpha^2 \frac{r^2 |f|^2}{R^4} \leq \frac{\pi}{2} \alpha^2.
\end{align}
Hence, 
\begin{align*}
  \left| \inf \cE_2[f] - \inf \int_{B(0,b)} \left( |\nabla f|^2 + \alpha^2 \frac{|f|^2}{r^2} \1_{A(R,b)} +  g \alpha \frac{ |f|^2}{R^2}\1_{B(0,R)} \right) \right|\leq \frac{\pi}{2} \alpha^2,
\end{align*}
where both infima are taken over $f \in H^1(B(0,b))$ with $f=1$ on $\partial B(0,b)$. The idea is that, with $\alpha=\beta(N-1)^{-1} \ll 1$, the error term of above will not affect the main order of the energy which, as we will see in \eqref{eq:averageenergytwobody}, is of order $\alpha$. Now, define
\begin{align*}
    \cE^1_2[f] := \int_{B(0,b)} \left( |\nabla f|^2 + \alpha^2 \frac{|f|^2}{r^2} \1_{A(R,b)} +  g \alpha \frac{ |f|^2}{R^2}\1_{B(0,R)} \right).
\end{align*}
By doing a variation of the energy functional and taking $\min (|f|,1)$, we obtain that there exists a radially symmetric minimizer $f$ of $\cE^1_2[f]$, which satisfies $0 \leq f \leq 1$ and
\begin{align*}
     \left(-\Delta +  \alpha^2 |\bx|^{-2} \1_{B(b,R)}  + \alpha g R^{-2} \1_{B(0,R)}\right)f =0, \quad \textup{weakly in } B(0,b).
\end{align*}
Hence,
\begin{align*}
     \left(-\Delta+ \alpha^2 |\bx|^{-2} \right)f =0, \quad \textup{weakly in } A(R,b).\\
      \left(-\Delta  + \alpha g R^{-2}\right)f =0,\quad \textup{weakly in } B(0,R).
\end{align*}
Now, we claim that up to order of $\alpha^2$, we can consider $f$ to be a constant in $B(0,R)$. To see this, note that 
\begin{align*}
     \left|\frac{1}{r} \partial_r \left(r \partial_r f\right) \right|= \alpha g R^{-2} f \leq \alpha g R^{-2}, \quad \textup{in } B(0,R),
\end{align*}
where we used $|f| \leq 1.$
Hence,
\begin{align*}
    |\partial_r f| \leq \frac{ \alpha g}{2 R^2} r, \quad \textup{in } B(0,R).
\end{align*}
This implies that 
\begin{align*}
    |f(r)-f(R)| \leq \alpha g,
\end{align*}
and 
\begin{align*}
     \cE^1_2[f] \geq -2 \pi \alpha^2 g^2 + \int_{B(0,b)} \left( |\nabla f|^2 + \alpha^2 \frac{|f|^2}{r^2} \1_{A(R,b)} +  g \alpha \frac{ |f(R)|^2}{R^2}\1_{B(0,R)} \right). 
\end{align*}
This completes the proof of the claim. In conclusion, we need to study the minimizer $f$ of 
\begin{align*}
 \cE^2_2[f] := \int_{B(0,b)} \left( |\nabla f|^2 + \alpha^2 \frac{|f|^2}{r^2} \1_{A(R,b)} +  g \alpha \frac{ |f(R)|^2}{R^2}\1_{B(0,R)} \right),
\end{align*}
which satisfies 
\begin{align*}
     \left(-\Delta+ \alpha^2 |\bx|^{-2} \right)f =0, \quad \textup{weakly in } A(R,b).
\end{align*}
Since $|\bx|^{\pm \alpha}$ generate all the solutions to the above equation, we get 
\begin{align*}
    f(r) = \lambda_1 r^{\alpha} + \lambda_2  r^{-\alpha}, \quad \textup{in }R \leq r \leq b. 
\end{align*}
Since $f$ is a constant on $B(0,R)$, we derive that
\begin{equation}
    f(r) = \begin{cases}
         \lambda_1 R^{\alpha} + \lambda_2  R^{-\alpha}, \quad &\textup{for } 0 \leq r\leq R,\\
         \lambda_1 r^{\alpha} + \lambda_2  r^{-\alpha}, \quad &\textup{for }R \leq r \leq b,\\
         1 , \quad & \textup{for } r \geq b,\label{def:fr}
    \end{cases}
\end{equation}
where 
\begin{align*}
   1 = f(b) = \lambda_1 b^{\alpha} + \lambda_2 b^{-\alpha}.
\end{align*}
By computing the energy functional, we get
(note the cancellation of logarithmically divergent cross terms!)
\begin{align*}
    \cE^2_2[f] = 2 \pi \alpha \left( \lambda^2_1 \left( b^{2\alpha} - R^{2\alpha} \right) + \lambda^2_2  \left( R^{-2\alpha} - b^{-2\alpha} \right)  + \frac{g}{2}  \left(\lambda_1 R^{\alpha} + \lambda_2  R^{-\alpha} \right)^2\right).
\end{align*}
To minimize $\cE^2_2[f]$, we solve $\lambda_1$ such that 
\begin{align*}
    \frac{\partial}{\partial \lambda_1}  \cE^2_2[f] = 0
\end{align*}
Note that 
\begin{align*}
    \lambda_2 = b^{\alpha} \left(1- \lambda_1 b^{\alpha}\right), \quad  \frac{\partial}{\partial \lambda_1} \lambda_2 = -b^{2\alpha}.
\end{align*}
Hence,
\begin{align*}
    &2 \lambda_1 \left( b^{2\alpha} - R^{2\alpha} \right) -2 b^{\alpha} \left(1- \lambda_1 b^{\alpha}\right) b^{2\alpha}  \left( R^{-2\alpha} - b^{-2\alpha} \right)  \\&+ g  \left(\lambda_1 R^{\alpha} + b^{\alpha} \left(1- \lambda_1 b^{\alpha}\right)  R^{-\alpha} \right) \left( R^{\alpha} - b^{2\alpha} R^{-\alpha} \right) =0. 
\end{align*}
In conclusion, 
\begin{align}
\label{eq:thecoeffioff}
    \lambda_1 = \frac{(2+g) b^{-\alpha}}{2 \left(1+ \left(\frac{R}{b} \right)^{2\alpha}\right)  + g \left(1- \left(\frac{R}{b} \right)^{2\alpha} \right)
   }, \quad \lambda_2 = \frac{(2-g) b^{-\alpha} R^{2\alpha}}{2 \left(1+ \left(\frac{R}{b} \right)^{2\alpha}\right)  + g \left(1- \left(\frac{R}{b} \right)^{2\alpha} \right)}.
\end{align}
By setting $\lambda_1,\lambda_2$ in $f$, we get 
\eqref{eq:trial-f}, i.e.
\begin{equation}\label{def:fr2}
f(r) = \begin{cases}
 \frac{4 \left(\frac{R}{b}\right)^{\alpha} }{2\left(1+ \left(\frac{R}{b}\right)^{2\alpha}\right) + g\left(1-\left(\frac{R}{b}\right)^{2\alpha}\right)} , \quad &\textup{in } 0 \leq r \leq R,\\\\
\frac{(2+g) \left(\frac{r}{b}\right)^{\alpha} + (2-g) \left(\frac{R}{b}\right)^{2\alpha} \left(\frac{b}{r}\right)^{\alpha} }{2\left(1+ \left(\frac{R}{b}\right)^{2\alpha}\right) + g\left(1-\left(\frac{R}{b}\right)^{2\alpha}\right)}, \quad &\textup{in } R \leq r \leq b,\\\\
1, \quad &\textup{in } r \geq b.
\end{cases}
\end{equation}
Moreover, 
\begin{equation}
\label{eq:averageenergytwobody}
\begin{aligned}
    \cE^2_2[f] &= 2 \pi \alpha \Biggl[  \left( \frac{2+g}{2\left(1+\left(\frac{R}{b}\right)^{2\alpha}\right) + g\left(1-\left(\frac{R}{b}\right)^{2\alpha}\right)}\right)^2 \left(1-\left(\frac{R}{b}\right)^{2\alpha}\right) \\ &\quad + \left ( \frac{2-g}{2\left(1+\left(\frac{R}{b}\right)^{2\alpha}\right) + g\left(1-\left(\frac{R}{b}\right)^{2\alpha}\right)} \right )^2 \left(\frac{R}{b}\right)^{2\alpha} \left(1-\left(\frac{R}{b}\right)^{2\alpha}\right) 
   \\& \quad + \frac{g}{2} \left( \frac{4}{2\left(1+\left(\frac{R}{b}\right)^{2\alpha}\right) + g\left(1-\left(\frac{R}{b}\right)^{2\alpha}\right)} \right)^2 \left(\frac{R}{b}\right)^{2\alpha} \Biggr] \\
   &= 2\pi  \alpha\,  \frac{(g+2)^2 - (g-2)^2 \left(\frac{R}{b}\right)^{4\alpha}}{\left((g+2) - (g-2) \left(\frac{R}{b}\right)^{2\alpha}\right)^2}\\
   &= 2\pi  \alpha\,  \frac{g+2 + (g-2) \left(\frac{R}{b}\right)^{2\alpha}}{g+2 - (g-2) \left(\frac{R}{b}\right)^{2\alpha}}
   = 2\pi  \alpha \frac{1+\frac{g}{2} - (1-\frac{g}{2}) \left(\frac{R}{b}\right)^{2\alpha}}{1+\frac{g}{2} + \left(1-\frac{g}{2}\right) \left(\frac{R}{b}\right)^{2\alpha}}.
   \end{aligned}
\end{equation}
This completes the proof of the left-hand side in \eqref{eq:minimumoftwobodyenergy}. To prove the right-hand side, by \eqref{eq:radiusRenergyerror}, we have 
\begin{align*}
    \mathcal{E}_2 \leq \inf \cE^2_2[f]  + \frac{\pi}{2 } \alpha^2 = 2\pi  \alpha \frac{1+\frac{g}{2} - (1-\frac{g}{2}) \left(\frac{R}{b}\right)^{2\alpha}}{1+\frac{g}{2} + \left(1-\frac{g}{2}\right) \left(\frac{R}{b}\right)^{2\alpha}} + \frac{\pi}{2 } \alpha^2,
\end{align*}
where the infimum is taken over all the radially symmetric $f \in H^1(B(0,b))$ where $f=1$ on $\partial B(0,b)$.
\end{proof}

\begin{proposition}
\label{prop:obviousboundonJastrow}
    The function $f$ in \eqref{eq:trial-f} is radially monotone increasing and satisfies 
\begin{align*}
   0 \leq f \leq 1, \quad \textup{in } \R^2.
\end{align*}    
\end{proposition}
\begin{proof}
For $\bx \in \R^2 \setminus B(0,b)$, by the definition of $f$ in \eqref{eq:trial-f} we get $ f(\bx) =1$. If $ \bx \in B(0,R)$, then $f(\bx) \geq 0$, and, by our assumption \eqref{eq:conditionong}, we have $g \geq 0$, and
    \begin{align*}
         f(\bx) &= \frac{4 \left(\frac{R}{b}\right)^{\alpha} }{2\left(1+ \left(\frac{R}{b}\right)^{2\alpha}\right) + g\left(1-\left(\frac{R}{b}\right)^{2\alpha}\right)}
         \leq  \frac{4 \left(\frac{R}{b}\right)^{\alpha} }{2\left(1+ \left(\frac{R}{b}\right)^{2\alpha}\right)} \leq 1.
    \end{align*}
     Now, by regularity it is enough to study the critical points of $f$ in $A(R,b)$, which gives the minimum and the maximum of $f$ in $A(R,b)$. Thus, let $\partial_r f(r_m)=0$ for $R<r_m <b$. Then,
\begin{align*}
  (2+g)  \frac{r_m^{\alpha-1}}{b^{\alpha}} - (2-g) \left(\frac{R}{b}\right)^{2\alpha} \frac{b^{\alpha}}{r_m^{\alpha+1}} =0.
\end{align*}
Hence, $ g \leq 2$ and 
\begin{align*}
    r_m = \left(\frac{2-g}{2+g} \right)^{\frac{1}{2\alpha}} R.
\end{align*}
Now, since also $g \geq 0$, we derive a contradiction as $r_m \leq R$ and we assumed that $R<r_m<b$. Therefore, $f$ is monotone increasing on $A(R,b)$.
\end{proof}

\section[ \qquad The many-body problem]{The many-body problem}\label{sec:manybody}

In this section, we will proceed to compute the $N$-body energy 
\begin{equation}
\frac{\bra{ \Psi_N^{\mathrm{trial}}}\ket{H_N  \Psi_N^{\mathrm{trial}}}}{\norm{ \Psi_N^{\mathrm{trial}}}^2},
\end{equation}
where $H_N$ and $\Psi_N^{\mathrm{trial}}$ are respectively defined in \eqref{def:HN} and \eqref{eq:trial-f}--\eqref{def:trial}. We first extract the main properties of the trial state needed for the calculation, then compute each term of $H_N$ one by one, sorting the results into the Lemmas \ref{lem:K}, for the kinetic term \ref{lem:V} for the potential, \ref{lem:W}, \ref{lem:Sdiag} for the two-body singular term, \ref{lem:three body} for the three-body and \ref{lem:J} for the mixed-two-body term.

\subsection{Jastrow trial state}

In this section, we study some elementary properties of the Jastrow trial state. Here, we use the convention that
$h_1 \lesssim h_2$ for parameters $h_1,h_2$ if there exists a universal constant $C>0$ such that $h_1 \leq C h_2$ and $h_1 = h_2 +O(h_3)$ for the parameters $h_1,h_2,h_3$ if 
\begin{align*}
    |h_1-h_2| \lesssim h_3.
\end{align*}
Further, we will use the shorthand notation $\norm{\cdot}_p$ to denote the norm in $L^p(\R^2)$ and $\norm{\cdot}_{p,w}$ to denote the weak $L^p(\R^2)$-norm, for every $p \geq 1$. Note that $\frac{1}{|\bx|}$ is in weak $L^2(\R^2)$.

We recall that the trial state is defined as 
\begin{equation}
\Psi^{\mathrm{trial}}= \PsiTrial = F \Phi \quad \text{where}\quad \Phi(\sx)=\prod_{j=1}^N u(\bx_j) \quad \text{and}\quad F(\sx)=\prod_{1\leq j<k\leq N}f(\bx_j -\bx_k)
\end{equation}
for every $\sx = (\bx_1,\cdot \cdot \cdot, \bx_N) \in \R^{2N}$, with $\Phi$ describing the condensate and $F$ implementing the effect of correlations on the scale $b>R$, where $f$ has been defined in \eqref{eq:trial-f}.
In the entire work, by regularization and taking $N$ large enough, we assume that $0 \leq \alpha < \frac{1}{4}$ and
\begin{align*}
    u \in C^{\infty}_c(\R^2), \quad \int_{\R^2} |u|^2 =1.
\end{align*}
By Proposition \ref{prop:obviousboundonJastrow}, $0 \leq F\leq 1$. In order to replace $F^2$ by $1$ in the calculations to come we establish the following lemma providing our main tool to control $F^2-1$. 

\begin{lemma}\label{lem:propjastrow}
Let $F$ be defined as in \eqref{eq:trial-f}. Then,
\begin{equation}\label{ine:jastrow211}
\left|F^2(\sx) -1\right|\leq \sum_{1\leq k<m\leq N}(1-f(\bx_k -\bx_m)^2)\1_{B(0,b)}(\bx_k -\bx_m).
\end{equation}
for every $\sx \in \R^{2N}$, and
\begin{equation}
\label{ine:bounds12}
\begin{aligned}
\norm{1-f^2}_1& \lesssim (1+g^2) b^2 \alpha,\\
\norm{1-f^2}^2_2&\lesssim (1+g^2) b^2 \alpha.
\end{aligned}
\end{equation}
\end{lemma}
\begin{proof} The first inequality follows from
\begin{equation*}
1\geq \prod_{1\leq k<m\leq N}f(\bx_k-\bx_m)^2\geq 1-\sum_{1\leq k<m\leq N}(1-f(\bx_k -\bx_m)^2),
\end{equation*}
where we used Lemma \ref{lem:productinequality} on the right-hand side inequality and Proposition \ref{prop:obviousboundonJastrow} on the left-hand side inequality. Then,
\begin{equation}\label{ine:jastrow21}
\left|F^2 -1\right|\leq \sum_{1\leq k<m\leq N}(1-f(\bx_k -\bx_m)^2)=\sum_{1\leq k<m\leq N}(1-f(\bx_k -\bx_m)^2)\1_{B(0,b)}(\bx_k -\bx_m).
\end{equation}
To prove \eqref{ine:bounds12}, we compute 
\begin{align*}
    &\frac{1}{2\pi}\int_{\R^2} |1-f^2| \, \dd \bx =
    \int_{0}^b (1-f^2(r))\,r \dd r = \frac{b^2}{2} - \int_0^R f^2(r) \, r \dd r - \int_R^b f^2(r) \, r \dd r \\&= \frac{b^2}{2} -  \frac{8 \left(\frac{R}{b}\right)^{2\alpha}  R^2}{\left(2\left(1+ \left(\frac{R}{b}\right)^{2\alpha}\right) + g\left(1-\left(\frac{R}{b}\right)^{2\alpha}\right)\right)^2} 
    \\ &\quad - \frac{  (2+g)^2  \frac{1}{2+2\alpha} \left(\frac{1}{b}\right)^{2\alpha} (b^{2+2\alpha}-R^{2+2\alpha}) + (2-g)^2  \left(\frac{R}{b}\right)^{4\alpha}  \frac{1}{2-2\alpha} b^{2\alpha} (b^{2-2\alpha} - R^{2-2\alpha}) }{\left(2\left(1+ \left(\frac{R}{b}\right)^{2\alpha}\right) + g\left(1-\left(\frac{R}{b}\right)^{2\alpha}\right)\right)^2}
\\ & \quad - \frac{(4-g^2)  \left(\frac{R}{b} \right)^{2\alpha} (b^2-R^2)}{\left(2\left(1+ \left(\frac{R}{b}\right)^{2\alpha}\right) + g\left(1-\left(\frac{R}{b}\right)^{2\alpha}\right)\right)^2}\\& = \frac{b^2 \alpha}{2\left(2\left(1+ \left(\frac{R}{b}\right)^{2\alpha}\right) + g\left(1-\left(\frac{R}{b}\right)^{2\alpha}\right)\right)^2} \left ( \frac{(2+g)^2}{1+\alpha}  - \frac{(2-g)^2}{1-\alpha}   \left(\frac{R}{b}\right)^{4\alpha} \right )
\\ &  \quad\quad +
\frac{\left(\frac{R}{b} \right)^{2\alpha } R^2 \alpha}{2\left(2\left(1+ \left(\frac{R}{b}\right)^{2\alpha}\right) + g\left(1-\left(\frac{R}{b}\right)^{2\alpha}\right)\right)^2}
\left(\frac{(4+g^2) \alpha + 4g}{1-\alpha^2} \right)
\\ & \lesssim (1+g^2) b^2 \alpha.
\end{align*}
Here, we have factored $b^2$ terms and $R^2$ terms, noted that 
\begin{align*}
     2\left(1+ \left(\frac{R}{b}\right)^{2\alpha}\right) + g\left(1-\left(\frac{R}{b}\right)^{2\alpha}\right)
     = 2+ g + (2-g) \left(\frac{R}{b}\right)^{2\alpha},
\end{align*}
and
used $0 \leq \alpha< \frac{1}{4}$, $0 \leq  R<b$, and \eqref{eq:conditionong}.
Moreover, by Proposition \ref{prop:obviousboundonJastrow}, we have $(1-f^2)^2 \leq (1-f^2)$. Hence, 
\begin{align*}
    \norm{1-f^2}^2_2 \leq \norm{1-f^2}_1 \lesssim (1+g^2) b^2 \alpha.
 \end{align*}
\end{proof}

\subsection{Density of the trial state}

In this part, we estimate the density of the trial state.
\begin{lemma} For the density of the trial state, we have the following estimate:
\begin{equation}\label{eq: normJastrow}
\norm{\Psi^{\mathrm{trial}}}^2_{L^2(\R^{2N})}=1+O\left(  N b^2 \beta (1+g^2) \norm{u}_4^4 \right).
\end{equation}
\end{lemma}
\begin{proof}
A direct calculation provides
\begin{align*}
\norm{\Psi^{\mathrm{trial}}}^2_{L^2(\R^{2N})}=\int_{\R^{2N}}|\Phi|^2+\int_{\R^{2N}}\left(F^2-1\right)|\Phi|^2= 1+\mathcal{E}
\end{align*}
where the error term $\mathcal{E}$ can be bounded using Lemma~\ref{lem:propjastrow}, Cauchy-Schwarz inequality, and Young's convolution inequality as
\begin{align}
|\mathcal{E}|&\leq \sum_{1\leq k<m\leq N}\int_{\R^{2N}}\left(1-f(\bx_k -\bx_m)^2\right)\prod_{j=1}^N|u(\bx_j)|^2 \mathrm{d}\sx\nonumber\\
&\leq \frac{1}{2}N(N-1) \norm{\left((1-f^2)\ast |u|^2\right)|u|^2}_{L^1(\R^2)}\nonumber\\
&\leq \frac{1}{2}N(N-1) \norm{|u|^2}_2^2\norm{1-f^2}_{L^1(\R^2)}\nonumber\\
&\lesssim (1+g^2) N b^2\beta \norm{u}_4^4.
\label{eq:theL2errorterm}
\end{align}
\end{proof}

Now, we prove the convergence of the densities.

\begin{lemma}
\label{lem:convergenceofdenisty}
    Let $\varrho_\Psi$ be as in \eqref{eq:one-body-density}. Then,
\begin{align*}
   \left\|N^{-1}\varrho_\Psi - |u|^2 \right\|_{L^1(\R^2)} = O\left( (1+g^2) \beta N b^2  \|u\|_4^4 \right).
\end{align*}
    
\end{lemma}
\begin{proof}
Note that by the definitions and Lemma~\ref{lem:propjastrow}, we have 
\begin{align*}
    \left\|N^{-1}\varrho_\Psi - |u|^2 \right\|_{L^1(\R^2)} 
&=  \int_{\R^2} \left| \int_{\R^{2(N-1)}} |\Psi(\bx_1,\bx_2,\ldots,\bx_N)|^2 \,\dd\bx_2 \ldots \dd\bx_N -|u(\bx_1)|^2 \right| \, \dd \bx_1
\\ &\leq   \int_{\R^2} \left| \int_{\R^{2(N-1)}} |u(\bx_1)|^2 \ldots |u(\bx_n)|^2 \,\dd\bx_2 \ldots \dd\bx_N -|u(\bx_1)|^2 \right| \, \dd \bx_1 \\&+ 
\int_{\R^{2N}}\left(1-F^2\right)|\Phi|^2
\quad\lesssim \quad 0 + (1+g^2) N b^2\beta \norm{u}_4^4,
\end{align*}
where we used \eqref{eq:theL2errorterm} on the last inequality.
\end{proof}

\bigskip

\subsection{Energy of the trial state}
In view of computing the energy of our trial state, we first compute the new kinetic term it adds to the energy. By using the definition \eqref{eq:trial-f}, we obtain
that $f=\lambda_1|\bx|^\alpha + \lambda_2|\bx|^{-\alpha}$ on $A(R,b)$
where $\lambda_1,\lambda_2$ are as in \eqref{eq:thecoeffioff} and depend on $g,\alpha,R,b$,
and $f$ is constant in both $B(0,R)$ and $B^c(0,b)$. 
Hence
\begin{equation}\label{eq:jastrow}
\begin{aligned}
\nabla_{\bx_i}F&=: \nabla_i F= \sum_{\substack{j=1\\j\neq i}}^N\frac{\nabla_i f(\bx_i-\bx_j)}{f(\bx_i-\bx_j)}F=\alpha\sum_{\substack{j=1\\j\neq i}}^N   K_{R,b}(\bx_i -\bx_j) F,
\end{aligned}
\end{equation}
where we define the expression
\begin{equation}\label{eq:def-K}
    K_{R,b}(\bx_i -\bx_j)
    := \alpha\sum_{\substack{j=1\\j\neq i}}^N   \frac{(\bx_i -\bx_j)}{|\bx_i -\bx_j|^2}     \frac{ \lambda_1 |\bx_i -\bx_j|^{\alpha} - \lambda_2 |\bx_i -\bx_j|^{-\alpha}}{ \lambda_1 |\bx_i -\bx_j|^{\alpha} + \lambda_2 |\bx_i -\bx_j|^{-\alpha}} \1_{A(R,b)}(\bx_i -\bx_j).
\end{equation} 
Note that, with our assumptions \eqref{eq:assumptions},
\begin{equation}
\begin{aligned}
\lim_{N \to \infty} \lambda_1 &= \frac{2+g}{2  \left(1+ e^{-2\beta \omega}\right) + g \left(1- e^{-2\beta \omega}\right)},\\ \quad \lim_{N \to \infty} \lambda_2 &= \frac{(2-g) e^{-2\beta \omega}}{2  \left(1+ e^{-2\beta \omega}\right) + g \left(1- e^{-2\beta \omega}\right)}.\label{eq:limitofcoeffoff}
\end{aligned}
\end{equation}
We will further make use of shorthand notations 
$f_{ij}:=f(\bx_i-\bx_j)$,  $u_j:=u(\bx_j)$
for the labels of integration, and
\begin{align*}
&(\bx)^{-\perp}_{R} := \frac{\bx^{\perp}}{|\bx|_R^2}, \quad\quad\quad \textup{for every } \bx \in \R^2 \setminus \{0\},\\
 &   (\bx)_{R,b}^{-1}:=\frac{\bx}{|\bx|^2} \1_{A(R,b)}, \quad \textup{for every } \bx \in \R^2 \setminus \{0\}.
\end{align*}

We are now ready to plug the trial state $\Psi^{\mathrm{trial}}$ in the expectation value of the Hamiltonian \eqref{def:HN}. We  use the definition \eqref{eq:jastrow} of $K_{R,b}$ to denote the new kinetic term due to the presence of the Jastrow factor and obtain 
\begin{align*}
&\bra{F\Phi}H_N\ket{F\Phi}=\\
&\quad\int_{\R^{2N}}\sum_{i=1}^N\bigg| F(\sx)\nabla_i \Phi(\sx) +\im\alpha  \sum_{j\neq i}(\bx_i-\bx_j)_R^{-\perp} F(\sx)\Phi(\sx) +\alpha  \sum_{j\neq i}K_{R,b}(\bx_i -\bx_j)F(\sx)\Phi(\sx)\bigg|^2\mathrm{d}\sx \\
&\quad +g \alpha R^{-2}\sum_{i=1}^N\sum_{k \neq i}\int_{\R^{2N}}\1_{B(0,R)}(\bx_i-\bx_k)|F(\sx)|^2|\Phi(\sx)|^2\mathrm{d}\sx+\sum_{i=1}^N\int_{\R^{2N}}V(\bx_i)|F(\sx)|^2|\Phi(\sx)|^2\mathrm{d}\sx
\\
&=\int_{\R^{2N}}\sum_{i=1}^N F^2(\sx)\left| \nabla_i \Phi(\sx) +\alpha \Phi(\sx)\sum_{j\neq i}\left[\im(\bx_i-\bx_j)_R^{-\perp} +K_{R,b}(\bx_i -\bx_j) \right ]\right|^2\mathrm{d}\bx\\
& \quad +g\alpha R^{-2}\sum_{i=1}^N\sum_{k \neq i}\int_{\R^{2N}}\1_{B(0,R)}(\bx_i-\bx_k)|F(\sx)|^2|\Phi(\sx)|^2\mathrm{d}\bx+\sum_{i=1}^N\int_{\R^{2N}}V(\bx_i)|F(\sx)|^2|\Phi(\sx)|^2\mathrm{d}\sx\\
&
=\sum_{i=1}^N \int_{\R^{2N}}|F(\sx)|^2\Big(|\nabla_i \Phi(\sx)|^2+V(\bx_i)|\Phi(\sx)|^2+g\alpha R^{-2}\sum_{k \neq i}\1_{B(0,R)}(\bx_i-\bx_k)|\Phi(\sx)|^2
\\
&\quad +\alpha \sum_{j\neq i} \Biggl[ \overline{\Phi}(\sx)\nabla_i \Phi(\sx)\cdot\left[-\im(\bx_i-\bx_j)_R^{-\perp} + K_{R,b}(\bx_i -\bx_j)\right]
\\
&\qquad\qquad +\Phi(\sx)\nabla_i\overline{\Phi}(\sx)\cdot\left[\im(\bx_i-\bx_j)_R^{-\perp} +K_{R,b}(\bx_i -\bx_j) \right ] \Biggr]\\
 &\quad +\alpha^2 |\Phi(\sx)|^2\sum_{j\neq i}\sum_{k\neq i}\left[-\im(\bx_i-\bx_j)_R^{-\perp} +K_{R,b}(\bx_i -\bx_j) \right ] \cdot\left[\im(\bx_i-\bx_k)_R^{-\perp} +K_{R,b}(\bx_i -\bx_k)\right ]\Big)\mathrm{d}\sx.
\end{align*}
We then split the energy per particle
\begin{equation}\label{def:spltting}
N^{-1}\bra{F\Phi}H_N\ket{F\Phi} = \mathcal{K}+\mathcal{V}+\mathcal{W}+\mathcal{S}+\mathcal{J},
\end{equation}
into the kinetic, the trapping potential, the repulsive scalar interaction (or Pauli-spin) term, the singular vector interaction, and the current, respectively, where
\begin{align}
\mathcal{K} & := \frac{1}{N}\int_{\R^{2N}}\sum_{i=1}^N F^2(\sx)|\nabla_i \Phi(\sx)|^2\mathrm{d}\sx,\nonumber\\
\mathcal{V} &:= \frac{1}{N}\sum_{i=1}^N\int_{\R^{2N}}V(\bx_i)|F(\sx)|^2|\Phi(\sx)|^2\mathrm{d}\sx,\nonumber\\
\mathcal{W} &:= \frac{g \alpha}{NR^2} \sum_{i=1}^N\sum_{k \neq i}\int_{\R^{2N}}\1_{B(0,R)}(\bx_i-\bx_k)|F(\sx)|^2|\Phi(\sx)|^2\mathrm{d}\sx,\label{def:W}\\
\mathcal{S} &:= \frac{\alpha^2}{N}\int_{\R^{2N}}\sum_{\substack{i,j,k=1\\ j\neq i, k\neq i}}^N\ F^2(\sx) |\Phi(\sx)|^2\left[-\im(\bx_i-\bx_j)_R^{-\perp} +K_{R,b}(\bx_i -\bx_j)\right ] \label{def:S}\\
    &\qquad\qquad \cdot\left[\im(\bx_i-\bx_k)_R^{-\perp} +K_{R,b}(\bx_i -\bx_j)\right ]\mathrm{d}\sx,\nonumber\\
\mathcal{J} &:= \frac{\alpha}{N}\int_{\R^{2N}}\sum_{i=1}^N\sum_{j\neq i}F^2(\sx)\bigg(\overline{\Phi}(\sx)\nabla_i \Phi(\sx)\cdot\left[-\im(\bx_i-\bx_j)_R^{-\perp} +K_{R,b}(\bx_i -\bx_j)\right ]
+\mathrm{h.c}\bigg ]\bigg)\mathrm{d}\sx.\nonumber
\end{align}
In the rest of this section, we study each of these terms separately.

\subsection{Kinetic and potential terms}


\begin{proposition}\label{lem:K}For the kinetic term, we have
\begin{equation*}
\mathcal{K}=\left(1+O\bigl( (b^2 \beta (1+g^2)  N+b\sqrt{\beta (1+g^2)  N}) \norm{u}^4_4 \bigr)\right)\int_{\R^2}|\nabla u|^2.
\end{equation*}
\end{proposition}

\begin{proof}
We can extract the Dirichlet energy term as
\begin{align}
\mathcal{K}&=N^{-1}\int_{\R^{2N}}\sum_{i=1}^N F(\sx)^2|\nabla_i u_i|^2\prod_{\substack{k=1,\\ k\neq i}}^N|u_k|^2\mathrm{d}\sx\nonumber\\
&=N^{-1}\sum_{i=1}^N\int_{\R^2}|\nabla_i u_i|^2\mathrm{d} \bx_i + N^{-1}\int_{\R^{2N}}\sum_{i=1}^N (F^2(\sx)-1)|\nabla_i u_i|^2\prod_{\substack{k=1,\\ k\neq i}}^N|u_k|^2\mathrm{d}\sx\nonumber\\
&=:\int_{\R^2}|\nabla u|^2 +\mathcal{E}
\end{align}
where we used that $\int_{\R^2} |u|^2 =1$. We now bound the remaining term $|\mathcal{E}|$ using \eqref{ine:jastrow21}
\begin{align*}
|\mathcal{E}|&\leq N^{-1} \sum_{i=1}^N \sum_{1\leq k<m\leq N} \int_{\R^{2N}}(1-f(\bx_k -\bx_m)^2)\1_{B(0,b)}(\bx_k -\bx_m)|\nabla_i u_i|^2\prod_{\substack{j=1,\\ j\neq i}}^N|u_j|^2\mathrm{d}\sx\\
&=:\mathcal{E}_1+\mathcal{E}_2,
\end{align*}
where in $\mathcal{E}_1$ we only consider the sum over indices $k$ and $m$, which are not equal to $i$, and in $\mathcal{E}_2$, the sum over indices where $i$ equals either $k$ or $m$. Then, by Cauchy-Schwarz inequality and Young's convolution inequality, we have
\begin{align}
 \mathcal{E}_1 &\leq (N-1)(N-2)\int_{\R^2}  |\nabla u|^2 \int_{\R^{4}}(1-f(\bx_1 -\bx_2)^2) |u(\bx_1)|^2|u(\bx_2)|^2\mathrm{d}\bx_1\mathrm{d}\bx_2\nonumber\\
&\leq N(N-1)\norm{\left[(1-f^2)\ast |u|^2\right]|u|^2}_{L^1(\R^2)}\int_{\R^2}  |\nabla u|^2\nonumber\\
&\leq N (N-1)\norm{|u|^2}_2\norm{(1-f^2)\ast |u|^2}_2\int_{\R^2}  |\nabla u|^2\nonumber\\
&\leq N (N-1)\norm{|u|^2}_2^2\norm{1-f^2}_1\int_{\R^2} |\nabla u|^2\nonumber\\
&\lesssim b^2\beta (1+g^2) N\norm{u}_4^4\int_{\R^2}  |\nabla u|^2
\end{align}
where we applied \eqref{ine:bounds12} on the last inequality. Moreover,
\begin{align}
 \mathcal{E}_2 &\leq (N-1) \int_{\R^{4}}(1-f(\bx_1 -\bx_2)^2)|\nabla_1u(\bx_1)|^2|u(\bx_2)|^2\mathrm{d}\bx_1\mathrm{d}\bx_2\nonumber\\
&\leq (N-1)\norm{\left[(1-f^2)\ast |u|^2\right]|\nabla u|^2}_{L^1(\R^2)}\nonumber\\
&\leq (N-1)\norm{|\nabla u|^2}_1\norm{(1-f^2)\ast |u|^2}_{\infty}\nonumber\\
&\leq (N-1)\norm{\nabla u|}_2^2\norm{|u|^2}_2\norm{1-f^2}_2\nonumber\\
&\lesssim \sqrt{\beta (1+g^2)  N} b \norm{u}^4_4\int_{\R^2} |\nabla u|^2
\end{align}
where we used \eqref{ine:bounds12} again.
\end{proof}

\begin{proposition}\label{lem:V}
The external  
potential term provides
\begin{equation*}
\mathcal{V}=\int_{\R^2}  V|u|^2+O\left( \left(N b^2 \beta (1+g^2) + \sqrt{N b^2 \beta (1+g^2) }  \right)\norm{u}^4_4 \int_{\R^2}  |V|\,|u|^2 \right).
\end{equation*}
\end{proposition}
\begin{proof}
We rewrite
\begin{align}
\mathcal{V}&=N^{-1}\sum_{i=1}^N\int_{\R^{2N}}V(\bx_i)F^2(\sx) |\Phi(\sx)|^2\mathrm{d}\sx\nonumber\\
&= \int_{\R^2} V|u|^2 + N^{-1}\sum_{i=1}^N\int_{\R^{2N}}(F^2(\sx)-1)V(\bx_i)|\Phi(\sx)|^2\mathrm{d}\sx\nonumber\\
&=: \int_{\R^2} V|u|^2 + \mathcal{E},
\end{align}
where we can estimate the error $|\mathcal{E}|$ using Lemma~\ref{lem:propjastrow} to obtain
\begin{align}
| \mathcal{E}|&\leq N^{-1}\sum_{i=1}^N\sum_{1\leq k<m\leq N}\int_{\R^{2N}}(1-f(\bx_k -\bx_m)^2)\,|V(\bx_i)|\,|\Phi(\sx)|^2\mathrm{d}\sx\nonumber\\
&=:\mathcal{E}_1 + \mathcal{E}_2.
\end{align}
Here $\mathcal{E}_2, \mathcal{E}_1$ are defined according to the fact that $\bx_k$ and $\bx_m$ are both different from $\bx_i$ or not, respectively. To estimate $ \mathcal{E}_1$, for which $k$ or $m$ equals $i$, we consider
\begin{align}
| \mathcal{E}_1|&\leq (N-1)\int_{\R^{4}}(1-f(\bx_1 -\bx_2)^2)\,|V(\bx_2)|\,|u(\bx_1)|^2|u(\bx_2)|^2\mathrm{d}\bx_1\mathrm{d}\bx_2\nonumber\\
&\leq (N-1)\norm{(1-f^2)\ast|u|^2}_{\infty} \int_{\R^2} \,|V|\,|u|^2\nonumber\\
&\leq (N-1)\norm{|u|^2}_2\norm{1-f^2}_{2} \int_{\R^2} \,|V|\,|u|^2\nonumber\\
&\lesssim \sqrt{ N b^2 \beta (1+g^2) } \norm{u}^4_4 \int_{\R^2} \,|V|\,|u|^2.
\end{align}
Here, we used Young's convolution inequality for the convolution on the third inequality and  \eqref{ine:bounds12} on the last inequality. Finally, for $ \mathcal{E}_2$, for which $k,m$ are different from $i$, we have
\begin{align}
| \mathcal{E}_2|&\leq (N-1)(N-2) \int_{\R^{4}}(1-f(\bx_1 -\bx_2)^2)\,|V(\bx_3)|\,|u(\bx_1)|^2|u(\bx_2)|^2|u(\bx_3)|^2\mathrm{d}\bx_1\mathrm{d}\bx_2\mathrm{d}\bx_3\nonumber\\
&\leq (N-1)(N-2) \norm{|u|^2}_2^2\norm{1-f^2}_{1} \int_{\R^2} |V|\,|u|^2\nonumber\\
&\lesssim N b^2 \beta (1+g^2) \norm{u}^4_4\int_{\R^2} |V|\,|u|^2.
\end{align}
Here, Cauchy-Schwarz inequality
and Young's convolution inequality are used on the second inequality and \eqref{ine:bounds12} is used on the last inequality. This completes the proof.
\end{proof}

\begin{lemma}\label{lem:W}The spin-interaction potential term provides
\begin{equation}
\begin{aligned}
\mathcal{W}&=\pi g \beta 
 f^2(R) 
\int_{\R^2} |u|^4\\
    &\quad+ O\left( g (\norm{\nabla u}_{\infty}  
 \norm{u}_{\infty}+R\norm{\nabla u}_{\infty}^2)\beta R + g (1+g^2) \beta^2 N b^2  \norm{u}^2_{\infty}\norm{u}^4_4 \right).
 \end{aligned}
\end{equation}
where the value $f(R)$ is given in \eqref{def:fr2}.

\end{lemma}
\begin{proof}
   We start from the definition \eqref{def:W} and write
    \begin{align}
        \mathcal{W}&=2 g N^{-1}\alpha R^{-2}\sum_{j=1}^N\sum_{k>j}\int_{\R^{2N}}\1_{B(0,R)}(\bx_j-\bx_k)|F(\sx)|^2|\Phi(\sx)|^2\mathrm{d}\sx \nonumber\\
        &=2 g N^{-1}\alpha R^{-2}\sum_{j=1}^N\sum_{k>j}\int_{\R^{2N}}\1_{B(0,R)}(\bx_j-\bx_k)\frac{F^2(\sx)}{f^2_{jk}}|\Phi(\sx)|^2f_{jk}^2\mathrm{d}\sx \nonumber\\
         &=2 f^2(R) g N^{-1}\alpha R^{-2}\sum_{j=1}^N\sum_{k>j}\int_{\R^{2N}}  
 \1_{B(0,R)}(\bx_j-\bx_k)|\Phi(\sx)|^2\mathrm{d}\sx \nonumber+\mathcal{E},
    \end{align}
    where
    \begin{equation}
        \mathcal{E}:=2 g N^{-1}\alpha R^{-2}\sum_{j=1}^N\sum_{k>j}\int_{\R^{2N}}\left(\frac{F^2(\sx)}{f^2_{jk}}-1\right)f^2_{jk}\,\1_{B(0,R)}(\bx_j-\bx_k)|\Phi(\sx)|^2\mathrm{d}\sx,\\
    \end{equation}
    which is an error term we treat at the end of the proof. We now use $\alpha=\beta(N-1)^{-1}$ to write
      \begin{align}
        \mathcal{W}-\mathcal{E}
         &=\frac{2 g f^2(R)  \alpha}{R^2} \frac{(N-1)}{2}\int_{\R^{4}}\1_{B(0,R)}(\bx_1-\bx_2)|u(\bx_1)|^2|u(\bx_2)|^2\mathrm{d}\bx_1\mathrm{d}\bx_2\nonumber\\
         &=  g f^2(R) \beta \int_{\R^{4}} \1_{B(0,1)}(\bx_1)|u(R \bx_1+\bx_2)|^2|u(\bx_2)|^2\mathrm{d}\bx_1\mathrm{d}\bx_2\nonumber\\
         &= g f^2(R) \beta \int_{\R^{4}}\1_{B(0,1)}(\bx_1)|u(\bx_2)|^4\mathrm{d}\bx_1\mathrm{d}\bx_2+\mathcal{E}',
    \end{align}
    where, by using Proposition \ref{prop:obviousboundonJastrow} and the bound $|u(R \bx_1+\bx_2) - u(\bx_2)| \leq R \|\nabla u\|_{\infty}$ for every $\bx_1 \in B(0,1), \bx_2 \in \R^2$, 
    it is obtained that
       \begin{align*}
       |\mathcal{E}'|
         &\leq g f^2(R) \beta \int_{\R^{4}}\1_{B(0,1)}(\bx_1)(2R\norm{\nabla u}_{\infty}\norm{ u}_{\infty}+R^2\norm{\nabla u}_{\infty}^2)|u(\bx_2)|^2\mathrm{d}\bx_1\mathrm{d}\bx_2 \\
         &\leq  \pi g \beta R \left( 2\norm{\nabla u}_{\infty}  
 \norm{u}_{\infty}+R\norm{\nabla u}_{\infty}^2 \right).
    \end{align*}
    We then have that
      \begin{align}
        \mathcal{W}-\mathcal{E}-\mathcal{E}'
         &=\pi g \beta f^2(R) \int_{\R^{2}}|u|^4.
    \end{align}
   Finally, to estimate the error term $\mathcal{E}$, we have
        \begin{align*}
     |\cE|
        &\leq 2 g N^{-1}\alpha R^{-2} f^2(R)\sum_{j=1}^N\sum_{k>j}\int_{\R^{2N}}\left| \sum_{1\leq n<m\leq N; (n,m) \neq (j,k)}(1-f^2_{nm})\right|  \1_{B(0,R)}(\bx_j-\bx_k)|\Phi|^2\mathrm{d}\bx,
    \end{align*}
    where we used (compare Lemma~\ref{lem:propjastrow})
    \begin{align*}
\frac{F^2(\sx)}{f^2_{jk}}=\frac{\prod_{1\leq n<m\leq N} f^2_{nm}}{f^2_{jk}}\geq 1-\sum_{1\leq n<m\leq N; (n,m) \neq (j,k)}(1-f^2_{nm}), \quad \textup{for every } \sx \in \R^{2N}.
\end{align*}
Now, we are left with two cases.\newline
\textbf{Case 1:} Define $\mathcal{E}_1$ as the sums in $\mathcal{E}$ where $j=n$ and $k\neq m$. Then,
        \begin{align*}
        \mathcal{E}_1
        &\leq g (N-1)(N-2) f^2(R) \alpha R^{-2}\int_{\R^{6}}(1-f^2_{12}) \1_{B(0,R)}(\bx_1-\bx_3)|u_1|^2|u_2|^2|u_3|^2\mathrm{d}\bx_1\mathrm{d}\bx_2\mathrm{d}\bx_3\\
        &\leq N g \beta R^{-2} \norm{u}^2_{\infty} \int_{\R^{6}}(1-f^2_{12})\1_{B(0,R)}(\bx_1-\bx_3)|u_1|^2|u_2|^2\mathrm{d}\bx_1\mathrm{d}\bx_2\mathrm{d}\bx_3
        \\ & \le  N g \beta \pi \norm{u}^2_{\infty} \int_{\R^{4}}(1-f^2_{12}) |u_1|^2|u_2|^2\mathrm{d}\bx_1\mathrm{d}\bx_2
        \\ & \le  N g \beta \pi \norm{u}^2_{\infty}\norm{|u|^2}_2 \norm{(1-f^2) \ast |u|^2}_2
        \\ & \le  N g \beta \pi \norm{u}^2_{\infty}\norm{|u|^2}^2_2 \norm{1-f^2  }_1
        \\ & \lesssim   g \beta^2 (1+g^2)  b^2  \norm{u}^2_{\infty}\norm{u}^4_4.
    \end{align*}
    Here, we used $|f(R)| \leq 1$ on the second inequality, Young's convolution inequality on the next to last inequality, and the bound \eqref{ine:bounds12} on the last inequality.
    \newline
\textbf{Case 2:} Finally, we define $\mathcal{E}_2$ as the sums in $\mathcal{E}$ where all the indices are distinct. Then, similarly to the computation above, we have
  \begin{align*}
        \mathcal{E}_2
        &\leq g (N-1)(N-2)^2 \alpha R^{-2}\int_{\R^{8}}(1-f^2_{12})\1_{B(0,R)}(\bx_3-\bx_4)|u_1|^2|u_2|^2|u_3|^2|u_4|^2\mathrm{d}\bx_1\mathrm{d}\bx_2\mathrm{d}\bx_3\mathrm{d}\bx_4\\
       &\leq g (N-2)^2 \beta R^{-2}\norm{u}_{\infty}^2\int_{\R^{8}}(1-f^2_{12})\1_{B(0,R)}(\bx_3-\bx_4)|u_1|^2|u_2|^2|u_3|^2\mathrm{d}\bx_1\mathrm{d}\bx_2\mathrm{d}\bx_3\mathrm{d}\bx_4\\
        &\leq g (N-2)^2 \pi \beta \norm{u}_{\infty}^2\int_{\R^{6}}(1-f^2_{12})|u_1|^2|u_2|^2\mathrm{d}\bx_1\mathrm{d}\bx_2
        \\ & \lesssim g \beta^2 (1+g^2) N b^2  \norm{u}^2_{\infty}\norm{u}^4_4,
    \end{align*}
    where we used \eqref{ine:bounds12} and Young's convolution inequality on the last inequality.
\end{proof}

\subsection{The singular term $\mathcal{S}$}
To study $\mathcal{S}$, we split it into its three-body part and its two-body part, which is also called the diagonal part. We then have $\mathcal{S}:=S^{\mathrm{diag}}+S^{3\mathrm{body}}$ where $S^{\mathrm{diag}}$ corresponds to the case $j=k$ in \eqref{def:S}and $S^{3\mathrm{body}}$ to the case $j\neq k$. In other words
\begin{equation}\label{def:S3b}
\begin{aligned}
S^{3\mathrm{body}}&:=N^{-1}\alpha^2\int_{\R^{2N}}\sum_{i=1}^N\sum_{j\neq i}\sum_{k\neq i\neq j} F^2(\sx) |\Phi(\sx)|^2 \\
& \quad\times \Big[-\im(\bx_i-\bx_j)_R^{-\perp}+ K_{R,b}(\bx_i -\bx_j)\Big ] \cdot \left[\im(\bx_i-\bx_k)_R^{-\perp} +K_{R,b}(\bx_i -\bx_k)\right ]\mathrm{d}\sx,
\end{aligned}
\end{equation}
recalling that 
\begin{equation}
    K_{R,b}(\bx_i -\bx_j)=(\bx_i-\bx_j)_{R,b}^{-1}  \frac{ \lambda_1 |\bx_i -\bx_j|^{\alpha} - \lambda_2 |\bx_i -\bx_j|^{-\alpha}}{ \lambda_1 |\bx_i -\bx_j|^{\alpha} + \lambda_2 |\bx_i -\bx_j|^{-\alpha}},
\end{equation}
and the two-body diagonal term
\begin{align}
S^{\mathrm{diag}}:&=N^{-1}\alpha^2\int_{\R^{2N}}\sum_{i=1}^N\sum_{j\neq i}F(\sx)^2 |\Phi(\sx)|^2\left|-\im(\bx_i-\bx_j)_R^{-\perp} + K_{R,b}(\bx_i -\bx_j)\right |^2\mathrm{d}\sx \nonumber\\
&=N^{-1}\alpha^2\int_{\R^{2N}}\sum_{i=1}^N\sum_{j\neq i}F(\sx)^2 |\Phi(\sx)|^2\bigg[|\bx_i-\bx_j|_R^{-\perp}\cdot|\bx_i-\bx_j|_R^{-\perp} \nonumber\\
&\quad + 
 |\bx_i-\bx_j|^{-2} \left( \frac{ \lambda_1 |\bx_i -\bx_j|^{\alpha} - \lambda_2 |\bx_i -\bx_j|^{-\alpha}}{ \lambda_1 |\bx_i -\bx_j|^{\alpha} + \lambda_2 |\bx_i -\bx_j|^{-\alpha}} \right)^2  \1_{A(R,b)}(\bx_i -\bx_j)\bigg ]\mathrm{d}\sx, \label{def:sdiag}
\end{align}
where we used that the cross terms cancel due to $(\bx_i-\bx_j)_R^{-\perp}\cdot (\bx_i-\bx_j)^{-1}_{R,b}=0$. In the next subsections, we first treat $S^{\mathrm{diag}}$ and then $S^{\mathrm{3body}}$.

\subsubsection{The diagonal term $S^{\mathrm{diag}}$}
In this part, we approximate the diagonal term $S^{\mathrm{diag}}$. 
\begin{lemma}\label{lem:Sdiag}
For any $u\in C_c^{\infty}(\R^2)$ with $\supp(u)\subset B(0,R_1)$ for $R_1>0$, we have 
\begin{align*} 
&S^{\mathrm{diag}}=  2\pi \beta \left(\lambda^2_1 \left(b^{2\alpha}-R^{2\alpha}\right) +\lambda^2_2 \left(R^{-2\alpha}- b^{-2\alpha} \right)\right) \int_{\R^{2}}|u|^4
\\ &+O\bigg(b \beta \left(\lambda^2_1 \left(b^{2\alpha}-R^{2\alpha}\right) +\lambda^2_2 \left(R^{-2\alpha}- b^{-2\alpha} \right)\right) ( \norm{\nabla u}_{\infty} \norm{u}_{\infty}+ b \|\nabla u\|^2_{\infty})\\&\quad\quad+b^2\beta^2 (1+g^2) \left(\lambda^2_1 \left(b^{2\alpha}-R^{2\alpha}\right) +\lambda^2_2 \left(R^{-2\alpha}- b^{-2\alpha} \right)\right)(N\norm{u}_4^4\norm{u}_{\infty}^2+\norm{u}_4^2\norm{u}_{8}^4)\bigg).
\end{align*}
\end{lemma}

\begin{proof}
We start from the definition \eqref{def:sdiag}
and factor a term $f^2_{ij}$ from $F^2$ inside the brackets which becomes
\begin{align}
S(\bx_i -\bx_j)&:=f^2_{ij}\frac{|\bx_i-\bx_j|^2}{|\bx_i-\bx_j|_R^{4}} \nonumber\\
&\quad+ f^2_{ij} 
 |\bx_i-\bx_j|^{-2} \left( \frac{ \lambda_1 |\bx_i -\bx_j|^{\alpha} - \lambda_2 |\bx_i -\bx_j|^{-\alpha}}{ \lambda_1 |\bx_i -\bx_j|^{\alpha} + \lambda_2 |\bx_i -\bx_j|^{-\alpha}} \right)^2  \1_{A(R,b)}(\bx_i -\bx_j)\nonumber\\
&=f^2_{ij}\frac{|\bx_i-\bx_j|^2}{|\bx_i-\bx_j|_R^{4}}+ \left( \lambda_1 |\bx_i -\bx_j|^{\alpha} - \lambda_2 |\bx_i -\bx_j|^{-\alpha}\right)^2\frac{\1_{A(R,b)}(\bx_i -\bx_j)}{|\bx_i -\bx_j|^2 }   \label{def:Sij}
\end{align}
such that we rewrite
\begin{align}
S^{\mathrm{diag}} &= \alpha^2\int_{\R^{2N}}\sum_{i=1}^N\sum_{j\neq i}\frac{F^2(\sx)}{f^2_{ij}} |\Phi(\sx)|^2S(\bx_i-\bx_j) \,\mathrm{d} \sx .\label{eq:Sdi}
\end{align}
In order to express $S$ in a more 
transparent way, we observe that by the definition \eqref{eq:trial-f}, we have $$f(\bx)=\1_{B(0,R)}(\bx)f(R)+(\lambda_1 |\bx|^{\alpha} + \lambda_2 |\bx|^{-\alpha})\1_{A(R,b)}(\bx)+ \1_{B^c(0,b)}(\bx)$$ for every $\bx \in \R^2$, and then that
\begin{equation}\label{split}
\frac{f^2(\bx) |\bx|^2}{|\bx|^4_R} =\frac{\1_{B(0,R)}(\bx) \, |\bx|^2}{R^4} f^2(R)+\frac{\1_{A(R,b)}(\bx)}{|\bx|^2} \left( \lambda_1 |\bx|^{\alpha} + \lambda_2 |\bx|^{-\alpha} \right)^2+\frac{\1_{B^c(0,b)}(\bx)}{|\bx|^2},
\end{equation}
allowing us to split $S$ as the sum of three terms living on three separated domains $\1_{B(0,R)}$, $\1_{A(R,b)}$, and $\1_{B^c(0,b)}$. By plugging the above in the first line of \eqref{def:Sij} and summing it with the third one, we obtain that
\begin{align*}
S(\by) &= \1_{B(0,R)}(\by) f^2(R) \frac{|\by|^2}{R^4}
 +2 \frac{\1_{A(R,b)}(\by)}{|\by|^2} \left( \lambda^2_1 |\by|^{2\alpha} + \lambda^2_2 |\by|^{-2\alpha}  \right)
 +\frac{\1_{B^c(0,b)}(\by)}{|\by|^2}\\
&=:S_1(\by)+S_2(\by)+S_3(\by).
\end{align*}
for every $\by \in \R^2$. Note that in the previous calculation, the terms in $\lambda_1\lambda_2$ canceled as they come with a plus sign in \eqref{split} and a minus sign in the third line of \eqref{def:Sij}. Plugging the above in \eqref{eq:Sdi}, we define the associated terms
\begin{equation}\label{def:splitS}
    S^{\mathrm{diag}} =: \sum_{i=1}^3 S^{\mathrm{diag}}_i.
\end{equation}
Among these terms, only $S^{\mathrm{diag}}_2$ will not be an error term. For that reason we study it first.
\newline
\textbf{Study of} $S^{\mathrm{diag}}_2$:
To begin, we replace the rest of the Jastrow factor $F^2$ by one:
\begin{align*}
 S^{\mathrm{diag}}_2&= N^{-1}\alpha^2  \int_{\R^{2N}}\sum_{i=1}^N\sum_{j\neq i}\frac{F^2(\sx)}{f_{ij}^2}  |\Phi(\sx)|^2S_2(\bx_i -\bx_j)\mathrm{d}\sx\\
& =2  N^{-1}\alpha^2 \int_{\R^{2N}}\sum_{i=1}^N\sum_{j\neq i} |\Phi(\sx)|^2 \1_{A(R,b)}(\bx_i -\bx_j)  \left(\lambda^2_1 |\bx_i -\bx_j|^{2\alpha-2} + \lambda^2_2 |\bx_i -\bx_j|^{-2\alpha-2} \right)\mathrm{d}\sx+\mathcal{E},
\end{align*}
where we will treat the following error term later in the proof:
\begin{align*}
\mathcal{E}&:=2 N^{-1}\alpha^2\int_{\R^{2N}}\sum_{i=1}^N\sum_{j\neq i} \left(\frac{F^2(\sx)}{f_{ij}^2}-1\right)|\Phi(\sx)|^2 \\
&\quad\quad\times\1_{A(R,b)}(\bx_i -\bx_j)  \left(\lambda^2_1 |\bx_i -\bx_j|^{2\alpha-2} + \lambda^2_2 |\bx_i -\bx_j|^{-2\alpha-2}  \right) \, \dd \sx.
\end{align*}
Then,
\begin{align*}
 S^{\mathrm{diag}}_2-\mathcal{E}&= 2N^{-1}\alpha^2 \int_{\R^{2N}}\sum_{i=1}^N\sum_{j\neq i} |\Phi(\sx)|^2 \1_{A(R,b)}(\bx_i -\bx_j)  \left(\lambda^2_1 |\bx_i -\bx_j|^{2\alpha-2} + \lambda^2_2 |\bx_i -\bx_j|^{-2\alpha-2}  \right)\mathrm{d} \sx\\
&=2(N-1)\alpha^2 \int_{\R^4}|u(\bx_1)|^2|u(\bx_2)|^2\1_{A(R,b)}(\bx_1 -\bx_2)  \\&\quad\quad\times\left(\lambda^2_1 |\bx_1 -\bx_2|^{2\alpha-2} + \lambda^2_2 |\bx_1 -\bx_2|^{-2\alpha-2}  \right)\mathrm{d}\bx_1\mathrm{d}\bx_2 \\
&=2\beta \alpha \int_{\R^4}|u(b\bx_1+\bx_2)|^2|u(\bx_2)|^2 \1_{A(R/b,1)}(\bx_1)\\
& \quad\quad\times\left(\lambda^2_1 b^{2\alpha} |\bx_1|^{2\alpha-2} + \lambda^2_2 b^{-2\alpha} |\bx_1|^{-2\alpha-2}  \right) \, \mathrm{d}\bx_1\mathrm{d}\bx_2 .
\end{align*}
Now, we use $|u(b\bx_1+\bx_2)-u(\bx_2)| \leq b \|\nabla u\|_{\infty}$ for every $\bx_1 \in B(0,1), \bx_2 \in \R^2$ to obtain
\begin{align*}
S^{\mathrm{diag}}_2-\mathcal{E}&= 2\beta\alpha\int_{\R^4}|u(\bx_2)|^4 \1_{A(R/b,1)}(\bx_1) \left(\lambda^2_1 b^{2\alpha} |\bx_1|^{2\alpha-2} + \lambda^2_2 b^{-2\alpha} |\bx_1|^{-2\alpha-2}  \right)\,\mathrm{d}\bx_1\mathrm{d}\bx_2+\mathcal{E}',
\end{align*}
where
\begin{equation*}
\begin{aligned}
|\mathcal{E}'|&\leq 2\beta\alpha \int_{\R^4}(2b \|\nabla u\|_{\infty} \|u\|_{\infty}+ b^2 \|\nabla u\|^2_{\infty})|u(\bx_2)|^2 \\&\quad\times\1_{A(R/b,1)}(\bx_1) \left(\lambda^2_1 b^{2\alpha} |\bx_1|^{2\alpha-2} + \lambda^2_2 b^{-2\alpha} |\bx_1|^{-2\alpha-2}  \right) \mathrm{d}\bx_1\mathrm{d}\bx_2
\\&\leq 2\pi b \beta \left(\lambda^2_1 \left(b^{2\alpha}-R^{2\alpha}\right) +\lambda^2_2 \left(R^{-2\alpha}- b^{-2\alpha} \right)\right) (2 \|\nabla u\|_{\infty} \norm{u}_{\infty}+ b \norm{\nabla u}^2_{\infty}) .
\end{aligned}
\end{equation*}
We finally get the main part of the $S_2$ term
\begin{align}
S^{\mathrm{diag}}_2-\mathcal{E}-\mathcal{E}'&= 2\beta\alpha\int_{\R^4}|u(\bx_2)|^4 \1_{A(R/b,1)}(\bx_1) \left(\lambda^2_1 b^{2\alpha} |\bx_1|^{2\alpha-2} + \lambda^2_2 b^{-2\alpha} |\bx_1|^{-2\alpha-2}  \right) \mathrm{d}\bx_1\mathrm{d}\bx_2\nonumber\\
&=2\pi\beta   \left(\lambda^2_1 \left(b^{2\alpha}-R^{2\alpha}\right) +\lambda^2_2 \left(R^{-2\alpha}- b^{-2\alpha} \right)\right) \int_{\R^2}      |u|^4.
\end{align}
Hence, there is left to estimate the two error terms $\mathcal{E}$ and $\cE'$. To bound $\mathcal{E}$, we use that for $1\leq i\neq j\leq N$ we have
\begin{align}
\frac{F^2}{f^2_{ij}}=\frac{\prod_{1\leq k<m\leq N}f^2_{km}}{f^2_{ij}}\geq 1-\sum_{\substack{1\leq k<m\leq N\\(k,m) \neq (i,j)}}(1-f^2_{km}).
\end{align}
In conclusion,
\begin{align*}
|\mathcal{E}|&\leq N^{-1}\alpha^2 \int_{\R^{2N}}\sum_{i=1}^N\sum_{j\neq i} \left|\frac{F^2(\sx)}{f_{ij}^2}-1\right||\Phi(\sx)|^2 \\& \quad\times\1_{A(R,b)}(\bx_i -\bx_j)  \left(\lambda^2_1 |\bx_i -\bx_j|^{2\alpha-2} + \lambda^2_2 |\bx_i -\bx_j|^{-2\alpha-2}  \right)\mathrm{d}\sx\\
&\leq N^{-1}\alpha^2 \int_{\R^{2N}}\sum_{i=1}^N\sum_{j\neq i} \left(\sum_{\substack{1\leq k<m\leq N\\(k,m) \neq (i,j)}}(1-f^2_{km})\right)|\Phi(\sx)|^2 \\ &\quad\times \1_{A(R,b)}(\bx_i -\bx_j)  \left(\lambda^2_1 |\bx_i -\bx_j|^{2\alpha-2} + \lambda^2_2 |\bx_i -\bx_j|^{-2\alpha-2}  \right)\mathrm{d} \sx.
\end{align*}
 We first study $\mathcal{E}_1$ defined as the sum of terms in $\mathcal{E}$, where none of the indices are equal. By using the symmetry of $\Phi$, we obtain
\begin{align*}
|\mathcal{E}_1|&\leq N^{-1}N^4 \alpha^2 \int_{\R^{4}}(1-f^2_{12})|u_1|^2|u_2|^2\mathrm{d}\bx_1\mathrm{d}\bx_2 \\
& \quad\times\int_{\R^4}|u_3|^2|u_4|^2\1_{A(R,b)}(\bx_3 -\bx_4)  \left(\lambda^2_1 |\bx_3 -\bx_4|^{2\alpha-2} + \lambda^2_2 |\bx_3 -\bx_4|^{-2\alpha-2}  \right)\mathrm{d}\bx_3\mathrm{d}\bx_4\\
&\leq 2\pi \frac{N^3}{(N-1)^2}\beta^2 \norm{|u|^2}_2^2\norm{1-f^2}_1\norm{|u|^2}_{\infty}\int_R^b \left(\lambda^2_1 r^{2\alpha  -1} + \lambda^2_2 r^{-2\alpha-1} \right) \, \dd r\\
&\lesssim  N \beta^2 (1+g^2) \left(\lambda^2_1 \left(b^{2\alpha}-R^{2\alpha}\right) +\lambda^2_2 \left(R^{-2\alpha}- b^{-2\alpha} \right)\right)\norm{|u|^2}_2^2\norm{u}^2_{\infty}b^2,
\end{align*}
where we used the bound \eqref{ine:bounds12} which provides $\norm{1-f^2}_1 \lesssim \alpha (1+g^2) b^2$ to compensate the $\alpha^{-1}$ coming from the integral between $R$ and $b$.

We call the remaining terms in $\mathcal{E}$ by $\mathcal{E}_2$, where one of $k,m$ is the same as one of $i,j$. Now, we control $\mathcal{E}_2$ by
\begin{align*}
|\mathcal{E}_2|&\leq N^2 \alpha^2 \int_{\R^{6}}(1-f^2_{12})|u_1|^2|u_2|^2|u_3|^2 \\  &\quad\times\1_{A(R,b)}(\bx_2 -\bx_3)  \left(\lambda^2_1 |\bx_2 -\bx_3|^{2\alpha-2} + \lambda^2_2 |\bx_2 -\bx_3|^{-2\alpha-2}  \right) \mathrm{d}\bx_1\mathrm{d}\bx_2\mathrm{d}\bx_3\\
&\leq 4 \beta^2\norm{|u|^2\left((1-f^2)\ast \left(|u|^2\left(\1_{A(R,b)}(\cdot)  \left(\lambda^2_1 |\cdot|^{2\alpha-2} + \lambda^2_2 |\cdot|^{-2\alpha-2}  \right) 
  \ast|u|^2\right)\right)\right)}_1\\
&\leq 4 \beta^2  \norm{|u|^2}_2\norm{(1-f^2)\ast\left(|u|^2\left(\1_{A(R,b)}(\cdot)  \left(\lambda^2_1 |\cdot|^{2\alpha-2} + \lambda^2_2 |\cdot|^{-2\alpha-2}  \right)\ast|u|^2\right)\right)}_2\\
&\leq 4 \beta^2 \norm{|u|^2}_2\norm{1-f^2}_1\norm{|u|^2\left(\1_{A(R,b)}(\cdot)  \left(\lambda^2_1 |\cdot|^{2\alpha-2} + \lambda^2_2 |\cdot|^{-2\alpha-2}  \right)\ast|u|^2\right)}_2\\
&\leq 4 \beta^2 \norm{|u|^2}_2\norm{|u|^2}_4\norm{1-f^2}_1\norm{\1_{A(R,b)}(\cdot)  \left(\lambda^2_1 |\cdot|^{2\alpha-2} + \lambda^2_2 |\cdot|^{-2\alpha-2}  \right)\ast|u|^2}_4\\
&\leq 4\beta^2 \norm{|u|^2}_2\norm{|u|^2}^2_4\norm{1-f^2}_1\norm{\1_{A(R,b)}(\cdot)  \left(\lambda^2_1 |\cdot|^{2\alpha-2} + \lambda^2_2 |\cdot|^{-2\alpha-2}  \right)}_1\\
&\lesssim b^2 \beta^2 (1+g^2) \left(\lambda^2_1 \left(b^{2\alpha}-R^{2\alpha}\right) +\lambda^2_2 \left(R^{-2\alpha}- b^{-2\alpha} \right)\right)  \norm{u}^2_4\norm{u}^4_8,
\end{align*}
where we used the Cauchy-Schwarz inequality and Young's convolution inequality, together with the bound \eqref{ine:bounds12}.
This concludes the study of $S^{\mathrm{diag}}_2$ providing
\begin{align}
S^{\mathrm{diag}}_2&=2\pi\beta  \left(\lambda^2_1 \left(b^{2\alpha}-R^{2\alpha}\right) +\lambda^2_2 \left(R^{-2\alpha}- b^{-2\alpha} \right)\right) \int_{\R^2}|u|^4\label{eq:S2d}\\
    &\quad+O\bigg(b \beta \left(\lambda^2_1 \left(b^{2\alpha}-R^{2\alpha}\right) +\lambda^2_2 \left(R^{-2\alpha}- b^{-2\alpha} \right)\right) ( \norm{\nabla u}_{\infty} \norm{u}_{\infty}+ b \|\nabla u\|^2_{\infty})\nonumber\\&\quad\quad+b^2\beta^2 (1+g^2) \left(\lambda^2_1 \left(b^{2\alpha}-R^{2\alpha}\right) +\lambda^2_2 \left(R^{-2\alpha}- b^{-2\alpha} \right)\right)(N\norm{u}_4^4\norm{u}_{\infty}^2+\norm{u}_4^2\norm{u}_{8}^4)\bigg).\nonumber
\end{align}

\textbf{Study of} $S^{\mathrm{diag}}_1$ \textbf{and} $S^{\mathrm{diag}}_3$: We bound the two error terms left. First,
\begin{align*}
|S^{\mathrm{diag}}_1|&\leq N^{-1}\left |\alpha^2\int_{\R^{2N}}\sum_{i=1}^N\sum_{j\neq i}\frac{F^2(\sx)}{f_{ij}^2}  |\Phi(\sx)|^2S_1(\bx_i -\bx_j)\right| \mathrm{d}\sx\\
& \leq \alpha^2 (N-1)\int_{\R^4}|u(\bx_1)|^2|u(\bx_2)|^2S_1(\bx_1-\bx_2)\mathrm{d}\bx_1\mathrm{d}\bx_2\nonumber\\
&\leq  \frac{\beta^2 f^2(R)}{(N-1) R^4} \int_{\R^4}|u(\bx_1)|^2|u(\bx_2)|^2 |\bx_1-\bx_2|^2  \1_{B(0,R)}(\bx_1-\bx_2)\mathrm{d}\bx_1\mathrm{d}\bx_2\nonumber\\ 
&\leq  \frac{\pi \beta^2}{2(N-1)} \norm{u}^2_{\infty},
\end{align*}
where we used $f\leq 1$, by Proposition \ref{prop:obviousboundonJastrow}, to bound the remaining $f$ in the Jastrow $F/f_{ij}$ and the symmetry of $\Phi$.
We now turn to $S^{\mathrm{diag}}_3$ and get
\begin{align*}
|S^{\mathrm{diag}}_3|&\leq N^{-1}\left |\alpha^2\int_{\R^{2N}}\sum_{i=1}^N\sum_{j\neq i}\frac{F^2(\sx)}{f_{ij}^2}  |\Phi(\sx)|^2S_3(\bx_i -\bx_j)\right| \mathrm{d}\sx\\
& \leq \frac{\beta^2}{(N-1)}\int_{\R^4}|u(\bx_1)|^2|u(\bx_2)|^2S_3(\bx_1-\bx_2)\mathrm{d}\bx_1\mathrm{d}\bx_2\nonumber\\
&\leq \frac{\beta^2}{(N-1)} \int_{\R^4}|u(\bx_1)|^2|u(\bx_2)|^2\frac{\1_{B^c(0,b)}(\bx_1-\bx_2)}{|\bx_1-\bx_2|^{2 }}\mathrm{d}\bx_1\mathrm{d}\bx_2\nonumber\\
&=\frac{\beta^2}{(N-1)} \norm{u}_{\infty}^4 \int_{B(0,R_1) \times B(0,R_1)}\frac{\1_{B^c(0,b)}(\bx_1-\bx_2)}{|\bx_1-\bx_2|^{2}}\mathrm{d}\bx_1\mathrm{d}\bx_2\\
& \leq \frac{\beta^2}{(N-1)} \|u\|^4_{\infty} |B(0,R_1)|  \ln \left ( \frac{2R_1}{b} \right ),
\end{align*}
and using that $u\in C^{\infty}_c(\R^2)$ with $\supp(u) \subset B(0,R_1)$ for some large $R_1>0$. 
We then obtain
\begin{align}
   |S^{\mathrm{diag}}_1|+ |S^{\mathrm{diag}}_3|&\lesssim \beta^2N^{-1}\left(\norm{u}_{\infty}^2 +\norm{u}_{\infty}^4R_1^2\log \left(\frac{R_1}{b} \right) \right).\label{eq:S1S4}
\end{align}
Finally, we apply \eqref{eq:S2d} and \eqref{eq:S1S4} on \eqref{def:splitS} to conclude the proof.
\end{proof}

\subsubsection{The three-body term}
We turn to the non diagonal three-body term defined in \eqref{def:S3b} as
\begin{equation*}
\begin{aligned}
S^{3\mathrm{body}}:=&\frac{\alpha^2}{N}\int_{\R^{2N}}\sum_{\{i,j,k\}}F^2(\sx) |\Phi(\sx)|^2\\
&\times \left[-\im(\bx_i-\bx_j)_R^{-\perp} + K_{R,b}(\bx_i -\bx_j)\right ]\cdot\left[\im(\bx_i-\bx_k)_R^{-\perp} + K_{R,b}(\bx_i -\bx_k)\right ]\mathrm{d}\sx\nonumber
\end{aligned}
\end{equation*}
where $\{i,j,k\}$ means the sum over $i$, $j$, and $k$ from $1$ to $N$ with $i\neq j\neq k$. We will show that only the first term of the above expression will contribute to the energy at main order with the two following Lemmas \ref{lem:three body} and \ref{lem:s3b}. The first Lemma \ref{lem:three body} below treats the effect of the new kinetic terms $K_{R,b}$ coming from the Jastrow, and Lemma \ref{lem:s3b} shows how to replace the remaining $F^2$ by one in the main term of Lemma \ref{lem:three body} by the use of Hardy's inequality.

\begin{lemma}\label{lem:three body}
For the three-body term $S^{3\mathrm{body}}$, we have
\begin{align*}
S^{3\mathrm{body}}&=N^{-1}\alpha^2\int_{\R^{2N}}\sum_{\{i,j,k\}}F^2(\sx) |\Phi(\sx)|^2 (\bx_i-\bx_j)_R^{-\perp} \cdot  (\bx_i-\bx_k)_R^{-\perp}\mathrm{d}\sx
\\ &\quad+ O\left(\beta^2   b^{2} (\lambda_1^2\, b^{2\alpha} + \lambda_2^2 \,b^{-2\alpha}) \norm{u}_{\infty}^4+ b\beta^2 (\lambda_1 b^{\alpha} +\lambda_2 b^{-\alpha}) \|u\|^2_{\infty} \left\|u\right\|^4_{\frac{8}{3}} \right).
\end{align*}
\end{lemma}
\begin{proof}

We expand $S^{3\mathrm{body}}$ in three different terms 
$$S^{3\mathrm{body}}:=S^{3\mathrm{body}}_R+S^{3\mathrm{body}}_b +S^{3\mathrm{body}}_c$$
with 
\begin{align*}
S^{3\mathrm{body}}_R&:=N^{-1}\alpha^2\int_{\R^{2N}}\sum_{\{i,j,k\}}F^2(\sx) |\Phi(\sx)|^2 (\bx_i-\bx_j)_R^{-\perp} \cdot (\bx_i-\bx_k)_R^{-\perp}\mathrm{d}\sx,\\
    S^{3\mathrm{body}}_b &:=N^{-1}\alpha^2  \int_{\R^{2N}}\sum_{\{i,j,k\}}F^2(\sx) |\Phi(\sx)|^2  K_{R,b}(\bx_i -\bx_j)\cdot K_{R,b}(\bx_i -\bx_k) \mathrm{d}\sx,\\
    S^{3\mathrm{body}}_c&:=2N^{-1}\alpha^2  \int_{\R^{2N}}\sum_{\{i,j,k\}}F^2(\sx) |\Phi(\sx)|^2\Re \bigg[-\im(\bx_i-\bx_j)^{-\perp}_R\cdot K_{R,b}(\bx_i -\bx_k)\bigg]\mathrm{d}\sx.
\end{align*}
 Now, we first estimate the two error terms $S^{3\mathrm{body}}_b$ and $S^{3\mathrm{body}}_c$. \newline
\textbf{Control of} $S^{3\mathrm{body}}_b$:
\begin{equation}
\begin{aligned}
 S^{3\mathrm{body}}_b &:=N^{-1}\alpha^2 \int_{\R^{2N}}\sum_{\{i,j,k\}}F^2(\sx) |\Phi(\sx)|^2\\& \quad\times (\bx_i-\bx_j)_{R,b}^{-1}  \frac{ \lambda_1 |\bx_i -\bx_j|^{\alpha} - \lambda_2 |\bx_i -\bx_j|^{-\alpha}}{ \lambda_1 |\bx_i -\bx_j|^{\alpha} + \lambda_2 |\bx_i -\bx_j|^{-\alpha}} \\&\quad\cdot (\bx_i-\bx_k)_{R,b}^{-1}  \frac{ \lambda_1 |\bx_i -\bx_k|^{\alpha} - \lambda_2 |\bx_i -\bx_k|^{-\alpha}}{ \lambda_1 |\bx_i -\bx_k|^{\alpha} + \lambda_2 |\bx_i -\bx_k|^{-\alpha}} \mathrm{d}\sx
\end{aligned}
\end{equation}
We factor from $F^2$ the terms $f_{ij}$ and $f_{ik}$ to obtain that
\begin{align*}
\left |S^{3\mathrm{body}}_b\right| &\leq\frac{\alpha^2  }{N}\int_{\R^{2N}}\sum_{\{i,j,k\}}\frac{F^2(\sx)}{f_{ij}f_{ik}} |\Phi(\sx)|^2\\&\quad\times\1_{A(R,b)}(\bx_i -\bx_j)  \left(\lambda_1 |\bx_i -\bx_j|^{\alpha-1} + \lambda_2 |\bx_i -\bx_j|^{-\alpha-1}\right)  \\&\quad\times\1_{A(R,b)}(\bx_i -\bx_k)  \left(\lambda_1 |\bx_i -\bx_k|^{\alpha-1} + \lambda_2 |\bx_i -\bx_k|^{-\alpha-1}\right)   \mathrm{d} \sx\nonumber\\
&\leq \alpha^2 (N-1)(N-2)\int_{\R^{6}}|u(\bx_1)|^2|u(\bx_2)|^2|u(\bx_3)|^2\\& \quad\times\1_{A(R,b)}(\bx_1 -\bx_2) \left(\lambda_1 |\bx_1 -\bx_2|^{\alpha-1} + \lambda_2 |\bx_1 -\bx_2|^{-\alpha-1}\right) \\&\quad\times \1_{A(R,b)}(\bx_2 -\bx_3) \left(\lambda_1 |\bx_2 -\bx_3|^{\alpha-1} + \lambda_2 |\bx_2 -\bx_3|^{-\alpha-1}\right) \mathrm{d}\bx_1\mathrm{d}\bx_2\mathrm{d}\bx_3\nonumber\\
&\leq\beta^2 \int_{\R^{2}}|u(\bx_1)|^2\mathrm{d}\bx_1\norm{|u(\cdot)|^2\ast\left( \1_{A(R,b)}(\cdot) \left(\lambda_1 |\cdot|^{\alpha-1} + \lambda_2 |\cdot|^{-\alpha-1}\right)\right)}_{\infty}^2\nonumber\\
&\leq \beta^2 \norm{|u|^2}^2_{\infty} \norm{ \1_{A(R,b)}(\cdot) \left(\lambda_1 |\cdot|^{\alpha-1} + \lambda_2 |\cdot|^{-\alpha-1}\right)}_{1}^2\nonumber\\
&\lesssim  \beta^2   b^{2}  (\lambda_1^2\, b^{2\alpha} + \lambda_2^2\,b^{-2\alpha})\norm{u}_{\infty}^4,
\end{align*}
 where we used Young's convolution inequality on the one to the last inequality.
\newline
\textbf{Control of} $S^{3\mathrm{body}}_c$:
We remind that
\begin{align*}
S^{3\mathrm{body}}_c:&=2N^{-1}\alpha^2  \int_{\R^{2N}}\sum_{\{i,j,k\}}F^2(\sx) |\Phi(\sx)|^2\Re \bigg[-\im(\bx_i-\bx_j)^{-\perp}_R\\&\quad\cdot(\bx_i-\bx_k)_{R,b}^{-1}  \frac{ \lambda_1 |\bx_i -\bx_k|^{\alpha} - \lambda_2 |\bx_i -\bx_k|^{-\alpha}}{ \lambda_1 |\bx_i -\bx_k|^{\alpha} + \lambda_2 |\bx_i -\bx_k|^{-\alpha}}\bigg]\mathrm{d}\sx
\end{align*}
Now, we factor from the Jastrow the term $f_{ik}$ and use Proposition \ref{prop:obviousboundonJastrow} to derive that
\begin{align*}
\left|S^{3\mathrm{body}}_c\right|&\leq2N^{-1}\alpha^2 \int_{\R^{2N}}\sum_{\{i,j,k\}} |\Phi(\sx)|^2 \frac{|\bx_i -\bx_j|}{|\bx_i -\bx_j|^2_R}\\&\quad\times \1_{A(R,b)}(\bx_i-\bx_k) \left(\lambda_1 |\bx_i-\bx_k|^{\alpha-1} + \lambda_2 |\bx_i-\bx_k|^{-\alpha-1}\right) \mathrm{d}\sx\nonumber.
\end{align*}
Hence, by the symmetry of $\Phi$, we obtain
\begin{align*}
\left|S^{3\mathrm{body}}_c\right|
&\leq \alpha^2 (N-1) (N-2)\int_{\R^{6}}|u(\bx_1)|^2|u(\bx_2)|^2|u(\bx_3)|^2\\
&\quad \times\frac{1}{|\bx_1-\bx_2|}   \1_{A(R,b)}(\bx_1 -\bx_3) \left(\lambda_1 |\bx_1-\bx_3|^{\alpha-1} + \lambda_2 |\bx_1-\bx_3|^{-(\alpha+1)}\right) \mathrm{d}\bx_1\mathrm{d}\bx_2\mathrm{d}\bx_3
\\ 
& \leq \beta^2 \int_{\R^2} |u|^2  \left( |u|^2(\cdot)  \ast \frac{1}{|\cdot|} \right) \left( |u|^2(\cdot)  \ast   \1_{A(R,b)(\cdot)}\left(\lambda_1 |\cdot|^{\alpha-1} + \lambda_2 |\cdot|^{-(\alpha+1)}\right)    \right)
\\ 
&\leq \beta^2 \|u\|^2_{\infty}  \left\||u|^2(\cdot)  \ast \frac{1}{|\cdot|}\right\|_{4}   \left\||u|^2(\cdot)    \ast   \1_{A(R,b)(\cdot)}\left(\lambda_1 |\cdot|^{\alpha-1} + \lambda_2 |\cdot|^{-(\alpha+1)}\right) \right\|_{\frac{4}{3}} 
\\ 
&\lesssim  \beta^2 \|u\|^2_{\infty}  \left\||u|^2\right\|_{\frac{4}{3}}^2    \left\||\cdot|^{-1}\right\|_{2,w} \left\| \1_{A(R,b)(\cdot)}\left(\lambda_1 |\cdot|^{\alpha-1} + \lambda_2 |\cdot|^{-(\alpha+1)}\right)\right\|_{1}
\\ 
& \lesssim  \beta^2 (\lambda_1 b^{1+\alpha } +\lambda_2 b^{1-\alpha}) \|u\|^2_{\infty}  \left\|u\right\|^4_{\frac{8}{3}}.
\end{align*}  
Here, we used H\"older's inequality on the third inequality, and weak Young's convolution inequality on the fourth inequality. This last estimate concludes the proof.
\end{proof}

We now have to replace $F^2$ by $1$ in the main term of Lemma \ref{lem:three body}. This is the content of the following lemma.
\begin{lemma}\label{lem:s3b}
For the term $S_R^{3\mathrm{body}}$, we have the estimate
\begin{align*}
&S_R^{3\mathrm{body}}=\beta^2\int_{\R^2}|u|^2|\bA^R[|u|^2]|^2
\\&+O\bigg(\frac{\beta^2}{N}\norm{u}^2_4\norm{u}_{\frac{8}{3}}^4 + \beta^3 N b^2 (1+g^2)\norm{u}^2_{\infty} \norm{u}^4_4\norm{u}_{\frac{8}{3}}^4+ \frac{ \beta^3 b^2  (1+g^2)}{N} ( \norm{u}^2_{\infty} +\norm{\nabla u}^2_{\infty}+1)\\& + \beta^3 N^{-1} \bigl(\lambda^2_1 ( b^{2\alpha} - R^{2\alpha} ) + \lambda^2_2  ( R^{-2\alpha} - b^{-2\alpha} ) \bigr) \norm{u}^2_{\infty}\bigg).
\end{align*}

\end{lemma}
\begin{proof}
We start from the main term of of Lemma \ref{lem:three body}. The aim is to replace the Jastrow factor by $1$.
We use the splitting
 \begin{align*}
&N^{-1}\alpha^2\int_{\R^{2N}}\sum_{\{i,j,k\}}F^2(\sx)\, |\Phi(\sx)|^2 (\bx_i-\bx_j)_R^{-\perp}\cdot (\bx_i-\bx_k)_R^{-\perp}\mathrm{d}\sx\nonumber\\
&= \alpha^2(N-1)(N-2)
\int_{\R^{2(N-3)}} \frac{F^2(\sx)\, |\Phi(\sx)|^2}{|u(\bx_1)|^2|u(\bx_2)|^2|u(\bx_3)|^2 f^2(\bx_1-\bx_2)f^2(\bx_2-\bx_3)f^2(\bx_1-\bx_3)}
\\&\quad\times \int_{\R^6}|u(\bx_1)|^2|u(\bx_2)|^2|u(\bx_3)|^2 f^2(\bx_1-\bx_2)f^2(\bx_2-\bx_3)f^2(\bx_1-\bx_3)
\\&\quad\times \sum_{\{\bx_1,\bx_2,\bx_3\}}(\bx_1-\bx_2)_R^{-\perp} \cdot(\bx_1-\bx_3)_R^{-\perp}\mathrm{d}\sx
\\& =  \alpha^2(N-1)(N-2)\int_{\R^6}|u(\bx_1)|^2|u(\bx_2)|^2|u(\bx_3)|^2(\bx_1-\bx_2)_R^{-\perp} \cdot(\bx_1-\bx_3)_R^{-\perp} \mathrm{d}\bx_1\mathrm{d}\bx_2\mathrm{d}\bx_3
\\&\quad+ \alpha^2(N-1)(N-2)\int_{\R^6}|u(\bx_1)|^2|u(\bx_2)|^2|u(\bx_3)|^2(\bx_1-\bx_2)_R^{-\perp} \cdot(\bx_1-\bx_3)_R^{-\perp} 
\\&\quad\times \int_{\R^{2(N-3)}}\left( \frac{F^2(\sx)}{ f^2(\bx_1-\bx_2)f^2(\bx_2-\bx_3)f^2(\bx_1-\bx_3)} -1\right)\frac{|\Phi(\sx)|^2}{|u(\bx_1)|^2|u(\bx_2)|^2|u(\bx_3)|^2}\mathrm{d}\sx
\\&\quad+ \alpha^2(N-1)(N-2)
\int_{\R^{2N}} |\Phi(\sx)|^2 \left( \frac{F^2(\sx)}{ f^2(\bx_1-\bx_2)f^2(\bx_2-\bx_3)f^2(\bx_1-\bx_3)}\right)
\\& \quad\times(f^2(\bx_1-\bx_2)f^2(\bx_2-\bx_3)f^2(\bx_1-\bx_3)-1 ) \sum_{\{\bx_1,\bx_2,\bx_3\}}(\bx_1-\bx_2)_R^{-\perp} \cdot(\bx_1-\bx_3)_R^{-\perp} \,\mathrm{d}\sx
\\ &=: \beta^2 \left(1-\frac{1}{N-1}\right)\int_{\R^2}|u(\bx)|^2|\bA^R[|u|^2]|^2(\bx) \,\mathrm{d}\bx +\mathcal{E}'+ \mathcal{E}.
\end{align*}
 First, we note that 
\begin{align}
\label{eq:ARYoungbound}
    \int_{\R^2}|u(\bx)|^2|\bA^R[|u|^2]|^2(\bx) \,\mathrm{d}\bx \leq \norm{|u|^2}_2\norm{|u|^2}_{4/3}^2\norm{|\cdot|^{-1}}_{2,w}^2,
\end{align}
by the Cauchy-Schwarz inequality and Young's weak convolution inequality. We now have to bound the two error terms $\mathcal{E}', \mathcal{E}$. For the term $\mathcal{E}'$, by symmetry, we have 
\begin{align*}
   |\mathcal{E}'|  &\leq\beta^2 N^2\int_{\R^2}|u(\bx_1)|^2|\bA^R[|u|^2]|^2(\bx_1)\mathrm{d}\bx_1  \int_{\R^{2(N-3)}}  (1-f(\bx_4-\bx_5)^2) |u(\bx_4)|^2|u(\bx_5)|^2\mathrm{d}\bx_4\mathrm{d}\bx_5
 \\&
 \leq  \beta^2 N^2 \norm{1-f^2}_1 \norm{u}^2_{\infty}\norm{|u|^2}^2_2\norm{|u|^2}_{4/3}^2\norm{|\cdot|^{-1}}_{2,w}^2\\
 &\lesssim \norm{|\cdot|^{-1}}_{2,w}^2 \beta^3  \frac{N^2}{N-1} b^2 (1+g^2) \norm{u}^2_{\infty} \norm{|u|^2}^2_2\norm{|u|^2}_{\frac{4}{3}}^2\\
 &\lesssim \beta^3Nb^2 (1+g^2)\norm{u}^2_{\infty}\norm{u}^4_4\norm{u}_{\frac{8}{3}}^4,
\end{align*}
where we used \eqref{eq:ARYoungbound} on the second inequality and Lemma~\ref{lem:propjastrow} on the third inequality. 
For $\cE'$, define $F_3 := f(\bx_1-\bx_2)f(\bx_2-\bx_3)f(\bx_1-\bx_3)$. 
Then, for every $\varepsilon>0$, we can use the inequality of Lemma \ref{lem:3body}, providing that
$$0\leq \sum_{\{\bx_1,\bx_2,\bx_3\}}(\bx_1-\bx_2)_R^{-\perp} \cdot(\bx_1-\bx_3)_R^{-\perp} \lesssim -\Delta_{\bx_1}$$
and use it to bound the error term.
For any $\eps>0$,
\begin{align*}
|\mathcal{E}|\leq |\mathcal{E}_{\varepsilon}|&:= \alpha^2 (N-1)(N-2)\left | \int_{\R^{2N}}\left|\sqrt{(1-F^2_3)^2+\varepsilon}\right|^2 |\Phi(\sx)|^2 \sum_{\{\bx_1,\bx_2,\bx_3\}}(\bx_1-\bx_2)_R^{-\perp} \cdot(\bx_1-\bx_3)_R^{-\perp} \mathrm{d}\sx\right|\\
& \lesssim \beta^2   \int_{\R^{2(N-3)}} \left|\frac{\Phi(\sx)}{u(\bx_1) u(\bx_2) u(\bx_3)}\right|^2 \int_{\R^6}|\nabla( u(\bx_1) u(\bx_2) u(\bx_3) \sqrt{(1-F^2_3)^2 +\varepsilon}  )|^2\mathrm{d}\sx
\\& \lesssim   \beta^2  \int_{\R^{2N}} \left|\frac{\Phi(\sx)}{u(\bx_1) u(\bx_2) u(\bx_3)}\right|^2   ((1-F_3^2)^2+\varepsilon) |\nabla (u(\bx_1) u(\bx_2) u(\bx_3)) |^2\mathrm{d}\sx
\\&\quad+   \beta^2  \int_{\R^{2N}}  | \Phi(\sx) |^2  F^2 |\nabla F_{3}|^2 \frac{(1-F_3^2)^2}{(1-F_3^2)^2 +\varepsilon} \mathrm{d}\sx
\\ &\lesssim   \beta^2  \int_{\R^{2N}} \left|\frac{\Phi(\sx)}{u(\bx_1) u(\bx_2) u(\bx_3)}\right|^2  ((1-F_3^2)^2
+\varepsilon) |\nabla (u(\bx_1) u(\bx_2) u(\bx_3))|^2\mathrm{d}\sx \\
&\quad+   \beta^2  \int_{\R^{2N}}| \Phi(\sx) |^2   |\nabla F_3|^2\mathrm{d}\sx,
\end{align*}
Letting $\varepsilon \to 0$ and using $0 \leq F_3 \leq 1$, 
we have 
\begin{align*}
|\mathcal{E}|&\lesssim     \beta^2  \int_{\R^{2N}} \left|\frac{\Phi(\sx)}{u(\bx_1) u(\bx_2) u(\bx_3)}\right|^2  (1-F_3^2)^2 |\nabla (u(\bx_1) u(\bx_2) u(\bx_3))|^2 \mathrm{d}\sx +  \beta^2  \int_{\R^{2N}}| \Phi(\sx) |^2   |\nabla F_3|^2\mathrm{d}\sx\\&= \mathcal{E}_1 +\mathcal{E}_2.
\end{align*}
Now, we prove that both the error terms $\mathcal{E}_1, \mathcal{E}_2$ are small as $N \to \infty.$ For $\mathcal{E}_1$, we have 
\begin{align*}
    \mathcal{E}_1 &\lesssim  \beta^2 \int_{\R^{2N}} \left|\frac{\Phi(\sx)}{u(\bx_1) u(\bx_2) u(\bx_3)}\right|^2 |\nabla (u(\bx_1) u(\bx_2) u(\bx_3))|^2   (1-f(\bx_1-\bx_2)^2) \1_{B(0,b)}(\bx_1-\bx_2)\mathrm{d}\sx
 \\  & \lesssim  \frac{ \beta^3 b^2  (1+g^2)}{N-1} ( \norm{u}^2_{\infty} +\norm{\nabla u}^2_{\infty}+1).
\end{align*}
For $\mathcal{E}_2$, we have 
\begin{align*}
   \mathcal{E}_2& \lesssim \beta^2  \int_{\R^{2N}}| \Phi(\sx) |^2 |\nabla f_{12}|^2   
 \mathrm{d}\sx
   \\ &\lesssim  \beta^2 \alpha^2 \norm{u}^2_{\infty} \int_{R}^b   (\lambda_1 r^{\alpha-1} + \lambda_2  r^{-\alpha-1})^2   \, r \dd r
   \\& \lesssim  \beta^2 \alpha  \left(\lambda^2_1 \left( b^{2\alpha} - R^{2\alpha} \right) + \lambda^2_2  \left( R^{-2\alpha} - b^{-2\alpha} \right) \right) \norm{u}^2_{\infty}
   \\ &\lesssim  \beta^3 N^{-1} \left(\lambda^2_1 \left( b^{2\alpha} - R^{2\alpha} \right) + \lambda^2_2  \left( R^{-2\alpha} - b^{-2\alpha} \right) \right) \norm{u}^2_{\infty}.
\end{align*}
This completes the proof.
\end{proof}

\subsection{The two-body current term}

The last term to treat is the current term
\begin{align*}
\mathcal{J}:&=N^{-1}\alpha\int_{\R^{2N}}\sum_{i=1}^N\sum_{j\neq i}F^2(\sx)
\bigg(\overline{\Phi}(\sx)\nabla_i \Phi(\sx)\left[-\im(\bx_i-\bx_j)_R^{-\perp} +K_{R,b}(\bx_i -\bx_j)\right ]+\mathrm{h.c}\bigg)\mathrm{d}\sx\nonumber\\
 &= \frac{2\alpha}{N}\Re\sum_{\substack{i,j=1\\j\neq i}}^N\int_{\R^{2N}} F^2(\sx)\overline{\Phi}(\sx)\nabla_i \Phi(\sx) \cdot
 \left[-\im(\bx_i-\bx_j)_R^{-\perp} +(\bx_i-\bx_j)_{R,b}^{-1}  \frac{ \lambda_1 |\bx_i -\bx_j|^{\alpha} - \lambda_2 |\bx_i -\bx_j|^{-\alpha}}{ \lambda_1 |\bx_i -\bx_j|^{\alpha} + \lambda_2 |\bx_i -\bx_j|^{-\alpha}}\right ]\mathrm{d}\bx\\
 &=:\mathcal{J}_R +\mathcal{J}_b,
\end{align*}
in which the new kinetic term $K_{R,b}$ will be shown to have no effect. Indeed the term $\mathcal{J}_b$  will be small and the rest of the Jastrow factor $F$ will be replaced by one at the cost of a small error.
\begin{lemma}\label{lem:J}
For current term $\mathcal{J}$, we have
\begin{align*} 
\mathcal{J}&=2\beta\Re\int_{\R^{4}}|u(\bx_1)|^2 \overline{u(\bx_2)}\nabla u(\bx_2) \cdot \left[-\im(\bx_1-\bx_2)_R^{-\perp} \right ]\mathrm{d}\bx_1\mathrm{d}\bx_2 \\&\quad+O\bigg(b(b^{\alpha}\lambda_1 + b^{-\alpha}\lambda_2) \beta  \norm{u}^2_{\infty} \norm{\nabla u}_2 +\beta^2 N b^2 (1+g^2) \norm{u}^5_4 \norm{u}^2_{8/3} \norm{\nabla u}_2
\\&\quad+ 2\pi \beta b   \norm{\nabla u}_2\norm{u}_4\norm{u}^2_8 +\sqrt{N} b \sqrt{1+g^2}   \norm{u}^3_4 \norm{\nabla u}_{2}\norm{u}^2_{\frac{8}{3}}
 + \beta^2 b^2 (1+g^2) \norm{u}_3 \norm{\nabla u}_2 \norm{u}^2_4\norm{u}^2_8 \bigg)
\end{align*}
\end{lemma}
\begin{proof}
We first bound
\begin{align*}
    \mathcal{J}_b=&2 N^{-1} \alpha \Re\int_{\R^{2N}}\sum_{i=1}^N\sum_{j\neq i} F^2(\sx)\overline{\Phi}(\sx)\nabla_i \Phi(\sx) \\&\cdot(\bx_i-\bx_j)_{R,b}^{-1}  \frac{ \lambda_1 |\bx_i -\bx_j|^{\alpha} - \lambda_2 |\bx_i -\bx_j|^{-\alpha}}{ \lambda_1 |\bx_i -\bx_j|^{\alpha} + \lambda_2 |\bx_i -\bx_j|^{-\alpha}} \mathrm{d}\sx,
\end{align*}
 as follows:
\begin{align*}
 \left|\mathcal{J}_b\right|&\leq 2 N^{-1} \alpha \int_{\R^{2N}}\sum_{i=1}^N\sum_{j\neq i} \left|\frac{F^2(\sx)}{f_{ij}}\right|\left|\overline{\Phi}(\sx)\right|\left|\nabla_i \Phi(\sx)\right| 
\\& \quad\times \1_{A(R,b)}(\bx_i -\bx_j) \left( \lambda_1 |\bx_i -\bx_j|^{\alpha-1} + \lambda_2 |\bx_i -\bx_j|^{-\alpha-1} \right)
  \mathrm{d}\sx\\
 &\leq 2 N^{-1} \alpha N (N-1)\int_{\R^4}|u(\bx_2)|^2|u(\bx_1)||\nabla_{\bx_1} u(\bx_1)| \\& \quad\times\1_{A(R,b)}(\bx_1-\bx_2) \left( \lambda_1 |\bx_1 -\bx_2|^{\alpha-1} + \lambda_2 |\bx_1 -\bx_2|^{-\alpha-1} \right)\\
 &\leq 2\beta \norm{u}_2\norm{\nabla u}_2\norm{|u(\cdot)|^2\ast \1_{A(R,b)(\cdot)}\left(\lambda_1 |\cdot|^{\alpha-1} + \lambda_2 |\cdot|^{-(\alpha+1)}\right)}_{\infty}\\
& \leq  2\beta\norm{\nabla u}_2  \norm{u}^2_{\infty}\norm{\1_{A(R,b)(\cdot)}\left(\lambda_1 |\cdot|^{\alpha-1} + \lambda_2 |\cdot|^{-(\alpha+1)}\right)}_{1}
 \\
 &\lesssim b(b^{\alpha}\lambda_1 + b^{-\alpha}\lambda_2) \beta  \norm{u}^2_{\infty} \norm{\nabla u}_2.
 \end{align*}
 We have now to replace $F^2$ by $1$ in $\mathcal{J}_R$ to obtain
 \begin{align*}
\mathcal{J}_R&=2N^{-1}\alpha\Re\int_{\R^{2N}}\sum_{i=1}^N\sum_{j\neq i} \overline{\Phi(\sx)}\nabla_i \Phi(\sx) \cdot \left[-\im(\bx_i-\bx_j)_R^{-\perp} \right ] \mathrm{d}\sx\nonumber\\
 &\quad+2N^{-1}\alpha\Re\int_{\R^{2N}}\sum_{i=1}^N\sum_{j\neq i} (F^2(\sx)-1)\overline{\Phi}(\sx)\nabla_i \Phi(\sx) \cdot \left[-\im(\bx_i-\bx_j)_R^{-\perp} \right ] \mathrm{d}\sx\nonumber\\
 &=2\beta\Re\int_{\R^{4}}|u(\bx_1)|^2 \overline{u(\bx_2)}\nabla u(\bx_2) \cdot \left[-\im(\bx_1-\bx_2)_R^{-\perp} \right ]\mathrm{d}\bx_1\mathrm{d}\bx_2+\mathcal{E},
 \end{align*}
 To make the proof complete, the final step is to bound the error term $\mathcal{E}$. We have that

 \begin{align*}
 |\mathcal{E}| &\leq 2N^{-1}\alpha\int_{\R^{2N}}\sum_{i=1}^N\sum_{j\neq i} |F^2(\sx)-1||\overline{\Phi}(\sx)\nabla_i \Phi(\sx)|\frac{|\bx_i-\bx_j|}{|\bx_i-\bx_j|^2_R}\mathrm{d}\sx\nonumber,
 \end{align*}
 and we use \eqref{ine:jastrow211} to obtain
 \begin{align*}
 |\mathcal{E}| &\leq \frac{2\beta}{N(N-1)}\int_{\R^{2N}}\sum_{i=1}^N\sum_{j\neq i}\sum_{1\leq k<m\leq N}(1-f(\bx_k -\bx_m)^2)\1_{B(0,b)}(\bx_k -\bx_m) |\overline{\Phi}(\sx)\nabla_i \Phi(\sx)|\frac{|\bx_i-\bx_j|}{|\bx_i-\bx_j|^2_R}\mathrm{d}\sx.
 \end{align*}
 We will treat separately the following three cases.\newline
 \textbf{Case 1}: Let consider $\mathcal{E}_1$ be the sums in $\mathcal{E}$, where $k\neq i$ and $m\neq j$.\newline
 The case $k\neq j$ and $m\neq i$ because we can exchange $k$ and $m$.
  In $\mathcal{E}_1$, we have four different sums, and we can reduce the integral by symmetry to a four particles integral

 \begin{align*}
 |\mathcal{E}_1| &\leq  \frac{2\beta}{N (N-1)}  N(N-1)(N-2)(N-3)\nonumber\\
 &\quad\times\int_{\R^{8}}(1-f(\bx_1-\bx_2)^2) |u(\bx_1)|^2|u(\bx_2)|^2|u(\bx_3)|^2|\overline{u(\bx_4)}\nabla_z u(\bx_4)|\frac{|\bx_4-\bx_3|}{|\bx_4-\bx_3|^2_R}\mathrm{d}\bx_1\mathrm{d}\bx_2\mathrm{d}\bx_3\mathrm{d}\bx_4\nonumber\\
 &\lesssim  \beta (N-1)^{2}\norm{|u|^2\left(|u|^2\ast(1-f^2) \right)}_1\norm{\nabla u}_2\norm{|u|(|u|^2\ast |\cdot|^{-1})}_2\nonumber\\
 &\lesssim  \beta (N-1)^{2}\norm{|u|^2}_2\norm{|u|^2\ast(1-f^2)}_2\norm{\nabla u}_2\norm{u}_4 \norm{|u|^2}_{4/3}\norm{|x|^{-1}}_{2,w}\nonumber\\
 &\lesssim \beta (N-1)^{2}\norm{|u|^2}^2_2\norm{1-f^2}_1\norm{\nabla u}_2\norm{u}_4 \norm{|u|^2}_{4/3}\norm{|x|^{-1}}_{2,w}\nonumber\\
 &\lesssim N b^2   \beta^2 (1+g^2)   \norm{u}^5_4 \norm{u}^2_{8/3} \norm{\nabla u}_2.
 \end{align*}
 Here, we used H\"older's inequality on the third inequality, Young's convolution inequality on the third and fourth inequalities, and \eqref{ine:bounds12} on the last inequality.
  \newline

  \textbf{Case 2}: In the second case, we define $\mathcal{E}_2$ as the sums in $\mathcal{E}$, where $k=i$ and $m=j$.\newline
  Note that the case $k=j$ and $m=i$ is similar to this case. To estimate $\mathcal{E}_2$, we cosider
 \begin{align*}
 |\mathcal{E}_2| &\leq 2\beta \int_{\R^{4}}(1-f(\bx_1-\bx_2)^2)|u(\bx_2)|^2|\overline{u(\bx_1)}\nabla_x u(\bx_1)|\frac{|\bx_1-\bx_2|}{|\bx_1-\bx_2|^2_R}\mathrm{d}\bx_1\mathrm{d}\bx_2\nonumber\\
 &\leq 2\beta\norm{\nabla u}_2\norm{|u|\left(|u|^2\ast\left((1-f(\cdot)^2) |\cdot|^{-1}\right)\right)}_2\nonumber\\
 &\leq 2\beta\norm{\nabla u}_2\norm{u}_4\norm{|u|^2\ast\left((1-f(\cdot)^2) |\cdot|^{-1}\right)}_4\nonumber\\
 &\leq 2\beta\norm{\nabla u}_2\norm{u}_4\norm{|u|^2}_4\norm{(1-f(\cdot)^2) |\cdot|^{-1}}_1\nonumber\\
 &\leq 4\pi \beta\norm{\nabla u}_2\norm{u}_4\norm{u}^2_8   \int_{0}^b 1   =4\pi \beta b  \norm{\nabla u}_2\norm{u}_4\norm{u}^2_8,
 \end{align*}
 where we used Proposition~\ref{prop:obviousboundonJastrow} on the last inequality.
 \newline
  \textbf{Case 3}: In this case, we define $\mathcal{E}_3$ as the sums in $\mathcal{E}$, where $k=i$ and $m\neq j$.\newline
  Note that the case where $m=i$ and $k\neq j$ is identical and can be estimated similarly.
 
 \begin{align*}
 |\mathcal{E}_3| &\leq 2\beta (N-2) \int_{\R^{6}}(1-f(\bx_1-\bx_2)^2) |u(\bx_2)|^2|u(\bx_3)|^2|\overline{u(\bx_1)}\nabla_{\bx_1} u(\bx_1)|\frac{|\bx_1-\bx_3|}{|\bx_1-\bx_3|^2_R}\mathrm{d}\bx_1\mathrm{d}\bx_2\mathrm{d}\bx_3\nonumber\\
 &\leq 2\beta (N-1)\norm{|u|^2\left[\left((1-f^2)\right) \ast\left(\overline{u}\nabla u (|u|^2\ast |\cdot|^{-1})\right)\right]}_1\nonumber\\
 &\leq 2\beta (N-1)\norm{|u|^2}_2\norm{\left((1-f^2)\right) \ast\left(\overline{u}\nabla u (|u|^2\ast |\cdot|^{-1})\right)}_2\nonumber\\
 &\leq 2\beta (N-1)\norm{u}^2_4\norm{(1-f^2)}_2\norm{\overline{u}\nabla u (|u|^2\ast |\cdot|^{-1})}_1\nonumber\\
 &\leq 2\beta (N-1)\norm{u}^2_4\norm{(1-f^2)}_2\norm{\overline{u}\nabla u}_{4/3}\norm{|u|^2\ast |\cdot|^{-1}}_4\nonumber\\
 &\leq 2\beta (N-1)\norm{u}^2_4\norm{(1-f^2)}_2\norm{\nabla u}_{2}\norm{u}_4\norm{|u|^2\ast |\cdot|^{-1}}_4\nonumber\\
 &\leq 2\beta (N-1)\norm{u}^3_4\norm{(1-f^2)}_2\norm{\nabla u}_{2}\norm{|u|^2}_{4/3}\norm{ |\cdot|^{-1}}_{2,w}\nonumber\\
 &\lesssim b\sqrt{N}\beta^{\frac{3}{2}}  \sqrt{1+g^2}  \,  \norm{u}^3_4 \norm{\nabla u}_{2}\norm{u}^2_{\frac{8}{3}}.
 \end{align*} 
 Here, we used H\"older's inequality on third, fifth, and sixth inequalities, Young's convolution inequality on fourth and one to the last inequalities, and \eqref{ine:bounds12} on the last inequality.
\newline
  \textbf{Case 4}: Finally, we consider $\mathcal{E}_4$ as the sums in $\mathcal{E},$ where $k=j$ and $m\neq i$, and similar computation resolves the sums in $\mathcal{E}$ with $m=j$ and $k\neq i$. We have
 \begin{align*}
 |\mathcal{E}_4| &\leq 2\beta (N-2) \int_{\R^{6}}(1-f(\bx_3-\bx_2)^2) |u(\bx_2)|^2|u(\bx_3)|^2|\overline{u(\bx_1)}\nabla_{\bx_1} u(\bx_1)|\frac{|\bx_1-\bx_3|}{|\bx_1-\bx_3|^2_R}\mathrm{d}\bx_1\mathrm{d}\bx_2\mathrm{d}\bx_3\nonumber\\
 &\leq 2\beta N\norm{\overline{u}\nabla u\left[\left((|u|^2\ast(1-f^2)) |u|^2\right)\ast|\cdot|^{-1}\right]}_1\nonumber\\
 &\leq 2\beta (N-1)\norm{\overline{u}\nabla u}_{4/3}\norm{\left(|u|^2\ast(1-f^2)\right) |u|^2\ast|\cdot|^{-1}}_4\nonumber\\
 &\leq 2 \beta (N-1)\norm{\overline{u}\nabla u}_{4/3}\norm{\left(|u|^2\ast(1-f^2)\right) |u|^2}_{4/3}\norm{|\cdot|^{-1}}_{2,w}\nonumber\\
 &\leq 2 \beta (N-1) \norm{\overline{u}\nabla u}_{4/3}\norm{|u|^2\ast(1-f^2) }_{2}\norm{|u|^2}_4\norm{|\cdot|^{-1}}_{2,w}\nonumber\\
 &\leq 2 \beta (N-1)\norm{\overline{u}\nabla u}_{4/3}\norm{(1-f^2) }_{1}\norm{|u|^2}_2\norm{|u|^2}_4\norm{|\cdot|^{-1}}_{2,w}\nonumber\\
 &\leq 2 \beta (N-1) \norm{u}_3 \norm{\nabla u}_2 \norm{(1-f^2) }_{1}\norm{|u|^2}_2\norm{|u|^2}_4\norm{|\cdot|^{-1}}_{2,w}\nonumber
 \\ &\lesssim b^2 \beta^2  (1+g^2)  \norm{u}_3 \norm{\nabla u}_2 \norm{u}^2_4\norm{u}^2_8.
 \end{align*}
  Here, we used H\"older's inequality on third, fifth, and seventh inequalities, Young's convolution inequality on fourth and sixth inequalities, and \eqref{ine:bounds12} on the last inequality.
\end{proof}

\subsection{The final bound}

We can now finally prove Theorem \ref{thm:E} by combining the previous lemmas.
 
\begin{proof}
 We start from the splitting \eqref{def:spltting} in which $\mathcal{S}=\mathcal{S}^{\mathrm{diag}}+\mathcal{S}^{\mathrm{3b}}$ to write
 \begin{equation}
N^{-1}\bra{F\Phi}H_N\ket{F\Phi}=\mathcal{K}+S^{\mathrm{diag}}+S^{3\mathrm{body}}+\mathcal{J}+\mathcal{V}+\mathcal{W}.
\end{equation}
The application of Lemmas \ref{lem:K}, \ref{lem:V}, \ref{lem:W}, \ref{lem:Sdiag}, \ref{lem:three body}, \ref{lem:s3b}, and \ref{lem:J} provides
 \begin{align*}
N^{-1}\bra{F\Phi}H_N\ket{F\Phi}&=
\int_{\R^2}|\nabla u|^2+2\beta\Re \int_{\R^2} \overline{u}\nabla u\cdot\bA^R[|u|^2]+\beta^2\int_{\R^2}|u|^2|\bA^R[|u|^2]|^2\\
&\quad+ 2\pi\beta\, G(  2\beta\omega_N ,g)\int_{\R^2} |u|^4+\int_{\R^2} V|u|^2+\cE_N\\
&=\int_{\R^2} |(-\im\nabla +\beta\bA^R[|u|^2])u|^2+2\pi \beta\, G(  2\beta \omega_N ,g)\int_{\R^2} |u|^4+\int_{\R^2} V|u|^2+\cE_N.
\end{align*}
Now, by using $R=e^{-N\omega_N}$  with $\omega_N\to\omega \in [0,\infty]$, we get
\begin{align}
\lim_{N \to \infty} G(  2\beta\omega_N ,g)=\lim_{N \to \infty} \frac{1+\frac{g}{2} - (1-\frac{g}{2}) \left(\frac{R}{b}\right)^{2\alpha}}{1+\frac{g}{2} + \left(1-\frac{g}{2}\right) \left(\frac{R}{b}\right)^{2\alpha}}=  G( 2\beta \omega ,g),
\end{align}
where we used the assumption $N^{-1}\log b\to 0$. Moreover,
\begin{align*}
|\cE_N|&\leq C(u)(\beta^3+1) (g^3+1)\big( R +b^2 N+b\sqrt{N}+N^{-1}+ b(b^{\alpha}\lambda_1 + b^{-\alpha}\lambda_2) + b^{2} (\lambda_1^2\, b^{2\alpha}+\lambda_2^2\, b^{-2\alpha})
\\ &\quad+ (N^{-1}+ R+b^2 N) \left(\lambda^2_1 \left( b^{2\alpha} - R^{2\alpha} \right) + \lambda^2_2  \left( R^{-2\alpha} - b^{-2\alpha} \right) \right) + N^{-1} R_1^2 |\log R_1 |\big)\\
&=o_N(1),
\end{align*}
by using $\beta\in \R_+$, $bN^2\to 0$, $\frac{\log b}{N} \to 0$, and  $u\in C_c^{\infty}(B(0,R_1))$, where $C(u)$ is a constant depending only on $u$. Here, we used that $0\leq R<b$,
\begin{align*}
     \lambda_1^2\, b^{2\alpha}+\lambda_2^2\, b^{-2\alpha}= \frac{(2+g)^2 + \left(\frac{R}{b}\right)^{4\alpha} (2-g)^2 } { 
 \left( 2\left(1+\left(\frac{R}{b}\right)^{2\alpha}\right) + g\left(1-\left(\frac{R}{b}\right)^{2\alpha}\right)   \right)^2} 
 \leq \frac{4+g^2}{2},
\end{align*}
and
\begin{align*}
    & \lambda^2_1 \left(b^{2\alpha}-R^{2\alpha}\right) +\lambda^2_2 \left(R^{-2\alpha}- b^{-2\alpha} \right) \\ &= \left( \frac{2+g}{2\left(1+\left(\frac{R}{b}\right)^{2\alpha}\right) + g\left(1-\left(\frac{R}{b}\right)^{2\alpha}\right)}\right)^2 \left(1-\left(\frac{R}{b}\right)^{2\alpha}\right) \\ &\quad + \left ( \frac{2-g}{2\left(1+\left(\frac{R}{b}\right)^{2\alpha}\right) + g\left(1-\left(\frac{R}{b}\right)^{2\alpha}\right)} \right )^2 \left(\frac{R}{b}\right)^{2\alpha} \left(1-\left(\frac{R}{b}\right)^{2\alpha}\right) 
     \\& \leq \frac{4+g^2}{2}.
\end{align*}
Hence, by \eqref{eq: normJastrow} and $b^2 N \to 0$, we derive that
 \begin{equation}
 N^{-1}\frac{\bra{F\Phi}H_N\ket{F\Phi}}{\norm{F\Phi}^2}=\mathcal{E}_{\beta,\gamma,V}[u] +\mathcal{E}'_N
 \end{equation}
 where $\gamma = 2\pi\beta G(2\beta\omega,g)$ and $\lim_{N \to \infty} \mathcal{E}'_N = 0$. Here, we used that
 \begin{equation}
 \left|\mathcal \int|(-\im\nabla +\beta\bA[|u|^2])u|^2- \mathcal \int|(-\im\nabla +\beta\bA^R[|u|^2])u|^2\right|\lesssim R(1+\mathcal{E}_{\beta,0,0}[u]^{3/2})
\end{equation}
by \cite[Proposition A.6]{LunRou-15} to take the final limit $R\to 0$.
\end{proof}

\appendix
\section[\qquad Appendix]{Appendix}\label{sec:appendix}

\subsection{Form domains for ideal anyons}\label{sec:appendix-domains}

In this appendix we make precise the form domain for ideal anyons using a many-body Hardy inequality.
Similar observations were made in \cite{RouYan-23b} with slightly different methods. 

We define the kinetic energy operator
$T_\alpha := (-i\nabla+\alpha\sA)^2 = \sum_{j=1}^N (-i\nabla_{\bx_j}+\alpha\bA_j)^2$ 
for $N$ ideal abelian anyons with statistics parameter $\alpha \in \R$
by Friedrichs extension, starting from the non-negative quadratic form
$\Psi \mapsto \int_{\R^{2N}} |(-i\nabla+\alpha\sA)\Psi|^2$
on the dense domain $C^\infty_c(\R^{2N} \setminus \bDelta) \cap L^2_\sym(\R^{2N};\C)$ (the \emph{minimal} form and operator domain) and taking its form closure.
Equivalently, we could have started from the \emph{maximal} form domain 
$\Psi \in L^2_\sym(\R^{2N};\C)$ s.t.
$(-i\nabla+\alpha\sA)\Psi \in L^2(\R^{2N};\C^{2N})$
in the sense of distributions; see \cite[Section~2.2]{LunSol-14}. 
Because of the following many-anyon Hardy inequality 
(fermionic case $\alpha=1$ proved in \cite[Thm.~2.8]{HofLapTid-08}, anyonic in \cite{LunSol-13a,LarLun-16,LunQva-20}),
we actually stay inside $H^1(\R^{2N};\C)$, 
with a possible vanishing condition on the diagonals (exclusion principle),
\emph{equivalent} to the fermionic one.

\begin{lemma}[Ideal anyons Hardy inquality]
For any $\alpha \in \R$, $N \ge 2$, define the ``fractionality''
$$
	\alpha_N := \inf_{\substack{p \in \{0,1,\ldots,N-2\}\\q \in \Z}} 
            \big|(2p+1)\alpha - 2q\big|.
$$
Then, as quadratic forms,
\begin{equation}\label{eq:Hardy-nn}
	T_\alpha \ge 
	\frac{2\alpha_2^2}{N}
	\sum_{j=1}^N |\bx_j-\bx_{\mathrm{nn}(j)}|^{-2},
\end{equation}
where $\mathrm{nn}(j)$ is the nearest neighbor of particle $j$, 
and
\begin{equation}\label{eq:Hardy-max}
	T_\alpha \ge  
	\frac{4}{N}\max\left\{ \alpha_N^2, \frac{\alpha_2^2}{N-1} \right\}
	\sum_{1 \le j<k \le N} |\bx_j-\bx_k|^{-2}.
\end{equation}
\end{lemma}
\begin{proof}
    By an improved many-anyon Hardy inequality, 
    \cite[Thm.~5.5]{LunQva-20} 
    (or by following the original proof in \cite{LunSol-13a}), 
    $$
        \int_{\R^{2N}} |(-i\nabla+\alpha\sA)\Psi|^2
        \ge 
        \frac{2}{N} \int_{\R^{2N}} 
        \sum_{\substack{1 \le j,k \le N \\ j \neq k}} 
        \frac{\hat{\alpha}_{\sx}((\bx_j+\bx_k)/2,|\bx_j-\bx_k|/2)^2}{|\bx_j-\bx_k|^2} |\Psi|^2,
    $$
    where
    $$
        \hat{\alpha}_{\sx}(\bX,r) := \inf_{q \in \Z} |(2p+1)\alpha - 2q|,
    $$
    and $p=p_{\sx}(\bX,r)$ is exactly the number of particles $\sx$ inside the disk $B(\bX,r)$. 
    Note that, relative to the particles $j \neq k$, 
    that does not include these two (or more) particles on the boundary of the disk.
    Hence, 
    $\hat{\alpha}_{\sx}(\bX,r) \ge \alpha_N$ for all $\bX \in \R^2$, $r>0$,
    and $\sx \in \R^{2N}$.
    Further, for any fixed $j$, if $k=\mathrm{nn}(j)$ is taken to be the nearest neighbor to $j$,
    then $p_{\sx}((\bx_j+\bx_k)/2,|\bx_j-\bx_k|/2) = 0$ and so $\hat{\alpha}_{\sx}((\bx_j+\bx_k)/2,|\bx_j-\bx_k|/2) = \alpha_2$.
    Thus, we obtain \eqref{eq:Hardy-nn}, and by
    $$
        |\bx_j-\bx_{\mathrm{nn}(j)}|^{-2} \ge (N-1)^{-1} \sum_{k \neq j} |\bx_j-\bx_k|^{-2},
    $$
    we also obtain \eqref{eq:Hardy-max}.
\end{proof}

\begin{lemma}[Ideal anyons form domain]
The form domain for 
$T_\alpha$ 
is exactly
$$
    U^{-2n} H^1_\sym
    \qquad \text{if} \ \alpha
    = 2n \in 2\Z,
$$
and 
$$
    H^1_\sym \cap L^2_{W_\sA} =  
    U^{\pm 1} H^1_\asym
    \qquad \text{if} \ \alpha = 2n+\alpha_2 \notin 2\Z, 
$$
where $U(\sx) := \prod_{j<k} \frac{z_j-z_k}{|z_j-z_k|}$ defined on $\C^{N} \setminus \bDelta$ maps $L^2_\sym \to L^2_\asym \to L^2_\sym$ isometrically, and $W_\sA := \frac{1}{2} (\sum_j \bA_j^2)_{\diag} = \sum_{j<k} |\bx_j-\bx_k|^{-2}$ controls the presence of exclusion on the diagonals.
\end{lemma}
\begin{proof}
Let $\Psi$ be in the form domain of $T_\alpha$, i.e.
$\Psi \in L^2_\sym$ and 
$(-i\nabla+\alpha\sA)\Psi \in L^2(\R^{2N};\C^{2N})$.
Then we have that
$|\Psi| \in H^1_\sym$
by the diamagnetic inequality; 
see \cite[Lem.~4]{LunSol-14}.
Further, by the above improved many-anyon Hardy inequality \eqref{eq:Hardy-max},
$$
    \frac{2\alpha_2^2}{N(N-1)} \int_{\R^{2N}} \frac{|\Psi|^2}{|\bx_j-\bx_k|^2} d\sx < \infty
    \qquad \forall j \neq k.
$$
Also note that $\sA \cdot \sA$ and $W_{\sA}$ are equivalent in the sense that
$$
    2W_{\sA} \le \sum_j \bA_j^2 = 2W_{\sA} + \sum_{j<k<l} R_{jkl}^{-2}
    \le C_N W_{\sA},
$$
by non-negativity of the three-body term (where $R_{jkl} \in [0,\infty]$ denotes the radius of circle defined by the points $\bx_j,\bx_k,\bx_l$; see \cite{HofLapTid-08}), and Cauchy-Schwarz.
Hence, if $\alpha_2 \neq 0$, then $\Psi \in L^2_{W_\sA}$, and
$\sA\Psi \in L^2(\R^{2N};\C^{2N})$, so $\Psi \in H^1_\sym(\R^{2N};\C)$.
In the case that $\alpha=0$, we trivially have $\Psi \in H^1_\sym$ (by definition),
and if $\alpha = 2n$, then we can use the symmetric gauge transformation $U^{-2n}$ (a.e. in $\R^{2N}$) to remove the gauge potential and reduce to $\alpha=0$.

In the case of ideal fermions ($\alpha=1$; $\Phi = U\Psi \in H^1_\asym$), we have for their bosonic representation $\Psi \in U^{-1} H^1_\asym$ with the gauge potential $\sA$, and thus, by the same Hardy inequality, 
actually $\Psi \in H^1_\sym \cap L^2_{W_\sA}$. 
We also note that $U^{2n}\Psi \in H^1_\sym \cap L^2_{W_\sA}$ for any $n \in \N$.

Conversely, given $\Psi \in H^1_\sym \cap L^2_{W_\sA}$, we have that $\Phi := U\Psi \in L^2_\asym \cap L^2_{W_\sA}$, and thus $\nabla\Phi = U(\nabla + \sA)\Psi \in L^2$, so $\Phi \in H^1_\asym$.
We also note that $U^{2n}\Phi \in H^1_\asym$ for any $n \in \N$.
\end{proof}

Finally, for $\Psi \in H^1$ we also have $\varrho_\Psi^{1/2} \in H^1$ by the Hoffmann-Ostenhof inequality \cite{Hof-77},
\begin{equation}\label{eq:HO-ineq}
    \int_{\R^{2N}} |\nabla \Psi|^2 \ge \int_{\R^2} \bigl|\nabla\sqrt{ \varrho_\Psi}\bigr|^2,
\end{equation}
where in general
\begin{equation}\label{eq:def-density}
    \varrho_\Psi(\bx) := 
    \sum_{j=1}^N \int_{\R^{2(N-1)}} |\Psi(\bx_1, \ldots, 
	\bx_{j-1}, \bx, \bx_{j+1}, \ldots, \bx_N)|^2 \prod_{k \neq j}d\bx_k.
\end{equation}

\subsection{Useful estimates}\label{sec:appendix-estimates}

Here we repeat some useful generic bounds, well known in the literature.

\begin{lemma}
\label{lem:productinequality}
Let $0 \leq a_i \leq 1$ for $i=1,\cdot \cdot \cdot,N$ and $N    
 \in \N$. Then,
\begin{align*}
    \prod_{i=1}^N a_i \leq 1 - \sum_{i=1}^N (1-a_i).
\end{align*}
\end{lemma}
\begin{proof}
    By considering $b_i = 1-a_i$ for every $i=1,\cdot \cdot \cdot, N$, it is equivalent to prove that 
    \begin{align}
    \label{eq:Ninequa}
        \prod_{i=1}^N (1-b_i) \leq 1- \sum_{i=1}^N b_i.
    \end{align}
To see this, we use an induction on $N$. For $N=1$, it is trivial. For $N=2$, we need to show that 
\begin{align*}
        \prod_{i=1}^2 (1-b_i) \leq 1- \sum_{i=1}^2 b_i,
    \end{align*}
   which is equivalent with 
   \begin{align}
   \label{eq:twoineq}
      b_1 + b_2 \geq 2 b_1 b_2.
    \end{align}
Since $b_1 + b_2 \geq 2 \sqrt{b_1 b_2}$ and $b_1, b_2 \in [0,1]$, we derive \eqref{eq:twoineq}. Finally, if \eqref{eq:Ninequa} holds for $N$, we obtain
\begin{align*}
        \prod_{i=1}^{N+1} (1-b_i) \leq (1-b_{N+1}) (1- \sum_{i=1}^N b_i) \leq 1- \sum_{i=1}^{N+1} b_i,
    \end{align*}
where we used \eqref{eq:Ninequa} for $N=2$ on the last inequality above.
\end{proof}

\begin{lemma}[Hardy bound on the magnetic curvature term]\label{lem:af_three_body} 
	We have for any $u \in L^2(\R^2)$ and $R\geq 0$, that
	\[
	\int_{\R^2} \left| \bAR[|u|^2] \right|^2 |u|^2 
	\le \frac{3}{2} \|u\|_{2}^4 \int_{\R^2} \left| \nabla |u| \right|^2
	\le \frac{3}{2} \|u\|_{2}^4 \int_{\R^2} \left| \nabla |u| \right|^2.
	\]
\end{lemma}
\begin{proof}
    The proof follows the one of \cite[Lemma A.1]{LunRou-15}, but with $\nabla^{\perp}w_R$ instead of $\nabla^{\perp}w_0$.
	We have
	\begin{align*}
		& \int_{\R^2} \left| \bAR[|u|^2](\bx) \right|^2 |u(\bx)|^2 \,\mathrm{d}\bx 
		= \iiint_{\R^6} \frac{\bx-\by}{|\bx-\by|_R^2} \cdot \frac{\bx-\bz}{|\bx-\bz|_R^2} |u(\bx)|^2 |u(\by)|^2 |u(\bz)|^2 \,\mathrm{d}\bx\mathrm{d}\by\mathrm{d}\bz \\
		&\leq \frac{1}{6} \int_{\R^6} \frac{1}{\rho(X)^2} \left| |u|^{\otimes 3} \right|^2 \mathrm{d}X 
		\le \frac{1}{2} \int_{\R^6} \left| \nabla_X |u|^{\otimes 3} \right|^2 \mathrm{d}X
		= \frac{3}{2} \int_{\R^2} \bigl| \nabla|u(\bx)| \bigr|^2 \mathrm{d}\bx
		\left( \int_{\R^2} |u(\bx)|^2 \mathrm{d}\bx \right)^2.
	\end{align*}
	Here, we used the symmetry and that 
	\begin{equation}\label{ine:HAR}
		0\leq\frac{\bx-\by}{|\bx-\by|_R^2} \cdot \frac{\bx-\bz}{|\bx-\bz|_R^2}+\frac{\by-\bx}{|\by-\bx|_R^2} \cdot \frac{\by-\bz}{|\by-\bz|_R^2}+\frac{\bz-\by}{|\bz-\by|_R^2} \cdot \frac{\bz-\bx}{|\bz-\bx|_R^2}\leq \frac{C}{\rho^2(\bx,\by,\bz)},
	\end{equation}
	as shown in \cite[Proof of Lemma 2.4, inequality (2.13)]{LunRou-15} with 
    $$\rho^2(\bx,\by,\bz):=|\bx-\by|^2+|\by-\bz|^2+|\bz-\bx|^2$$ 
    and where we concluded applying Hardy's inequality \cite[Lemma 2.5]{LunRou-15}.
\end{proof}
\begin{lemma}[Three-body term]\label{lem:3body}
    We have that, as operators on $L^{2}_{\mathrm{sym}}(\mathbb{R}^{6})$,
	\begin{equation}\label{eq:W3C}
		0\leq \nabla^{\perp}w_{R}(\bx-\by)\cdot\nabla^{\perp}w_{R}(\bx-\bz)
            \le -C\Delta_\sx,
	\end{equation}
    for a univeral constant $C>0$.
\end{lemma}
The proof can be found in \cite[Lemma 2.4]{LunRou-15}. It is follows the same lines as the proof of Lemma \ref{lem:af_three_body}.

\subsection{The CSS/afP model}\label{sec:appendix-CSS}

Some fundamental questions concerning stability for the CSS/afP functional \eqref{def:Eaf} have been answered in
\cite{AtaLunNgu-24}, 
and for convenience we gather the results referred to here.

For any $\beta,\gamma \in \R$, $V \in L^1_\loc(\R^2;\R)$,
we define the ground-state energy (per particle)
\begin{equation}\label{eq:AFP-gse}
    E_{\beta,\gamma,V} := 
		\inf \left\{ \cE_{\beta,\gamma,V}[u] : 
		u \in H^1(\R^2;\C), 
		\int_{\R^2} |V| |u|^2 < \infty,
		\int_{\R^2} |u|^2 = 1 \right\}.
\end{equation}

The Euler--Lagrange equation for a ground state $u$ of \eqref{eq:AFP-gse} 
with finite energy is
\begin{equation}\label{eq:CSS-EL}
\Big[-\left(\nabla + i\beta\bA[|u|^2]\right)^2 
- 2\beta \nabla^\perp w_0 * \bigl[\beta\bA[|u|^2]|u|^2 + \bJ[u]\bigr] 
+ 2\gamma|u|^2 + V\Big] u = \lambda u,
\end{equation}
where
	$\bJ[u] := \frac{i}{2}(u\nabla \bar{u} - \bar{u}\nabla u)$
denotes the current of $u$, and $\lambda = \lambda(u) \in \R$ is a constant:
\begin{equation}\label{eq:CSS-EL-lambda}
    \lambda = 2E_{\beta, \gamma, V} - \int_{\R^2} \left[|\nabla u|^2 + V|u|^2 - \beta^2 \left|\bA\left[|u|^2\right]\right|^2|u|^2\right];
\end{equation}
see \cite[Appendix]{CorLunRou-17} and \cite[Remark~3.11]{AtaLunNgu-24}.
For $\beta=0$, \eqref{eq:CSS-EL} reduces to the \emph{local}, cubic nonlinear Schr\"odinger (NLS) equation 
associated to the well-known Gross-Pitaevskii (NLS) functional. For $\gamma > 0$ it is defocusing and for $\gamma < 0$ focusing.

For any $\beta \in \R$, 
we define the corresponding critical coupling
\begin{align}\label{eq:defgamma}
	\gamma_*(\beta) := \inf \left\{ \frac{ \cE_{\beta,0,0}[u]}{\int_{\R^2} |u|^4}
	: u \in H^1(\R^2;\C), \int_{\R^2} |u|^2 = 1 \right\},
\end{align}
and $\CLGN := \gamma_*(0)$ the Ladyzhenskaya--Gagliardo--Nirenberg $H^1 \hookrightarrow L^4$ embedding constant.
We also define the ``nonlinear Landau level'' (NLL) for any $\beta \neq 0$:
\begin{equation} \label{eq:defNLL}
		\NLL(\beta) := \left\{ u \in H^1(\R^2;\C) :
			\cE_{\beta,0,0}[u] = 2\pi|\beta|\int_{\R^2} |u|^4,
			\int_{\R^2} |u|^2 = 1 \right\}.
\end{equation}

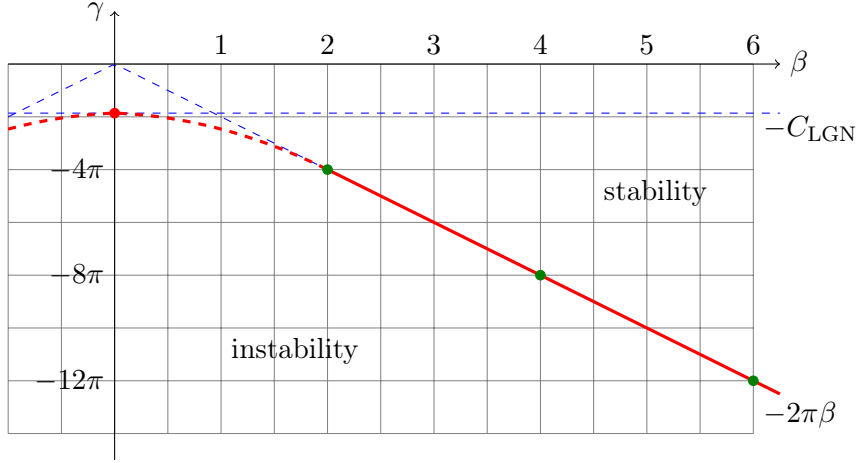
\begin{figure}
    \begin{tikzpicture}[domain=-2:12,scale=0.7]
	    \draw[very thin,color=gray] (-2,-7) grid (12,0);
	    \draw[->] (-2,0) -- (12.5,0) node[right] {$\beta$};
	    \draw[->] (0,-7.5) -- (0,1) node[left] {$\gamma$};
	    \draw[thin,dashed,color=blue,domain=-2:12.5,samples=50] plot (\x,{-0.93});
	    \draw[thin,dashed,color=blue,domain=-2:4,samples=50] plot (\x,{-0.5*abs(\x)});
	    \draw[very thick,dashed,color=red,domain=-2:4,samples=50] plot (\x,{-0.5*abs(\x) - 0.93*(4-abs(\x))^2/16});
	    \draw[very thick,color=red,domain=4:12.5,samples=50] plot (\x,-\x/2);
		\draw [fill,color=red] (0,-0.93) circle [radius=0.09];
		\draw [fill,color=darkgreen] (4,-2) circle [radius=0.09];
		\draw [fill,color=darkgreen] (8,-4) circle [radius=0.09];
		\draw [fill,color=darkgreen] (12,-6) circle [radius=0.09];
	    \node[above] at (2,0) {$1$};
	    \node[above] at (4,0) {$2$};
	    \node[above] at (6,0) {$3$};
	    \node[above] at (8,0) {$4$};
	    \node[above] at (10,0) {$5$};
	    \node[above] at (12,0) {$6$};
	    \node[left] at (0,-2) {$-4\pi$};
	    \node[left] at (0,-4) {$-8\pi$};
	    \node[left] at (0,-6) {$-12\pi$};
	    \node[below right] at (12,-0.8) {$-\CLGN$};
	    \node[below right] at (12,-6.2) {$-2\pi\beta$};
		\node [below right] at (9,-2) {stability};
		\node [below right] at (2,-5) {instability};
    \end{tikzpicture}
\caption{Sketch of $-\gamma_*(\beta)$ with exact $\NLL$s at $\beta \in 2\Z \setminus \{0\}$.}
\label{fig:stability}
\end{figure}

The following two theorems summarize our present knowledge concerning the stability of the problem \eqref{eq:AFP-gse}.

\begin{theorem}[Magnetic stability; {\cite[Theorem~2]{AtaLunNgu-24}}] \label{thm:magneticstability}
The following holds:
\begin{enumerate} 
\item\label{itm:mstab-gamma}
	We have that $\beta \mapsto \gamma_*(\beta)$ is a Lipschitz function and satisfies
	\begin{equation}\label{eq:mag-bermuda}
        \gamma_*(\beta) > \max\{\CLGN, 2\pi\beta\}
        \quad \text{for every $0<\beta <2$,}
	\end{equation}
    and 
	\begin{equation}\label{eq:mag-susy}
        \gamma_*(\beta) = 2\pi\beta 
        \quad \text{for every $\beta \geq 2$.}
	\end{equation}
\item\label{itm:mstab-mini}
	Any minimizer of \eqref{eq:defgamma}, if it exists, is smooth. For small enough $0 < \beta < 2$, there exists a minimizer. For $\beta \geq 2$,
    minimizers exist if and only if $\beta \in 2 \N$, 
	and are of the form 
	\begin{equation}\label{eq:mag-solution}
		u = u_{P,Q} := \sqrt{\frac{2 }{\pi \beta}} \, \frac{\overline{P' Q - P Q'}}{|P|^2 + |Q|^2} ,
	\end{equation}
	where $P,Q$ are two coprime and linearly independent complex polynomials satisfying 
	$$
		\max (\deg(P),\deg(Q)) = \frac{\beta}{2}.
	$$
\item\label{itm:mstab-symm}
	Finally, $u_{P,Q} = u_{\tilde{P},\tilde{Q}}$ for two such pairs of polynomials $(P,Q),(\tilde{P},\tilde{Q})$ if and only if
	$(P,Q)=\Lambda(\tilde{P},\tilde{Q})$ for some constant $\Lambda \in \R^+ \times \sSU(2)$.
\end{enumerate}
\end{theorem}

\begin{theorem}[Stability for the almost-bosonic anyon gas; {\cite[Theorem~3]{AtaLunNgu-24}}]\label{thm:stability}
    Let $\beta \in \R$ and the potential $V$ be smooth and bounded from below. 
    The critical coupling $\gamma_*(\beta)$ defined in \eqref{eq:defgamma}
    is exactly the critical value for stability of 
    the system \eqref{eq:AFP-gse} in the sense that 
    $E_{\beta,\gamma,V} > -\infty$ if $\gamma \ge -\gamma_*(\beta)$
    and $E_{\beta,\gamma,V} = -\infty$ if $\gamma < -\gamma_*(\beta)$.
    If $V=0$ 
    then $E_{\beta,-\gamma_*(\beta),0}=0$ for all $\beta$, 
    and if furthermore
    $\beta\ge 2$ and $-\gamma=\gamma_*(\beta)$, 
    which in this case equals $2\pi\beta$,
    then zero-energy ground states exist if and only if $\beta \in 2\N$ 
    and are then given exactly by the $2\beta$-dimensional soliton manifold
	\begin{multline}\label{eq:NLLthm}
		\NLL(\beta=2n) = \Biggl\{ u = \frac{1}{\sqrt{\pi n}} \, \frac{\overline{P' Q - P Q'}}{|P|^2 + |Q|^2} : \ 
			\text{$P,Q$ coprime and linearly independent}\\
			\text{complex polynomials s.t. $\max\{\deg P,\deg Q\}=n$} 
			\Biggr\},
	\end{multline}
	whereas $\NLL(\beta) = \emptyset$ if $\beta \notin 2\Z$.
	Finally, for any $\beta,\gamma \in \R$,
	\begin{equation}\label{eq:AFP-KLT}
		E_{\beta,\gamma,V} \ge E^{\rm KLT}_{\beta,\gamma,V} 
		:= \inf \left\{
		\int_{\R^2} \left[ (\gamma_*(\beta)+\gamma) \varrho^2 + V\varrho \right]
		: \varrho \in L^2(\R^2;\R_+), \ \int_{\R^2} \varrho = 1
		\right\}.
	\end{equation}
\end{theorem}

Finally, the following compactness result was originally proved for $\gamma = 0$ in \cite[Proposition~3.7]{LunRou-15}.

\begin{proposition}[Existence of minimizers]\label{prop:minimizers}
    Let $\beta \in \R$
    and $\gamma > -\gamma_*(\beta)$, 
    and assume $V \ge 0$ sufficiently trapping in the sense that the Schr\"odinger operator $-\Delta + V$ has compact resolvent.
    Then there exists $u \in H^1(\R^2) \cap L^2_V$ with 
    $\int_{\mathbb{R}^2}\left|u\right|^2=1$ and 
    $\cE_{\beta,\gamma,V}[u]=E_{\beta,\gamma,V}$.
\end{proposition}
\begin{proof}
    Take a minimizing sequence
    $$
	\left(u_n\right)_{n} \subset  H^1(\R^2) \cap L^2_V, \quad 
        \left\|u_n\right\|_{L^2}=1, \quad  
        \lim_{n \rightarrow \infty} \mathcal{E}_{\beta,\gamma,V}\left[u_n\right] = E_{\beta,\gamma,V}.
    $$
    Then, with the given conditions, 
    $$	\mathcal{E}_{\beta,\gamma,V}[u_n] \ge \min\left\{1,1 - \frac{\gamma}{\gamma_*(\beta)} \right\} \int_{\R^2} \left| (\nabla + \im\beta\bA[|u_n|^2]) u_n\right|^2
		\ge 0
    $$
    is uniformly bounded.
    It then follows from Lemma~\ref{lem:af_three_body}
    and the diamagnetic inequality that
    $$
	\int_{\R^2} |\bA[|u_n|^2]|^2 |u_n|^2
		\lesssim \int_{\R^2} \bigl|\nabla |u_n|\bigr|^2 \lesssim \int_{\mathbb{R}^2}\left|\left(\nabla+ \im\beta \mathbf{A}[|u_n|^2]\right) u_n\right|^2
    $$
    is uniformly bounded.
    Therefore, $\{u_n\}$ is uniformly bounded in both $L^2(\mathbb{R}^2)$, $L_V^2$, and $H^1(\mathbb{R}^2)$ (and hence in $L^p(\mathbb{R}^2)$ for any $p \in[2, \infty)$, by interpolation). By the Banach--Alaoglu theorem, there exists $u \in H^1(\R^2) \cap L^2_V$ and a weakly convergent subsequence (still denoted $u_n$) such that
    $$
        u_n \rightarrow u \text{ weakly in } L^2(\mathbb{R}^2) \cap L_V^2 \cap L^p(\mathbb{R}^2), \quad \nabla u_n \rightarrow \nabla u \text{ weakly in } L^2(\mathbb{R}^2) .
    $$
    Moreover, since $(-\Delta+V+1)^{-1 / 2}$ is compact we have that
    $$
        u_n = (-\Delta+V+1)^{-1 / 2}(-\Delta+V+1)^{1 / 2} u_n
    $$
    is actually strongly convergent (again extracting a subsequence). Hence
    \begin{align}
    \label{eq:strongL2convergence}
        u_n \rightarrow u \text { strongly in } L^2(\mathbb{R}^2).
    \end{align}  
    Also, $\mathbf{A}[|u_n|^2]$ converges pointwise a.e. to $\mathbf{A}\left[|u|^2\right]$ by weak convergence of $u_n$ in $L^p$, and
    $$
	\lim_{n\to\infty}\left\|\mathbf{A}[|u_n|^2] u_n\right\|_{L^2}^2 = \left\|\mathbf{A}\left[|u|^2\right] u\right\|_{L^2}^2
    $$
    by dominated convergence. It then follows that
    \begin{align*}
			\left\|\left(\nabla + \im\beta \mathbf{A}\left[|u|^2\right]\right) u\right\|_{L^2} & =\sup _{\|v\|=1}\left|\left\langle\nabla u + \im\beta \mathbf{A}\left[|u|^2\right] u, v\right\rangle\right| \\
			& =\sup _{\|v\|=1} \lim _{n \rightarrow \infty}\left|\left\langle\nabla u_n + \im\beta \mathbf{A}[|u_n|^2] u_n, v\right\rangle\right| \\
			& \leq \liminf _{n \rightarrow \infty} \sup _{\|v\|=1}\left|\left\langle\nabla u_n + \im\beta \mathbf{A}[|u_n|^2] u_n, v\right\rangle\right| \\
			& =\liminf _{n \rightarrow \infty}\left\|\left(\nabla + \im\beta \mathbf{A}[|u_n|^2]\right) u_n\right\|_{L^2}.
    \end{align*}
    Since $\|\cdot\|_{L_V^2}$ is weakly lower semicontinuous, 
    and $\int_{\R^2} |u_n|^4 \to \int_{\R^2} |u|^4$, by the Ladyzhenskaya--Gagliardo--Nirenberg interpolation inequality on $\R^2$ and \eqref{eq:strongL2convergence}, we have 
    $$
		E_{\beta,\gamma,V} = \liminf _{n \rightarrow \infty} \mathcal{E}_{\beta,\gamma,V}\bigl[u_n\bigr] \geq \mathcal{E}_{\beta,\gamma,V}[u] \geq E_{\beta,\gamma,V},
    $$
    which shows that $u$ is indeed the desired minimizer.
\end{proof}

%
%
\addcontentsline{toc}{section}{References}
\def\MR#1{} 
\newcommand{\etalchar}[1]{$^{#1}$}


\end{document}